\documentclass[11pt, reqno]{article}

\pdfoutput=1

\usepackage{graphicx}
\usepackage{jheppub}
\usepackage{amssymb}
\usepackage{amsmath}
\usepackage[usenames,dvipsnames]{xcolor}
\usepackage{epsfig}
\usepackage{dcolumn}
\usepackage{tikz}
\usetikzlibrary{shapes.geometric, arrows}
\usepackage{upgreek}
\usepackage{setspace}
\usepackage{enumitem}
\usepackage{array,multirow,bigdelim,arydshln}
\usepackage{appendix}
\usepackage{xparse}
\usepackage[utf8]{inputenc}
\usepackage{hyperref}
\hypersetup{
	colorlinks,
	urlcolor=Maroon,
	linkcolor=Maroon,
	citecolor=Maroon
}

\usepackage{amsthm}
\usepackage{amsmath}
\newtheorem{thm}{Theorem}[section]
\newtheorem{prop}[thm]{Proposition}

\newtheorem{defn}[thm]{Definition}
\theoremstyle{definition}

\newtheorem{conjecture}[thm]{Conjecture}
\theoremstyle{remark}

\usepackage{mathtools}
\usepackage{scalerel}

\usepackage{float}
\restylefloat{table}

\NewDocumentCommand{\binomial}{omm}
{%
	\genfrac(){0pt}{}{#2}{#3}%
	\IfValueT{#1}{_{\!#1}}%
}
\NewDocumentCommand{\eulerian}{omm}
{%
	\genfrac<>{0pt}{}{#2}{#3}%
	\IfValueT{#1}{_{\!#1}}%
}

\def \s {\sigma}

\usepackage{underscore}

\newenvironment{claim}[1]{\par\noindent\underline{Claim:}\space#1}{}

\usepackage{latexsym}
\usepackage{tikz}

\usepackage[percent]{overpic}

\usepackage{multirow} 
\usepackage{slashed}

\usepackage[]{shuffle}

\def\yz#1\yz {{\color{blue} [[YZ: #1]] }}

\def\yzg#1\yzg {{\color{gray} [[YZ: #1]] }}

\def\yzz#1\yzz {{\color{gray} [[To be verified: #1]] }}

\def\ne#1\ne {{\color{green} [[NE: #1]] }}

\usepackage{longtable}

\usepackage{cleveref}

\usepackage{diagbox}

\usepackage{blkarray}
\newcommand{\matindex}[1]{\mbox{\scriptsize#1}}

\title{Generalized Color Orderings: CEGM Integrands and Decoupling Identities}

\author[a]{Freddy Cachazo,}
\author[b]{Nick Early,}
\author[a]{and Yong Zhang}
\affiliation[a]{Perimeter Institute for Theoretical Physics, Waterloo, ON N2L 2Y5, Canada.}
\affiliation[b]{Max Planck Institute for Mathematics in the Sciences, Leipzig, Germany.}
\emailAdd{fcachazo@pitp.ca}
\emailAdd{nick.early@mis.mpg.de}
\emailAdd{yzhang@pitp.ca}

\abstract{
In a recent paper, we defined generalized color orderings (GCO) and Feynman diagrams (GFD) to compute color-dressed generalized biadjoint amplitudes. In this work, we study the Cachazo-Early-Guevara-Mizera (CEGM) representation of generalized partial amplitudes and ``decoupling" identities. This representation is a generalization of the Cachazo-He-Yuan (CHY) formulation as an integral over the configuration space $X(k,n)$ of $n$ points on $\mathbb{CP}^{k-1}$ in generic position. 

Unlike the $k=2$ case, Parke-Taylor-like integrands are not enough to compute all partial amplitudes for $k>2$. Here we give a set of constraints that integrands associated with GCOs must satisfy and use them to construct all $(3,n<9)$ integrands, all $(3,9)$ integrands up to four undetermined constants, and $95 \%$ of $(4,8)$ integrands up to 24 undetermined constants. 

$k=2$ partial amplitudes are known to satisfy identities. Among them, the so-called $U(1)$ decoupling identities are the simplest ones. These are characterized by a label $i$ and a color ordering in $X(2,|[n]\setminus \{i\}|)$. Here we introduce decoupling identities for $k>2$ determined combinatorially using GCOs. Moreover, we identify the natural analog of $U(1)$ identities as those characterized by a pair of labels $i\neq j$, and a pair of GCOs, one in $X(k,|[n]\setminus \{i\}|)$ and the other in $X(k-1,|[n]\setminus \{j\}|)$. We call them {\it double extension} identities. 

We also provide explicit connections among different ways of representing GCOs, such as configurations of lines, configurations of points, and reorientation classes of uniform oriented matroids (chirotopes).

}

\begin{document}
	\maketitle
	\addtocontents{toc}{\protect\setcounter{tocdepth}{1}}
	\def \tr {\nonumber\\}
	\def \nn {\nonumber}
	\def \la {|}
	\def \ra {|}
	\def \lan {\langle}
	\def \ran {\rangle}
	\def \dd {\Theta}
	\def\hset{\texttt{h}}
	\def\gset{\texttt{g}}
	\def\sset{\texttt{s}}
	\def \be {\begin{equation}}
		\def \ee {\end{equation}}
	\def \ba {\begin{eqnarray}}
		\def \ea {\end{eqnarray}}
	\def \bg {\begin{gather}}
		\def \eeg {\end{gather}}
	\def \k {\kappa}
	\def \h {\hbar}
	\def \r {\rho}
	\def \l {\lambda}
	\def \be {\begin{equation}}
		\def \en {\end{equation}}
	\def \bes {\begin{eqnarray}}
		\def \ens {\end{eqnarray}}
	\def \red {\color{Maroon}}
	\def \pt {{\rm PT}}
	\def \s {\mathfrak{s}}
	\def \t {\mathfrak{t}}
	\def \v {\mathfrak{v}}
	\def \C {\textsf{C}}
	\def \tp {||}
	\def \p {x}
	\def \x {z}
	\def \V {\textsf{V}}
	\def \ls {{\rm LS}}
	\def \ma {\Upsilon}
	\def \SL {{\rm SL}}
	\def \GL {{\rm GL}}
	\def \w {\omega}
	\def \e {\epsilon}
	
	\numberwithin{equation}{section}


\section{Introduction}

Tree-level scattering amplitudes in the biadjoint cubic scalar theory with group $SU(N)\times SU(\tilde N)$ can be color-decomposed as (see \cite{Mangano:1990by} for a review)
\begin{equation}\label{coBS}
    {\mathcal M}(\{ k_i,a_i,{\tilde a}_i\} ) = \sum_{\alpha,\beta \in S_{n}/{\mathbb{Z}_n}} {\text {tr}}\left( T^{a_{\alpha (1)}}T^{a_{\alpha (2)}} \cdots T^{a_{\alpha (n)}}\right){\text {tr}}\left( T^{{\tilde a}_{\beta (1)}}T^{{\tilde a}_{\beta (2)}} \cdots T^{{\tilde a}_{\beta (n)}}\right)m_n(\alpha,\beta). 
\end{equation}
Here the $T^{\alpha_i},T^{{\tilde a}_j}$ are generators of $SU(N)$ and $SU(\tilde{N})$, respectively, $\alpha$ and $\beta$ are called {\it color orderings}, $m_n(\alpha,\beta)$ are {\it partial amplitudes} while ${\mathcal M}(\{ k_i,a_i,{\tilde a}_i\} ) $ is the {\it color-dressed amplitude}. 

In \cite{Cachazo:2022pnx}, color-dressed generalized amplitudes were introduced. This was done for any $(k,n)$ by defining generalized color orderings (GCO) using arrangements of $(k{-}2)$-planes (see e.g. \cite{bjorner1999oriented}) in $\mathbb{RP}^{k-1}$ and generalized Feynman diagrams (GFD) using arrangements of metric trees \cite{{herrmann2008draw,Borges:2019csl,Cachazo:2019xjx}}. The result is expressed as  
\begin{equation}\label{coBSintro}
    {\mathcal M}^{(k)}_n = \sum_{I,J} {\bf c}(\Sigma_I){\bf c}(\Sigma_J)\, m_n^{(k)}(\Sigma_I,\Sigma_J), 
\end{equation}
where the sum is over all $(k,n)$ generalized color orderings, $\Sigma_I$ denotes a generalized color ordering, while ${\bf c}(\Sigma_I)$ is the corresponding color factor. Here as in \cite{Cachazo:2022pnx}, ${\bf c}(\Sigma_I)$ are treated as formal variables.

The $(2,n)$ biadjoint amplitudes, $m_n(\alpha,\beta)$, admit a Cachazo-He-Yuan (CHY) formulation as an integral over $X(2,n)$, the configuration space of $n$ points in $\mathbb{CP}^1$, with integrand given by the so-called Parke-Taylor factors \cite{Cachazo:2013hca,Cachazo:2013iea}. For later convenience, we denote the set of all $(k,n)$ GCOs as ${CO_{k,n}}$.

In 2019, a generalization of the CHY formalism to $X(k,n)$, the configuration space of $n$ points in $\mathbb{CP}^{k-1}$, was introduced by the first two authors, Guevara and Mizera (CEGM) \cite{CEGM} (see also \cite{Drummond:2019qjk,Arkani-Hamed:2019mrd}),    which has been studied extensively in the literature ({\it c.f.} \cite{GarciaSepulveda:2019jxn,Cachazo:2019ble,Drummond:2020kqg,He:2020ray,Abhishek:2020sdr}). Unlike the $k=2$ case, not all $(k,n)$ partial amplitudes, $m_n^{(k)}(\Sigma_I,\Sigma_J)$, can be computed using Parke-Taylor factors as integrands in the CEGM formulation \cite{Cachazo:2022pnx}.    

Only when both $\Sigma_I$ and $\Sigma_J$ belong to the special type of $(k,n)$ color orderings called type 0, the partial amplitude is computed using higher $k$ Parke-Taylor (PT) factors, where type 0 corresponds to GCOs obtained from $(2,n)$ color orderings by ``pruning'' several labels at a time (i.e. they are {\it descendants} of a $(2,n)$ color orderings).

In this work, we present a series of constraints that CEGM integrands, associated with $(k,n)$ generalized color ordering, must satisfy and which we implement in an algorithm in section \ref{sec2}. In a nutshell, a $(k,n)$ color ordering, $\Sigma$, is a collection of ${n \choose k-2}$ orders of type $(2,n-k+2)$, satisfying a certain compatibility conditions. The integrand associated with it, $I_{\Sigma}$, is a rational function of the Plucker coordinates $\Delta_{i_1i_2\cdots i_k}$ associated with $X(k,n)$, with only simple poles determined by the simplices in the arrangement of hyperplanes defined by $\Sigma$. The numerator is constrained by the torus weight and by the condition that under a simplex flip (see  \cref{sec3}) that connects $\Sigma$ to $\Sigma'$, the residues of $I_\Sigma$ and $I_{\Sigma'}$ coincide. This simple observation determines all integrands for $k=3$ and $n<9$. Modulo relabeling, there are $4$ for $n=6$, $11$ for $n=7$, and $135$ for $n=8$. For $n=9$ we find that all $4381$ types are fixed up to four constants, which are independent of the coordinates on $X(3,9)$. We also determine $2469$ of the $2604$ types of $(4,8)$ integrands up to 24 undetermined constants.

Using the integrands we found, we have verified all $(3,6), (3,7), (4,7)$ and $(3,8)$ color-ordered amplitudes obtained from CEGM integrals agree with those obtained using GFDs explained in \cite{Cachazo:2022pnx}, which is a strong consistency check both for the integrands and the GFDs.

In the second part of this work, we study the generalization of decoupling identities. In $(2,n)$ amplitudes, if one of the two generators assigned to, say particle $n$, commutes with those of the other particles, then the color-dressed amplitude can be shown to vanish. This implies identities among the partial amplitudes $m_n(\alpha,\beta)$ known as $U(1)$ decoupling identities\footnote{The name $U(1)$ comes from the common practice of assigning $T=\mathbb{I}$ to the particle being decoupled. If the group under consideration was $U(N)$, then $T=\mathbb{I}$ would be the generator of the $U(1)$ subgroup.}. 

The way partial amplitudes organize is by the particular $(2,n-1)$ color ordering a given $(2,n)$ ordering descends to after projecting out label $n$. 

The CHY formulation gives a simple geometric interpretation of this identity in terms of a fibration of $X(2,n)$ over $X(2,n-1)$. Consider the real model of $X(2,n-1)$ as the boundary of a disk with $n-1$ points on it. The points split the boundary into $n-1$ segments. Placing point $n$ on a given segment corresponds to a given color ordering and hence to a Parke-Taylor factor. The poles in the Parke-Taylor factor that depend on $n$ correspond to the points defining the segment. An identity is obtained by adding all Parke-Taylor factors corresponding to all color orderings obtained while point $n$ traverses the circle. 

The $(3,n)$ analog of this procedure reveals a rich structure which we fully explore for $n \leq 9$.

We identify the analog of the $U(1)$ decoupling identities, which we call ``fundamental" identities, and make a proposal for all $k$ and $n$ (see section \ref{ref4B}). Fundamental identities are characterized by the choice of two labels $i\neq j$ and two generalized color orderings, one is a $(k,n-1)$ GCO, $\Sigma_1$, on the set $[n]\setminus \{ i\}$ and the other is a $(k-1,n-1)$ GCO, $\Sigma_2$, on the set $[n]\setminus \{ j\}$. In other words, one constructs a, possibly empty, set of $(k,n)$ GCOs that ``project" to $\Sigma_1$ under a $k$-preserving projection and to $\Sigma_2$ under a $k$-decreasing projection. If the set is not empty then the integrands corresponding to the GCOs in it satisfy a fundamental identity.  

We also discuss how the duality between $(k,n)$ and $(n-k,n)$ reveals even more identities and apply it to the case $k=n-2$ for which we identify all identities as shuffle identities in the dual $(2,n)$ system.

The rest of this paper is organized as follows: A brief review of the CEGM formulation is given in \cref{sec2}. Then in \cref{sec3}, we show how to construct integrands that geometrize the combinatorial properties of GCOs, and in \cref{sec4}, we discuss how to derive higher $k$ irreducible decoupling identities in a purely combinatorial way
using GCOs. In \cref{sec:translation}, we point out ways to relate configurations of points to GCOs and to chirotopes.  In \cref{ref4B}, we introduce double extension identities and argue that they are fundamental and are the natural analog of $U(1)$ identities. We end with a discussion of future directions in \cref{sec5}. Most data is presented either in the appendices or in an ancillary file.

\section{Review of CEGM Formulation\label{sec2}}

In this section, we review the $(k,n)$ CEGM formalism \cite{CEGM} which is a generalization of the CHY formalism obtained for $k=2$. The construction is based on the configuration space of $n$ points on $\mathbb{CP}^{k-1}$, usually denoted by $X(k,n)$, and defined as
\be                
X(k,n) : = {\rm SL}(k,\mathbb{C})\backslash M^*(k,n) / {\rm T}^n , 
\ee 
where $M^*(k,n)$ is the space of $k\times n$ matrices with non-vanishing maximal minors and ${\rm T}^n = \left(\mathbb{C}^*\right)^{n}$ is the algebraic torus that acts on each column by rescalings. 

The key ingredient in the CEGM formulation is the scattering equations. These are the equations which determine the critical points of 
\be\label{cegmPot} 
	{\cal S}^{(k)}_n := \!\!\!\!\sum_{a_1<a_2<\ldots <a_k}\!\!\!\! \s_{a_1a_2\ldots a_k}\log |\Delta_{a_1a_2\ldots a_k}| \quad {\rm with} \quad \Delta_{a_1a_2\ldots a_k} := {\rm det}\left[ \begin{array}{cccc}
		M_{1a_1} & M_{1a_2} & \cdots & M_{1a_{k}}\\
		M_{2a_1} & M_{2a_2} & \cdots & M_{2a_{k}}\\
		\vdots   & \vdots   &  \ddots & \vdots \\      
		M_{k a_1} & M_{k a_2} & \cdots &  M_{k a_k}
	\end{array}  \right] .
\ee 

Here $\s_{a_1a_2\ldots a_k}$ are the generalized Mandelstam invariants and partial amplitudes are rational functions of them. They are assumed to be entries of a generic completely symmetric rank $k$ tensor subject only to a ``masslessness" condition, $\s_{a_1a_2\ldots a_{k-2},b,b}=0$ and ``momentum conservation", $\sum_{a_1,a_2,\ldots ,a_{k-1}=1}^n\s_{a_1a_2\ldots a_{k-1},b}=0$ for all $b$.  

A CEGM integral has the form
\be\label{fullInt}   
\int_{X(k,n)} d\mu_{k,n} \, {\cal I}_L {\cal I}_R ,
\ee  
with $d\mu_{k,n}$ a measure that localizes the integral to the solutions to the scattering equations. The simplest way to write down the measure is by choosing a chart of $X(k,n)$. We postpone its definition to \eqref{defMeasure} and instead discuss its transformation properties under the torus action, 
\be
M_{ia}\to t_i M_{ia}, \quad d\mu_{k,n} \to \left(\prod_{i=1}^n t_i\right)^{2k}d\mu_{k,n}.
\ee 
In order to have a well-defined integral it must be that the full integrand in \eqref{fullInt}, ${\cal I}_L {\cal I}_R$, transforms opposite to $d\mu_{k,n}$. It is standard to assign the same weight to ${\cal I}_L$ and to ${\cal I}_R$ and so 
\be\label{torusS} 
M_{ia}\to t_i M_{ia}, \quad {\cal I}_L\to \left(\prod_{i=1}^n t_i\right)^{-k}{\cal I}_L.
\ee 
This property is one of the keys to finding integrands in the next section.

Let us introduce an explicit parameterization of $X(k,n)$ given by 
\be\label{cegmPara} 
	M := \left[ \begin{array}{ccccccccc}
		1  &    0   & \cdots & 0      & 1 & y_{1,k+1}  & y_{1,k+2} & \cdots & y_{1,n}  \\
		0  &    1   & \cdots & 0      & 1 &  y_{2,k+1}  & y_{2,k+2} & \cdots & y_{2,n} \\
	\vdots & \vdots & \ddots & \vdots & \vdots & \vdots & \vdots & \ddots & \vdots \\
		0  &    0   & \cdots & 1      & 1 &    1  & 1   & \cdots & 1 
\end{array}  \right] .
\ee 
Here we have used $SL(k,\mathbb{C})$ and the torus action to fix the form of the matrix representative. Now it is easy to write down the measure
\be\label{defMeasure} 
d\mu_{k,n} := \prod_{i=1}^{k-1}\prod_{a=k+1}^{n}dy_{ia}\,\delta
\left( \frac{\partial 	{\cal S}^{(k)}_n}{\partial y_{ia}}\right). 
\ee 
Some advanced techniques to solve the scattering equations are developed in \cite{Agostini:2021rze,Sturmfels:2020mpv,Cachazo:2019ble} and the CEGM formula becomes 
\be   
\int_{X(k,n)} d\mu_{k,n} {\cal I}_L {\cal I}_R = \sum_{\rm solutions} \left({\rm det}\left[ \frac{\partial^2 {\cal S}^{(k)}_n}{\partial y_{ia}\partial y_{jb}}\right]\right)^{-1} \, {\cal I}_L {\cal I}_R,
\ee  
where the sum is over the solutions to the equations $\partial {\cal S}^{(k)}_n / \partial y_{ia} = 0$ and the integrand is evaluated on the solutions. 

Our proposal is that there exists an integrand ${\cal I}(\Sigma )$ associated with each $(k,n)$ color ordering $\Sigma$ such that 
\be\label{cegm3n}
m^{(k)}_n(\Sigma, \tilde\Sigma) = \sum_{\rm solutions} \left({\rm det}\left[ \frac{\partial^2 {\cal S}^{(k)}_n}{\partial y_{ia}\partial y_{jb}}\right]\right)^{-1} \, {\cal I}(\Sigma ){\cal I}(\tilde\Sigma )
\ee
agrees with the evaluation of $m^{(k)}_n(\Sigma, \tilde\Sigma)$ as a sum over generalized Feynman diagrams (up to an overall sign). Moreover, in analogy with the CHY construction, once the integrands in the CEGM formula \eqref{cegm3n} are fixed, then $m^{(k)}_n(\Sigma, \tilde\Sigma)$ is also fixed and the sign needed to make the sum over generalized Feynman diagrams agree with it can be uniquely determined for all partial amplitudes. We have  verified this proposal  for $(3,6),(3,7), (4,7) $ and $(3,8)$ using the integrands constructed in \cref{sec3}.


In the original CEGM paper \cite{CEGM}, most of the integrands studied were combinations of Parke-Taylor factors, which are
defined as
\be\label{PTform}
	{\rm PT}^{(k)}(\alpha ) := \frac{1}{\Delta_{\alpha_1,\alpha_2,\ldots ,\alpha_k}\Delta_{\alpha_2,\alpha_3,\ldots ,\alpha_{k+1}}\cdots \Delta_{\alpha_n,\alpha_{1},\cdots ,\alpha_{k-1}}}.
\ee
These integrands are associated with a special type of $(k,n)$ color orderings known as descendants of a $(2,n)$ ordering. Given $\alpha$, a $(2,n)$ color ordering, one constructs a $(k,n)$ ordering by computing the array of dimension $n^{k-2}$ with component $l_1,l_2,\ldots ,l_{k-2}$ obtained by deleting those labels from $\alpha$ to get a $(2,n-k+2)$ color ordering. These are called type 0 GCOs and they are in bijection with $(2,n)$ orderings. Therefore, there are $(n-1)!/2$ of them. 

The rest of this work is devoted to developing techniques for computing the integrand associated with generalized color orderings of other types and to the study of the linear relations they satisfy.

\section{CEGM Integrands from GCOs\label{sec3}}

This section is devoted to the computation of CEGM integrands associated with generalized color orderings. We present constraints that integrands must satisfy and use them to develop an algorithm that aids in their computation. For $(3,n<9)$ GCOs, the algorithm completely determines all integrands.   

In a nutshell, the connection between generalized color orderings and CEGM integrands is found by identifying the poles in the integrands with the $(k-1)$-simplices in the GCOs and by imposing that whenever two GCOs are connected via a simplex-flip, the residues of the corresponding integrands must agree up a sign.

\subsection{Review of Generalized Color Orderings (GCOs) \label{sec3d1}}

Let us start by recalling the appropriate definitions given in \cite{Cachazo:2022pnx} regarding generalized color ordering as well as a short review of their most important properties.   

\begin{defn}\label{GCOkkkk}
For $k\ge 3$, a $(k,n)$ generalized {\it color ordering} is an ${n\choose k-2}$-tuple 
\be
\label{defineGCO}
\Sigma^{[k]} =  \{ \sigma^{(i_1,i_2,\cdots,i_{k-2})}| \{i_1,\cdots, i_{k-2}\} \subset [n] \}\,,
\ee   
where $\sigma^{(i_1,i_2,\cdots,i_{k-2})}$ is a $(2,n-k+2)$ color ordering constructed as follows. 

Let 
$\{ H_{1},H_{2},\ldots ,H_{n} \}$
be an arrangement of $n$ projective $({k-}2)$-planes in generic position in $\mathbb{RP}^{(k-1)}$. Intersecting any $(k{-}2)$ such $H$'s, $\{ H_{i_1},H_{i_2},\ldots ,H_{i_{k-2}} \}$, produces a line, $L^{(i_1,i_2,\cdots,i_{k-2})}$. The line so defined intersects the remaining $(n{-}k{+}2)$ $H$'s each on a point, resulting in a sequence of points on the line which defines a $(2,n{-}k{+}2)$ 
color ordering $\sigma^{(i_1,i_2,\cdots, i_{k-2})}$.
\end{defn}

In the definition we chose to construct $\Sigma^{[k]}$ out of $(2,n{-}k{+}2)$ color orderings. However, it is sometimes convenient to 
note that since each $H_i$ is an $\mathbb{RP}^{k-2}\subset \mathbb{RP}^{k-1}$, then $H_i\cap H_j$ is an $({k-}3)$-plane for all $j\in [n]\setminus \{i\}$. This means that we have an arrangement of $n-1$ $({k-}3)$-planes in $\mathbb{RP}^{k-1}$, i.e. a $(k{-}1,n-1)$ color ordering,
\be
\label{component}
\Sigma^{(i),[k-1]} =  \{ \sigma^{(i,i_2,\cdots,i_{k-2})}| \{i_2,\cdots, i_{k-2}\} \subset [n] \setminus \{i\} \} . 
\ee 
Clearly, we have\footnote{The union implies that duplicates are not included.}
\be
\Sigma^{[k]}=\bigcup_{i=1}^n
\Sigma^{(i),[k-1]}\,.
\ee 

By definition,  removing a $(k{-}2)$-plane, say $H_{i}$, from the arrangement with $n>k+2$ must result in another  arrangement but with $(n-1)$ $(k{-}2)$-planes. Therefore, the operation must give a $(k,n-1)$ color ordering, on the set $[n]\setminus \{i\}$, i.e. label $i$ does not participate. This operation is called a ($k$-preserving) projection \cite{Cachazo:2022pnx},
\be
\label{projectionkkkk}
\pi_i\left(\Sigma ^{[k]}\right): = \left\{ \pi_i\left( 
\sigma^{(i_1,i_2,\cdots,i_{k-2})}
\right)
| \{i_1,\cdots, i_{k-2}\} \subset [n] \setminus \{i\} 
\right\}\,.
\ee 
The projection $\pi_i$ acting on a $k=2$ ordering $\sigma^{(i_1,i_2,\cdots,i_{k-2})}$ just means to remove the label $i$ from the set regardless of its position.

The projection operation is the key to finding the dual of a GCO \cite{Cachazo:2022pnx} under the $X(k,n)\simeq X(n-k,n)$ duality. More important for us here is that it is also the key to defining pseudo-GCOs recursively. The reason the ``pseudo" qualifier has to be added is that it is known that the recursive procedure can give rise to arrangements that are not realizable\footnote{Realizability here means that it is possible to find an arrangement of hyperplanes with the prescribed properties. When this fails, bending the hyperplanes makes the arrangement possible. Such curved hyperplanes are called pseudo-hyperplanes.} (see \cite{bjorner1999oriented} for details).

\begin{defn}\label{generalpseudoGCO}
  A ${n \choose k-2}$-tuple of standard color orderings with $n>k+2$ is said to be a $(k,n)$ pseudo-GCO if all its projections are $(k,n-1)$ pseudo-GCOs, while $(k,k+2)$ pseudo-GCOs are all descendants of $(2,k+2)$ color orderings.
\end{defn}

A GCO must be a pseudo-GCO and a pseudo-GCO which is not a GCO is called a non-realizable pseudo-GCO.

Let us now generalize the notion of a triangle flip introduced for $(3,n)$ color orderings in \cite{Cachazo:2022pnx} to all $k$.

{\it A (combinatorial) $(k-1)$-simplex} of $\Sigma^{[k]}$ is any subset of $k$ elements $\{i_1,i_2,\ldots ,i_k \}\subset [n]$ such that for any partition of it of the form $\{ a,b \} \cup \{ c_1,c_2,\ldots ,c_{k-2} \} = \{i_1,i_2,\ldots ,i_k \}$ the labels $a$ and $b$ are consecutive in the $(2,n-k+2)$ ordering $\sigma^{(c_1,c_2,\ldots ,c_{k-2})}$.

{\it A (combinatorial) simplex flip} of a GCO $\Sigma^{[k]}$ is an operation that produces a new pseudo-GCO (potentially a GCO) as follows. Given a simplex $\{i_1,i_2,\ldots ,i_k \}$ of $\Sigma^{[k]}$ as above, send $\sigma^{(c_1,c_2,\ldots ,c_{k-2})}\to \sigma^{(c_1,c_2,\ldots ,c_{k-2})}|_{a \leftrightarrow b}$ if $\{ a,b \} \cup\{ c_1,c_2,\ldots ,c_{k-2} \} = \{i_1,i_2,\ldots ,i_k \}$ and leave all other $k=2$ orderings invariant.

\subsection{Constructing Integrands \label{sec3d2}}

Now we are ready to start building the connection between GCOs and CEGM integrands. 

Let $\Sigma^{[k]}$ be a $(k,n)$ color ordering and ${\cal I}(\Sigma^{[k]})$ its CEGM integrand, then we claim that all poles of ${\cal I}(\Sigma^{[k]})$ are simple and of the form $1/\Delta_{i_1,i_2,\ldots ,i_{k}}$. Moreover, if a simplex flip of $\Sigma^{[k]}$ along $\{j_1,j_2,\ldots ,j_{k}\}$ leads to another GCO then $1/\Delta_{j_1,j_2,\ldots ,j_{k}}$ is a pole of ${\cal I}(\Sigma^{[k]})$.  

Having defined the poles, the final step is to obtain information regarding the residues of ${\cal I}(\Sigma^{[k]})$ to find its numerator. 

A preliminary step is to note that the number of poles of ${\cal I}(\Sigma^{[k]})$ is $n+p$, with $p\geq 0$, and therefore by using the torus transformation of ${\cal I}(\Sigma^{[k]})$ given in \eqref{torusS} with all $t_i=t$, one has that the numerator must be a polynomial of degree $p$ in the Plucker variables. This means that whenever the number of poles in ${\cal I}(\Sigma^{[k]})$ is $n$ or $n+1$,  we can completely determine it using the known torus transformation of ${\cal I}(\Sigma^{[k]})$. When $p=0$ the numerator is $1$ and when $p=1$ the numerator must be a single Plucker variable, $\Delta_{a_1,a_2,\ldots ,a_k}$, with $\{a_1,a_2,\ldots ,a_k \}$ the labels that participate in $k+1$ Plucker variables in the denominator. 

Let us call the GCOs integrands with $p=0,1$ {\it basic integrands}. All GCOs of type 0 and type I give rise to basic integrands.

The algorithm for computing ${\cal I}(\Sigma^{[k]})$ given a $GCO$ is based on the observation that if $\Sigma^{[k]}$ and $\tilde\Sigma^{[k]}$ are related by a simplex flip, then they share a pole and their residues at the pole must agree (up to a sign). 

\vspace{2mm}
\noindent \paragraph{Algorithm I: Building a System of Integrands}  ~

{\it Input}: All $(k,n)$ GCOs.

{\it Output}: Rational functions associated with each GCO with poles in correspondence with simplices of the GCOs and matching residues whenever there exists a simplex flip.

\begin{enumerate}
    \item Construct a list $L$ with all basic integrands with their corresponding GCOs. 
    \item Consider a GCO in $L$ and perform all possible simplex flips. Select one of the GCO generated in the process for which the integrand has the smallest $p$ and is not already in $L$. Call it $\Sigma_{\rm test}$.
    \item Construct an ansatz for ${\cal I}(\Sigma_{\rm test}) = P_{\rm test}(\Delta)/Q_{\rm test}(\Delta)$ where $Q_{\rm test}(\Delta)$ is a product of all Plucker variables associated with simplices of $\Sigma_{\rm test}$ and $P_{\rm test}(\Delta)$ is a general polynomial of degree $d$ in Plucker variables with unknown coefficients in which each monomial is has a transformation under the torus action determined by \eqref{torusS}.  Of course, we just need to keep such monomials in the polynomial $P_{\rm test}(\Delta)$ that are linearly independent\footnote{There are various ways to construct bases of the graded components of the Plucker algebra, see for example \cite{fulton1996}.  A standard prescription is to take all semi-standard Young tableaux for a given multi-degree with columns $J_1,\ldots, J_d$, say; the basis elements are then products of Pluecker coordinates $\Delta_{J_1}\cdots \Delta_{J_d}$.}.

    \item 
    Denote a GCO generated by any simplex flip as $\Sigma_{\rm other}$.
    Even though the explicit form of $P_{\rm other}(\Delta)$ might not be known, we can already utilize the requirements of matching the residues of ${\cal I}(\Sigma_{\rm test})$ and ${\cal I}(\Sigma_{\rm old})$ at the corresponding pole,  say $1/\Delta_{\rm flip}$,
    as follows. Require 
   \be
   \label{firstutilization}
 \left. P_{\rm test}(\Delta)\right|_{\Delta_{\rm flip},\Delta'=0} =0 \quad \forall ~ \Delta' ~s.t. ~
  \left( \left. {\hat Q}_{\rm test}(\Delta)\right|_{\Delta_{\rm flip},\Delta'=0}\!\neq 0 ~ \wedge ~
 \left. {\hat Q}_{\rm other}(\Delta)\right|_{\Delta_{\rm flip},\Delta'=0} \!\! = 0 \right),
\ee
where $\hat Q:=Q/\Delta_{\rm flip}$, i.e. the polynomial $Q$ with $\Delta_{\rm flip}$ removed. Note that the above equations are sign agnostic, so we can solve them (for all choices of $\Delta_{\rm flip}$ and $\Delta'$) simultaneously and simplify the ansatz $ P_{\rm test}(\Delta)$.

\item 
By construction, some GCOs generated by  simplex flips are already in $L$, say $\Sigma_{\rm old}$, and we can get more constraints on the ansatz $P_{\rm test}(\Delta)$ than  \eqref{firstutilization} by actually matching the residues,
 \be 
 \label{secnodutilization}
 \left. P_{\rm test}(\Delta){\hat Q}_{\rm old}(\Delta)\right|_{\Delta_{\rm flip}=0} = \pm  \left. P_{\rm old}(\Delta){\hat Q}_{\rm test}(\Delta)\right|_{\Delta_{\rm flip}=0}\,.
 \ee
In practice, step 4 provides a simpler ansatz that makes the matching in this step easier.

    \item If using all constraints completely determines $P_{\rm test}(\Delta)$ up to an overall sign then add $\Sigma_{\rm test}$ to the list $L$. Otherwise, store it in another list and reconsider it once the list $L$ has grown further and includes more GCOs that can impose constraints on $\Sigma_{\rm test}$.
    \item Repeat steps 2-6 until all integrands have been found.
\end{enumerate}

We have implemented this algorithm to find all integrands for $(3,6)$, $(3,7)$ , $(3,8)$ as well as $(4,7)$.

However, Step 6 is too restrictive for  $(3,9)$,  $(4,8)$ and beyond. Instead one has to use the following. 

Step 6': Using all available constraints, determine as many unknown coefficients in $P_{\rm test}(\Delta)$ as possible in terms of the rest. Add $\Sigma_{\rm test}$ to the list $L$ and treat it as if it were known.

Note that by adding a partially known integrand, say $I_{\rm test'}(\Delta)$, to the list $L$, Step 5 can produce more equations for a new ``test" integrand, $I_{\rm test}(\Delta)$. The new constrains may not only fix $I_{\rm test}(\Delta)$ but also $I_{\rm test'}(\Delta)$.

It is important to mention that there are further subtleties for $(3,9)$, $(4,8)$, and beyond as they have pseudo-GCOs that are not GCOs. This means some combinatorial simplex flips of a valid GCO may fail to lead to valid GCOs. In this work, we choose to discard such simplex flips when we apply Algorithm I. 

This way, 
using the modified version of the algorithm, we computed 
all $4381$ types of $(3,9)$ integrands in terms of only four unknown parameters. We also computed $2469$ types of $(4,8)$ integrands in terms of $24$ parameters, while the remaining $2604-2469=135$ types of integrands are left for future research since their computation requires a significant amount of computing resources.  

All our results for integrands are present in an ancillary file.

Our findings for $(3,9)$ show that starting at $n=9$ one has to combine the algorithm with some explicit evaluations of the CEGM formula to fix the undetermined coefficients by comparing the results to a sum over generalized Feynman diagrams. We expect that a similar phenomenon should also happen to $(4,8)$. It is tempting to relate the presence of unfixed parameters to that of non-realizable pseudo-GCOs. However, at this point, we do not have a direct way to link them.

\subsection{Examples
 \label{secExamples}}

Let us present the $(3,6)$ and $(3,7)$ results here and some representatives from $(3,8)$, $(3,9)$ and   
  $(4,8)$ integrands.

	\def \inter #1,#2 \inter { (intersection of  A#1--B#1 and A#2--B#2) }

\def \tri #1,#2,#3 \tri {  \inter #1,#2 \inter   -- \inter #2,#3 \inter  --  \inter #3,#1 \inter  --\inter #1,#2 \inter }

	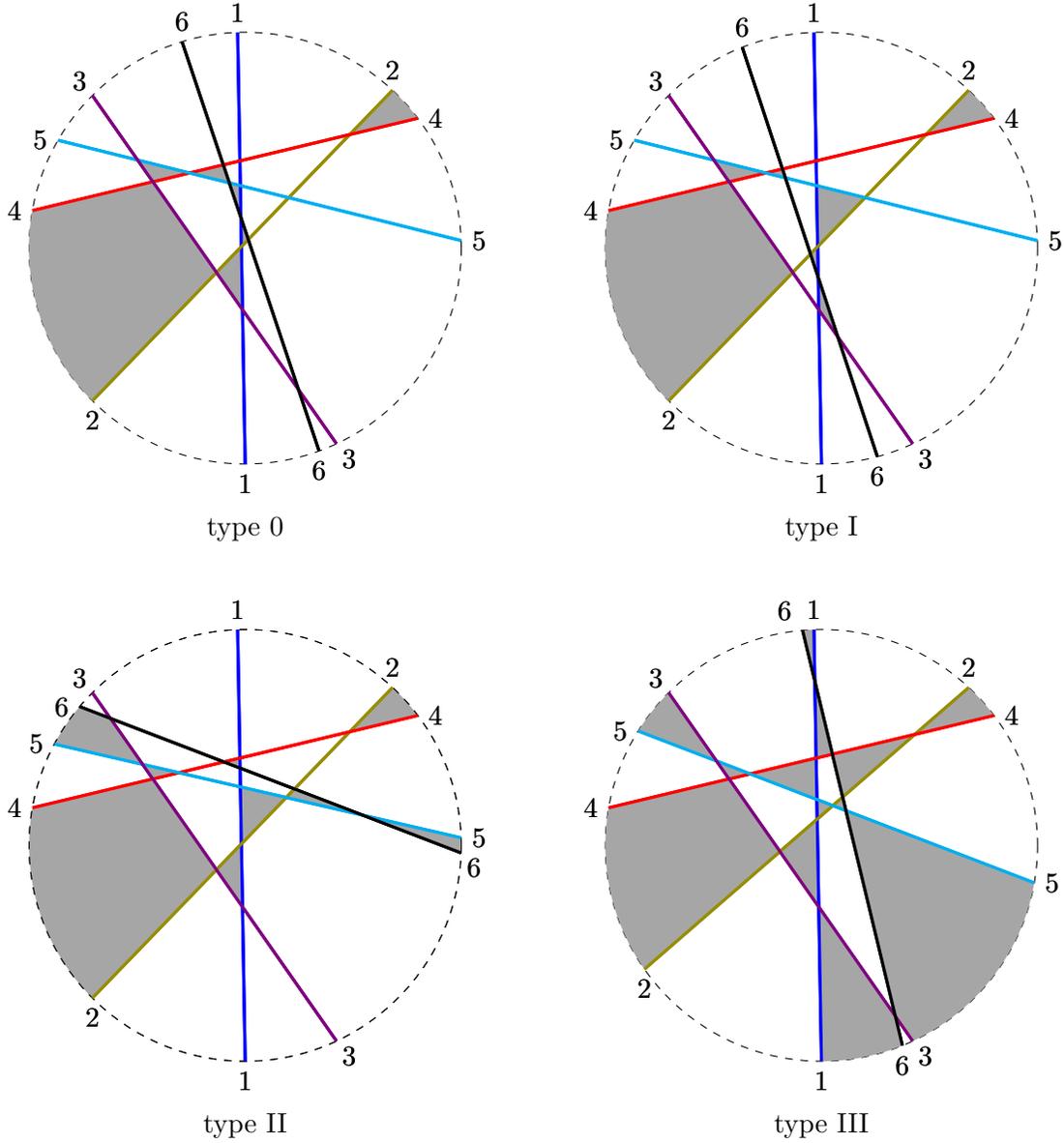
\begin{figure}[h!]
	\centering

 \begin{tikzpicture}[scale=.99]
  
\begin{scope}[xshift=-4cm,yshift=4cm,scale=3]

         \draw[ dashed](0,0) circle (1);
         
    \draw[blue,very thick] (90+2:1) coordinate(A1)  node [above]{\color{black}1} -- (-90:1) coordinate(B1)  node [below]{\color{black}1}  ; 
    
     \draw[olive,very thick] (45+2:1) coordinate(A2)  node [above]{\color{black}2} -- (-115-20:1) coordinate(B2)  node [below]{\color{black}2} ; 
     
    \draw[red,very thick] (35+2:1) coordinate(A3)  node [right]{\color{black}4} -- (-180-10:1)  coordinate(B3)  node [left]{\color{black}4} ; 
          
      \draw[cyan,very thick] (0+2:1) coordinate(A4)  node [right]{\color{black}5} -- (-180-30:1) coordinate(B4)  node [left]{\color{black}5} ; 
                   
    \draw[violet,very thick] (-65:1) coordinate(A5)  node [below right=-2pt]{\color{black}3} -- (-180-45:1) coordinate(B5)  node [ above left=-2pt]{\color{black}3} ;

    \draw[black,very thick] (110-3:1) coordinate(A6)   node [above]{\color{black}6} -- (-70:1) coordinate(B6)  node [below]{\color{black}6}  ;

    \begin{scope}
    
       
       \clip (0,0) circle (1cm);
    
    \fill [fill=gray!70, very thick]  \tri 1,4,6 \tri ;
    
        \fill [fill=gray!70, very thick]  \tri 1,2,6 \tri ;
        
          \fill [fill=gray!70, very thick]  \tri 1,2,5 \tri ;
          
                   \fill [fill=gray!70, very thick]  \tri 3,4,5 \tri ;
                   
                         \fill [fill=gray!70, very thick]  \tri 3,4,6 \tri ;

     \fill [fill=gray!70, very thick] \inter 2,3 \inter -- (A2)--+(A2) --(A3) --\inter 2,3 \inter;

          \fill [fill=gray!70, very thick] \inter 3,5 \inter --\inter 2,5 \inter --  (B2)--+(B2) --(B3) --\inter 2,3 \inter;

                 \fill [fill=gray!70, very thick] \inter 3,5 \inter --\inter 2,5 \inter --  (B2)--+(B2) --(B3) --\inter 2,3 \inter;
    
           \end{scope}

       \draw[blue,very thick] (90+2:1) coordinate(A1)  node [above]{\color{black}1} -- (-90:1) coordinate(B1)  node [below]{\color{black}1}  ; 
    
     \draw[olive,very thick] (45+2:1) coordinate(A2)  node [above]{\color{black}2} -- (-115-20:1) coordinate(B2)  node [below]{\color{black}2} ; 
     
    \draw[red,very thick] (35+2:1) coordinate(A3)  node [right]{\color{black}4} -- (-180-10:1)  coordinate(B3)  node [left]{\color{black}4} ; 
          
      \draw[cyan,very thick] (0+2:1) coordinate(A4)  node [right]{\color{black}5} -- (-180-30:1) coordinate(B4)  node [left]{\color{black}5} ; 
                   
    \draw[violet,very thick] (-65:1) coordinate(A5)  node [below right=-2pt]{\color{black}3} -- (-180-45:1) coordinate(B5)  node [ above left=-2pt]{\color{black}3} ;

    \draw[black,very thick] (110-3:1) coordinate(A6)   node [above]{\color{black}6} -- (-70:1) coordinate(B6)  node [below]{\color{black}6}  ;

    \node at (-90:1.3) {type 0};

\end{scope}

\begin{scope}[xshift=4cm, yshift=4cm,scale=3]
  \draw[ dashed](0,0) circle (1);
  
    \draw[blue,very thick] (90+2:1)  coordinate(A1)  node [above]{\color{black}1} -- (-90:1)  coordinate(B1)  node [below]{\color{black}1}  ; 
    
     \draw[olive,very thick] (45+2:1)  coordinate(A2)  node [above]{\color{black}2} -- (-115-20:1)  coordinate(B2)  node [below]{\color{black}2} ; 
     
    \draw[red,very thick] (35+2:1)  coordinate(A3)  node [right]{\color{black}4} -- (-180-10:1)  coordinate(B3)  node [left]{\color{black}4} ; 
          
      \draw[cyan,very thick] (0+2:1)  coordinate(A4)  node [right]{\color{black}5} -- (-180-30:1)  coordinate(B4)  node [left]{\color{black}5} ; 
                   
    \draw[violet,very thick] (-65:1)  coordinate(A5)  node [below right=-2pt]{\color{black}3} -- (-180-45:1)  coordinate(B5)  node [ above left=-2pt]{\color{black}3} ;

    \draw[black,very thick] (110+1.5:1)  coordinate(A6)  node [above]{\color{black}6} -- (-75:1)  coordinate(B6)  node [below]{\color{black}6}  ;

    \begin{scope}
    
       
       \clip (0,0) circle (1cm);
    
    \fill [fill=gray!70, very thick]  \tri 1,2,6 \tri ;
    
        \fill [fill=gray!70, very thick]  \tri 1,2,6 \tri ;
        
          \fill [fill=gray!70, very thick]  \tri 1,6,5 \tri ;
          
                   \fill [fill=gray!70, very thick]  \tri 3,4,5 \tri ;
                   
                         \fill [fill=gray!70, very thick]  \tri 3,4,6 \tri ;
                         
                           \fill [fill=gray!70, very thick]  \tri 1,2,4 \tri ;

     \fill [fill=gray!70, very thick] \inter 2,3 \inter -- (A2)--+(A2) --(A3) --\inter 2,3 \inter;

          \fill [fill=gray!70, very thick] \inter 3,5 \inter --\inter 2,5 \inter --  (B2)--+(B2) --(B3) --\inter 2,3 \inter;

                 \fill [fill=gray!70, very thick] \inter 3,5 \inter --\inter 2,5 \inter --  (B2)--+(B2) --(B3) --\inter 2,3 \inter;
    
           \end{scope}

    \draw[blue,very thick] (90+2:1)  coordinate(A1)  node [above]{\color{black}1} -- (-90:1)  coordinate(B1)  node [below]{\color{black}1}  ; 
    
     \draw[olive,very thick] (45+2:1)  coordinate(A2)  node [above]{\color{black}2} -- (-115-20:1)  coordinate(B2)  node [below]{\color{black}2} ; 
     
    \draw[red,very thick] (35+2:1)  coordinate(A3)  node [right]{\color{black}4} -- (-180-10:1)  coordinate(B3)  node [left]{\color{black}4} ; 
          
      \draw[cyan,very thick] (0+2:1)  coordinate(A4)  node [right]{\color{black}5} -- (-180-30:1)  coordinate(B4)  node [left]{\color{black}5} ; 
                   
    \draw[violet,very thick] (-65:1)  coordinate(A5)  node [below right=-2pt]{\color{black}3} -- (-180-45:1)  coordinate(B5)  node [ above left=-2pt]{\color{black}3} ;

    \draw[black,very thick] (110+1.5:1)  coordinate(A6)  node [above]{\color{black}6} -- (-75:1)  coordinate(B6)  node [below]{\color{black}6}  ;

        \node at (-90:1.3) {type I};

\end{scope}

\begin{scope}[xshift=-4cm,yshift=-4.3cm, scale=3]

  \draw[ dashed](0,0) circle (1);
    \draw[blue,very thick] (90+2:1) coordinate(A1) node [above]{\color{black}1} -- (-90:1) coordinate(B1) node [below]{\color{black}1}  ; 
    
     \draw[olive,very thick] (45+2:1) coordinate(A2) node [above]{\color{black}2} -- (-115-20:1) coordinate(B2) node [below]{\color{black}2} ; 
     
    \draw[red,very thick] (35+2:1) coordinate(A3) node [right]{\color{black}4} -- (-180-10:1) coordinate(B3) node [left]{\color{black}4} ; 
          
      \draw[cyan,very thick] (0+2:1) coordinate(A4) node [right]{\color{black}5} -- (-180-30+2:1) coordinate(B4) node [left]{\color{black}5} ; 
                   
    \draw[violet,very thick] (-65:1) coordinate(A5) node [below right=-2pt]{\color{black}3} -- (-180-45:1) coordinate(B5) node [ above left=-2pt]{\color{black}3} ;

    \draw[black,very thick] (-2:1) coordinate(A6) node [below right=-2pt]{\color{black}6} -- (-180-40:1) coordinate(B6) node [left]{\color{black}6}  ;

       \begin{scope}
    
       
       \clip (0,0) circle (1cm);
    
    \fill [fill=gray!70, very thick]  \tri 1,5,2 \tri ;
        \fill [fill=gray!70, very thick]  \tri 1,2,4 \tri ;
                \fill [fill=gray!70, very thick]  \tri 2,4,6 \tri ;
                                \fill [fill=gray!70, very thick]  \tri 1,3,6 \tri ;
                                
                                    \fill [fill=gray!70, very thick]  \tri 3,4,5 \tri ;

             \fill [fill=gray!70, very thick] \inter 4,6 \inter -- (A4)--+(A4) --(A6) --\inter 4,6 \inter;

        \fill [fill=gray!70, very thick] \inter 5,6 \inter --\inter 4,5 \inter --  (B4)--+(B4) --(B6) --\inter 5,6 \inter;

     \fill [fill=gray!70, very thick] \inter 2,3 \inter -- (A2)--+(A2) --(A3) --\inter 2,3 \inter;

          \fill [fill=gray!70, very thick] \inter 3,5 \inter --\inter 2,5 \inter --  (B2)--+(B2) --(B3) --\inter 2,3 \inter;

           \end{scope}

     \draw[ dashed](0,0) circle (1);
    \draw[blue,very thick] (90+2:1) coordinate(A1) node [above]{\color{black}1} -- (-90:1) coordinate(B1) node [below]{\color{black}1}  ; 
    
     \draw[olive,very thick] (45+2:1) coordinate(A2) node [above]{\color{black}2} -- (-115-20:1) coordinate(B2) node [below]{\color{black}2} ; 
     
    \draw[red,very thick] (35+2:1) coordinate(A3) node [right]{\color{black}4} -- (-180-10:1) coordinate(B3) node [left]{\color{black}4} ; 
          
      \draw[cyan,very thick] (0+2:1) coordinate(A4) node [right]{\color{black}5} -- (-180-30+2:1) coordinate(B4) node [left]{\color{black}5} ; 
                   
    \draw[violet,very thick] (-65:1) coordinate(A5) node [below right=-2pt]{\color{black}3} -- (-180-45:1) coordinate(B5) node [ above left=-2pt]{\color{black}3} ;

    \draw[black,very thick] (-2:1) coordinate(A6) node [below right=-2pt]{\color{black}6} -- (-180-40:1) coordinate(B6) node [left]{\color{black}6}  ;

      \node at (-90:1.3) {type II};
\end{scope}

\begin{scope}[xshift=4cm,yshift=-4.3cm, scale=3]

  \draw[ dashed](0,0) circle (1);

    \draw[blue,very thick] (90+2:1) coordinate(A1) node [above]{\color{black}1} -- (-90:1) coordinate(B1) node [below]{\color{black}1}  ; 
    
     \draw[olive,very thick] (45+2:1) coordinate(A2) node [above]{\color{black}2} -- (-115-30:1) coordinate(B2) node [below]{\color{black}2} ; 
     
    \draw[red,very thick] (35+2:1) coordinate(A3) node [right]{\color{black}4} -- (-180-10:1) coordinate(B3) node [left]{\color{black}4} ; 
          
      \draw[cyan,very thick] (0-10:1) coordinate(A4) node [right]{\color{black}5} -- (-180-30-2:1) coordinate(B4) node [left]{\color{black}5} ; 
                   
    \draw[violet,very thick] (-65:1) coordinate(A5) node [below right=-2pt]{\color{black}3} -- (-180-45:1) coordinate(B5) node [ above left=-2pt]{\color{black}3} ;

    \draw[black,very thick] (-65-3:1) coordinate(A6) node [below=0pt]{\color{black}6} -- (90+5:1) coordinate(B6) node [above left]{\color{black}6}  ;

    \begin{scope}
    
       
       \clip (0,0) circle (1cm);
    
   \fill [fill=gray!70, very thick]  \tri 3,4,5 \tri ;

             \fill [fill=gray!70, very thick]  \tri 1,3,6 \tri ;    
             
               \fill [fill=gray!70, very thick]  \tri 1,2,4 \tri ;

                \fill [fill=gray!70, very thick]  \tri 1,3,4 \tri ;    
                
                   \fill [fill=gray!70, very thick]  \tri 2,4,6 \tri ;    
                   
                                      \fill [fill=gray!70, very thick]  \tri 2,3,6 \tri ;   
                                       
                                        \fill [fill=gray!70, very thick]  \tri 1,2,5 \tri ;

                    \fill [fill=gray!70, very thick] \inter 1,6 \inter -- (B6)--+(B6) --(A1) --cycle;

        \fill [fill=gray!70, very thick] \inter 1,5 \inter --\inter 5,6 \inter --  (A6)--+(A6) --(B1) --cycle;

             \fill [fill=gray!70, very thick] \inter 4,5 \inter -- (B4)--+(B4) --(B5) --cycle;

        \fill [fill=gray!70, very thick] \inter 4,6 \inter --\inter 5,6 \inter --  (A5)--+(A5) --(A4) --cycle;

                     \fill [fill=gray!70, very thick] \inter 2,3 \inter -- (A2)--+(A2) --(A3) --cycle;

        \fill [fill=gray!70, very thick] \inter 3,5 \inter --\inter 2,5 \inter --  (B2)--+(B2) --(B3) --cycle;

           \end{scope}

\draw[blue,very thick] (90+2:1) coordinate(A1) node [above]{\color{black}1} -- (-90:1) coordinate(B1) node [below]{\color{black}1}  ; 
    
     \draw[olive,very thick] (45+2:1) coordinate(A2) node [above]{\color{black}2} -- (-115-30:1) coordinate(B2) node [below]{\color{black}2} ; 
     
    \draw[red,very thick] (35+2:1) coordinate(A3) node [right]{\color{black}4} -- (-180-10:1) coordinate(B3) node [left]{\color{black}4} ; 
          
      \draw[cyan,very thick] (0-10:1) coordinate(A4) node [right]{\color{black}5} -- (-180-30-2:1) coordinate(B4) node [left]{\color{black}5} ; 
                   
    \draw[violet,very thick] (-65:1) coordinate(A5) node [below right=-2pt]{\color{black}3} -- (-180-45:1) coordinate(B5) node [ above left=-2pt]{\color{black}3} ;

    \draw[black,very thick] (-65-3:1) coordinate(A6) node [below=0pt]{\color{black}6} -- (90+5:1) coordinate(B6) node [above left]{\color{black}6}  ;

    \node at (-90:1.3) {type III};

\end{scope}

\end{tikzpicture}

 \caption{Examples of GCOs corresponding to the four types of $(3,6)$ GCOs. All triangles in each arrangement of lines have been shaded. There are 6,6,7, and 10 triangles respectively. 
 To recognize the triangle bounded by lines $L_2$, $L_3$, $L_4$ in the first graph requires using that points on the boundary of the disk marked with the same label are
identified.
 \label{triangles36} }
\end{figure}

\subsubsection{$(3,6)$ Integrands}

We give a full explanation on 
 how to construct the $(3,6)$ integrands.  There are four types of GCOs, with representatives,
\begin{align}
\label{type036GCO}
& \Sigma_{0}   =((2 3 4 5 6),(1 3 4 5 6),(1 2 4 5 6),(1 2 3 5 6),(1 2 3 4 6),(1 2 3 4 5)),  \\
\label{typeI36GCO}
&\Sigma_{I}   = ((2 5 4 3 6),(1 5 4 3 6),(1 2 4 5 6),(1 2 3 5 6),(1 2 3 4 6),(1 2 5 4 3)),\\
\label{typeII36GCO}
&\Sigma_{II}   =((2  3465 ),(1 3465 ),(1 2 4 5 6),(1 2 3 5 6),(1 2 6 3 4),(1 2 5 3 4)), \\
\label{typeIII36GCO}
& \Sigma_{III}   =((2  3645 ),(1 3465 ),(1 2 4 5 6),(1 5 3 2 6),(1 2 6 3 4),(1 3524 )).
\end{align}
Using the combinatorial technique explained in \cref{sec3d1} (or in \cref{appa} for general polygons), one finds that they have 6,6,7 and 10 triangles respectively,
\begin{align}
&\{\{1,2,3\},\{1,2,6\},\{1,5,6\},\{2,3,4\},\{3,4,5\},\{4,5,6\}
\}\,,
\\
&
\{
\{1,2,5\},\{1,2,6\},\{1,3,6\},\{2,3,4\},\{3,4,5\},\{4,5,6\}
\}\,,
\\
&
\{
\{1,2,3\},\{1,2,5\},\{1,4,6\},\{2,3,4\},\{2,5,6\},\{3,4,5\},\{3,5,6\}
\}\,,
\\
&
\{
\{1,2,3\},\{1,2,5\},\{1,3,6\},\{1,4,5\},\{1,4,6\},\{2,3,4\},\{2,4,6\},\{2,5,6\},\{3,4,5\},\{3,5,6\}\}\,.
\end{align}
 All triangles in the arrangements of lines are shown explicitly in \cref{triangles36}. Even though the figures are intuitive, 
the combinatorial way of finding all triangles is very efficient even for large values of $n$. 

Obviously,
type 0 and type I integrands are of basic type 
with $p = 0$. So they have trivial numerators,
\begin{align}
\label{36integrand0}
& {\cal I}(\Sigma_0) = \frac{1}{\Delta_{123}\Delta_{234}\Delta_{345}\Delta_{456}\Delta_{561}\Delta_{612}}\,,
\\ 
\label{36integrandI}
&{\cal I}(\Sigma_I) = \frac{1}{\Delta_{125} \Delta_{234}\Delta_{345}\Delta_{456}\Delta_{136} \Delta_{612}}\,.
\end{align}     
Type II is also of the basic type with $p=1$. Hence a single minor is needed in the numerator.  We see
 lines $L_2,L_4$ and $L_5$ participate in four triangles while $L_1,L_3$ and $L_6$ in three, which means  the minor is  $\Delta_{235}$,
\be     
\label{36integrandII}
{\cal I}(\Sigma_{II}) = \frac{\Delta _{235}}{\Delta _{123} \Delta _{125} \Delta _{146} \Delta _{234} \Delta _{256} \Delta _{345} \Delta _{356}} .
\ee  

The last of the $(3,6)$ GCOs is type III. This GCO has ten triangles. This means that the numerator of ${\cal I}(\Sigma_{III})$ has to be a degree four polynomial in the minors. Moreover, it easy to see that each line participates in five triangles and therefore the numerator must be homogeneous of degree two in each label. 
The most general form of the numerator is 
\begin{align}
\label{p4ansatz}
    P_4(\Delta)=&
x_1\Delta_{123}\Delta_{345}\Delta_{561}\Delta_{246}+x_2\Delta_{234}\Delta_{456}\Delta_{612}\Delta_{351}+x_3 \Delta _{156}^2 \Delta _{234}^2+x_4 \Delta _{145} \Delta _{156} \Delta _{234} \Delta _{236}
  \nonumber
  \\
  &+x_5 \Delta _{145}^2 \Delta _{236}^2+x_6 \Delta _{126} \Delta _{156} \Delta _{234} \Delta _{345}+x_7 \Delta _{126} \Delta _{145} \Delta _{236} \Delta _{345}
 +x_8 \Delta _{126}^2 \Delta _{345}^2
   \nonumber
   \\&+x_9 \Delta _{124} \Delta _{156} \Delta _{234} \Delta _{356}+x_{10} \Delta _{124} \Delta _{126} \Delta _{345} \Delta _{356}+x_{11} \Delta _{124}^2 \Delta _{356}^2
     \nonumber
     \\& +x_{12} \Delta _{123} \Delta _{156} \Delta _{234} \Delta _{456}
     +x_{13} \Delta _{123} \Delta _{145} \Delta _{236} \Delta _{456}+x_{14} \Delta _{123} \Delta _{126} \Delta _{345} \Delta _{456}\nonumber \\ & +x_{15} \Delta _{123} \Delta _{124} \Delta _{356} \Delta _{456}+x_{16} \Delta _{123}^2 \Delta _{456}^2\,.
\end{align}
Other monomials with correct torus weight such as $\Delta _{124} \Delta _{156} \Delta _{236} \Delta _{356}$ can be expanded in terms of those appearing in \eqref{p4ansatz}.

The ansatz for the integrand thus becomes
\begin{equation}
\label{p4ansatzinte}
    {\cal I}(\Sigma_{III}) =\frac{P_4(\Delta)}{\Delta_{145}\Delta_{136}\Delta_{234}\Delta_{256}\Delta_{125}\Delta_{356}\Delta_{345}\Delta_{146}\Delta_{246}\Delta_{123}} .
\end{equation}
We now study its residues. Consider first the pole at $\Delta_{345}=0$. Performing the corresponding triangle flip it is easy to identify the new GCO as of type II with integrand,
\be 
\label{p4neighbour}
\frac{\Delta_{126}}{\Delta_{125} \Delta_{123} \Delta_{146} \Delta_{136}
  \Delta_{246} \Delta_{256} \Delta_{345}}.
\ee  
Factoring out $1/\Delta_{345}$ in both expressions and requiring that they match gives 
\be\label{poly345}   
P_4(\Delta)|_{\Delta_{345}=0} = \Delta_{1 AB} \Delta_{2 AB}  \Delta_{AB 6} \Delta_{1 2 6} .
\ee
where $AB$ is the line that contains points $3,4,5$ in the configuration 
 of  points in ${\mathbb {RP}}^2$.  This implies that $P_4(\Delta)$ must vanish if $\Delta_{345}=\Delta_{1 2 6} =0$. Imposing this condition to the ansatz \eqref{p4ansatz} fixes 9 of the 16 parameters,
\begin{align}
    \{x_3\to 0,\,x_4\to 0,\,x_5\to 0,\,x_9\to 0,\,x_{11}\to 0,\,x_{12}\to 0,\,x_{13}\to 0,x_{15}\to 0,\,x_{16}\to 0\}\,.
\end{align}

Repeating the computation but with the pole at $\Delta_{256}=0$ gives
\be\label{poly246}   
P_4(\Delta)|_{\Delta_{256}=0} = - \Delta_{1 AB} \Delta_{5 AB}  \Delta_{AB 3} \Delta_{135}\,,
\ee
where $AB$ is the line  that contains points $2,4,6$. This implies that $P_4(\Delta)$ must vanish if $\Delta_{256}=\Delta_{135} =0$.  This time, the condition fixes three parameters,
\begin{align}
    \{x_8\to x_7,\,\,\,x_{10}\to x_6,\,\,\,x_{14}\to x_1-x_6\}\, .
\end{align}

Repeating the computation for all the remaining 8 poles in \eqref{p4ansatzinte}, we fix all 16 parameters up to an overall rescaling,
\begin{align}
     P_4(\Delta)\to P'_4(\Delta)=&
x_1\Delta_{123}\Delta_{345}\Delta_{561}\Delta_{246}-x_1\Delta_{234}\Delta_{456}\Delta_{612}\Delta_{351}\,.
\end{align}
We see how powerful the constraints \eqref{firstutilization} are and $P'_4(\Delta)$ can be thought of as a simplified version of the ansatz.  At this point, it is much easier to match (up to a sign) the residue of \eqref{p4ansatzinte} with that of its neighbours such as \eqref{p4neighbour} using \eqref{secnodutilization}. Doing so, we finally get $x_1\to1$.

This gives
\begin{equation}
\label{inteIII}
    {\cal I}(\Sigma_{III}) =\frac{\Delta_{123}\Delta_{345}\Delta_{561}\Delta_{246}-\Delta_{234}\Delta_{456}\Delta_{612}\Delta_{351} }{\Delta _{123} \Delta _{125} \Delta _{136} \Delta _{145} \Delta _{146} \Delta _{234} \Delta _{246} \Delta _{256} \Delta _{345} \Delta _{356}} .
\end{equation}

For this particular example,
there is a more beautiful explanation for the numerator of \eqref{inteIII}.
Note that the structure of \eqref{poly345} and \eqref{poly246} is the same. In fact, the same is true for all ten poles and we conclude that $P_4(\Delta)$ is a polynomial that must vanish whenever the six points are on any of the $\binom{6}{3}/2=10$ configurations of two lines with three points in each. All these configurations are degenerate conics and it is well known that the Veronese polynomial vanishes on all of them,  
\be    
V:= \Delta_{123}\Delta_{345}\Delta_{561}\Delta_{246}-\Delta_{234}\Delta_{456}\Delta_{612}\Delta_{351}.
\ee    
It is surprising at first that $V$ cannot do what is needed since it manifestly vanishes at only four of the ten configurations, e.g. when $\Delta_{123}=\Delta_{456}=0$. However, one can check that under any permutation of labels, $V$ either stays the same or becomes $-V$. 

Let us mention that these four functions give rise to canonical forms in the sense of \cite{Arkani-Hamed:2020cig,Arkani-Hamed:2017tmz} for the regions closely related to the chambers of $X(3,6)$. In fact, they were independently obtained in \cite{unpublishedYong} using ``triangulations'' to build canonical forms. More  relevant discussions are put in \cref{gcotochirotope}.

\subsubsection{ $(3,7)$, $(4,7)$ and $(3,8)$ Integrands}

In \cref{sec3d2} we explained how to use Algorithm I to construct all $(3,6)$ integrands. The same procedure can be directly applied to construct $(3,7)$,$(4,7)$, and $(3,8)$ integrands.

There are eleven types of $(3,7)$ integrands, with seven of them of basic type: $4$ with $p=0$ and $3$ with $p=1$. Here we list a representative for each of the eleven types:
\allowdisplaybreaks[1]
\begin{align}
\label{37integrand0}
& {\cal I}^{(3,7)}(\Sigma_0 ) = \frac{1}{\Delta _{123} \Delta _{127} \Delta _{167} \Delta _{234} \Delta _{345} \Delta _{456} \Delta_{567}},
\\ 
& {\cal I}^{(3,7)}(\Sigma_I ) = \frac{1}{\Delta _{123} \Delta _{127} \Delta _{156} \Delta _{234} \Delta _{345} \Delta _{467} \Delta _{567}},
\\
&
{\cal I}^{(3,7)}(\Sigma_{II} ) = \frac{\Delta _{347}}{\Delta _{123} \Delta _{127} \Delta _{156} \Delta _{234} \Delta _{345} \Delta _{367} \Delta _{457} \Delta _{467}},
\\
& {\cal I}^{(3,7)}(\Sigma_{III} ) =\frac{1}{\Delta _{123} \Delta _{127} \Delta _{156} \Delta _{234} \Delta _{367} \Delta _{456} \Delta _{457}},
\\
&
{\cal I}^{(3,7)}(\Sigma_{IV} ) =\frac{\Delta _{237} \Delta _{567}}{\Delta _{123} \Delta _{127} \Delta _{156} \Delta _{234} \Delta _{267} \Delta _{357} \Delta _{367} \Delta _{456} \Delta _{457}},
\qquad\qquad
\\
&
{\cal I}^{(3,7)}(\Sigma_{V} ) =\frac{\Delta _{567}}{\Delta _{123} \Delta _{156} \Delta _{167} \Delta _{234} \Delta _{267} \Delta _{357} \Delta _{456} \Delta _{457}},
\\
&
{\cal I}^{(3,7)}(\Sigma_{VI} ) =\frac{\Delta _{127} \Delta _{145} \Delta _{236} \Delta _{567}-\Delta _{123} \Delta _{157} \Delta _{267} \Delta _{456}}{\Delta _{123} \Delta _{126} \Delta _{147} \Delta _{156} \Delta _{157} \Delta _{234} \Delta _{257} \Delta _{267} \Delta _{345} \Delta _{367} \Delta _{456}},
\\
&
\label{37sigma7}
{\cal I}^{(3,7)}(\Sigma_{VII} ) =
\frac{\Delta _{136} \Delta _{245}-\Delta _{126} \Delta _{345}}{\Delta _{123} \Delta _{126} \Delta _{147} \Delta _{156} \Delta _{234} \Delta _{257} \Delta _{345} \Delta _{367} \Delta _{456}}
\\
& {\cal I}^{(3,7)}(\Sigma_{VIII} ) =\frac{1}{\Delta _{123} \Delta _{145} \Delta _{167} \Delta _{234} \Delta _{267} \Delta _{357} \Delta _{456}},
\\
&
{\cal I}^{(3,7)}(\Sigma_{XI} ) =-\frac{\Delta _{237} \Delta _{345} \Delta _{567}}{\Delta _{123} \Delta _{145} \Delta _{167} \Delta _{234} \Delta _{257} \Delta _{267} \Delta _{347} \Delta _{356} \Delta _{357} \Delta _{456}},
\\
& {\cal I}^{(3,7)}(\Sigma_{X} ) =\frac{\Delta _{235}}{\Delta _{123} \Delta _{126} \Delta _{145} \Delta _{234} \Delta _{257} \Delta _{356} \Delta _{357} \Delta _{467}}\,,
\label{37integrand10}
\end{align}
\allowdisplaybreaks[0]
 where the 11 $(3,7)$ GCOs $\Sigma_{\bullet}$ of different types  are put in \cref{37gco}.

Similarly, we can get the eleven types of $(4,7)$ integrands using Algorithm I. As expected, they are dual to the $(3,7)$ integrands by sending $\Delta_{i,j,k}\to
{\rm sign}(i,j,k,o,p,q,r) \Delta_{o,p,q,r}$ where the sign is determined by an ordering. For example, we have 
\be 
{\cal I}^{(4,7)}(\Sigma'_{VII} ) =
\frac{
-\Delta _{2457} \Delta _{1367}+\Delta _{3457} \Delta _{1267}}{\Delta _{4567} \Delta _{3457} \Delta _{2356} \Delta _{2347} \Delta _{1567} \Delta _{1346} \Delta _{1267} \Delta _{1245} \Delta _{1237}}\,,
\ee
where the (4,7) GCO $\Sigma'_{VII}$ is dual to the (3,7) GCO $\Sigma_{VII}$ used in \eqref{37sigma7}. 

There are $135$ types of $(3,8)$ integrands, with $43$ of basic type: $11$ with $p=0$ and $32$ with $p=1$.

Here we present two examples, one with $p=2$ and one with $p=3$ and provide a complete list of integrands as well as their corresponding GCOs in the ancillary file,
\allowdisplaybreaks[1]
\begin{align}
&  {\cal I}^{(3,8)}(\Sigma_{43}) = \frac{\Delta _{238} \Delta _{678}}{\Delta _{123} \Delta _{128} \Delta _{156} \Delta _{234}
   \Delta _{278} \Delta _{368} \Delta _{378} \Delta _{458} \Delta _{467} \Delta _{567}},
\\   
&  {\cal I}^{(3,8)}(\Sigma_{75}) = \frac{\Delta _{238} \Delta _{568} \Delta _{678}}{\Delta _{123} \Delta _{156} \Delta _{178}
   \Delta _{234} \Delta _{268} \Delta _{278} \Delta _{358} \Delta _{368} \Delta _{458} \Delta
   _{467} \Delta _{567}}.  
\end{align}
\allowdisplaybreaks[0]
The subscripts in $\Sigma_{43}$ and $\Sigma_{73}$ are the locations in the list of $135$ pairs of integrands and GCOs in the ancillary file. \footnote{In the list of pairs of integrands and GCOs for $(3,8)$ in the ancillary file, the GCO in the first pair is of type 0 and denoted as $\Sigma_0$ and the remaining denoted as $\Sigma_1,\Sigma_2, \cdots$ consecutively. So are the lists for $(3,9)$ and $(4,8)$.} These (3,8) GCOs are also shown in Appendix B of \cite{Cachazo:2022pnx}.

We comment that starting at $(3,7)$, there are different GCOs whose corresponding arrangements of lines have the same set of triangles. For example, it is easy to check the following two (3,7) GCOs of type VII 
\begin{align}
\label{eq37two}
    ((236457),(134567),(124567),(162357),(162347),(145237),(123645))\,,
    \\
    ((234657),(134657),(124576),(123576),(126347),(125734),(126345))\,,
\end{align}
share the same set of nine triangles,
\begin{align}
    \{\{1,2,3\},\{1,2,7\},\{1,4,6\},\{1,5,7\},\{2,3,4\},\{2,5,6\},\{3,4,5\},\{3,6,7\},\{4,5,7\}\}\,,
\end{align}
which are shown in \cref{two37arrangements}.
In the same figure, we can clearly see they also share a hexagon, with edges $\{{1,2,3,4,5,7}\}$. However, they have different sets of quadrangles and pentagons. For example, the sets of pentagons are
\begin{align}
 \{\{1,3,2,4,6\},\{1,5,6,2,7\},\{3,5,4,7,6\}\}\,,
    \\
\{ \{1,2,7,6,3\},\{1,5,7,4,6\},\{2,4,3,5,6\} \}\, ,
\end{align}
respectively. 
This means that the two GCOs are indeed different.

In the appendix, we generalize the method introduced in \cite{Cachazo:2022pnx} for finding triangles to one for finding all polygons based on a simple combinatorial rule applied to the $k=2$ color orderings in the collection that make the $k=3$ GCO. Having an efficient algorithm for determining all polygons is very useful for identifying GCOs. 

Back to the two $(3,7)$ GCOs, the corresponding integrands of type VII share the same denominators but have different numerators,
\begin{align}
    &\frac{\Delta _{127} \Delta _{345}-\Delta _{125} \Delta _{347}}{\Delta _{123} \Delta _{127} \Delta _{146} \Delta _{157} \Delta _{234} \Delta _{256} \Delta _{345} \Delta _{367} \Delta _{457}}\,,
    \\
    &
    \frac{\Delta _{127} \Delta _{345}-\Delta _{137} \Delta _{245}}{\Delta _{123} \Delta _{127} \Delta _{146} \Delta _{157} \Delta _{234} \Delta _{256} \Delta _{345} \Delta _{367} \Delta _{457}}
    \,.
\end{align}

The two GCOs are related by $2\leftrightarrow 7,3\leftrightarrow 5$ as are their integrands.

We see that for $k>2$, integrands are not always uniquely characterized by their explicit sets of poles.

We have verified that
the CEGM integrals \eqref{cegm3n} with the integrands computed in this work for all $(3,6)$, $(3,7)$   $(4,7)$ and $(3,8)$ GCOs agree with 
the ``physical'' generalized partial amplitudes $m^{(3)}_n(\Sigma, \tilde\Sigma)$ computed as sums over GFDs.
This match is also a consistency check for the list of GFDs given in \eqref{cegm3n}. For example, recall that in $(3,7)$ GFDs can have quartic vertices, a feature not seen in $(3,6)$ GFDs. Indeed, these GFD with quartic vertices are needed to reproduce the CEGM result \eqref{cegm3n}.  We provide a Mathematica notebook as an ancillary file to evaluate the CEGM integrals for any two GCOs for $(3,6), (3,7), (4,7)$ and $ (3,8)$ where we have borrowed some julia codes in \cite{Sturmfels:2020mpv} to solve the scattering equations given in  \eqref{defMeasure}. \footnote{YZ would like to thank Simon Telen for his patient explanations for the julia codes.}

\def \inter #1,#2 \inter { (intersection of  A#1--B#1 and A#2--B#2) }

\def \tri #1,#2,#3 \tri {  \inter #1,#2 \inter   -- \inter #2,#3 \inter  --  \inter #3,#1 \inter  --\inter #1,#2 \inter }
  
  \begin{figure}
  \begin {tikzpicture}[scale=1.1]
 \draw[red, very thick] (2.29745,8.01573) coordinate(A3) node [above]{3}-- (5.69628,1.17731) coordinate (B3) node [below]{3};
 \draw[blue,very thick] (1.21191,5.3566) coordinate(A7) node [left]{7}-- (6.19674,1.88746) coordinate (B7) node [right]{7};
 \draw [very thick](0.753396,4.91808) coordinate(A4) node [left]{4}-- (6.41714,3.95712) coordinate (B4) node [right]{4};
 \draw [very thick](0.809918,3.01222) coordinate(A1) node [left]{1}-- (6.40458,4.2298) coordinate (B1) node [right]{1};
 \draw[red, very thick] (0.634861,1.60483) coordinate(A5) node [left]{5}-- (6.45257,7.3897) coordinate (B5) node [right]{5};
 \draw[blue, very thick] (2.02338,0.213689) coordinate(A2) node [below]{2}-- (6.3334,7.75733) coordinate (B2) node [right]{2};
 \draw [very thick] (4.6172,0.208444) coordinate(A6) node [below]{6}-- (6.20611,7.98551) coordinate (B6) node [above]{6};


\filldraw [fill=gray!70, very thick]  \tri 1,5,7 \tri ;

\filldraw [fill=gray!70, very thick]  \tri 4,5,7 \tri ;

\filldraw [fill=gray!70, very thick]  \tri 4,5,3 \tri ;

\filldraw [fill=gray!70, very thick]  \tri 2,3,4 \tri ;

\filldraw [fill=gray!70, very thick]  \tri 1,2,3 \tri ;

\filldraw [fill=gray!70, very thick]  \tri 1,2,7 \tri ;

\filldraw [fill=gray!70, very thick]  \tri 2,5,6 \tri ;

\filldraw [fill=gray!70, very thick]  \tri 1,4,6 \tri ;

\filldraw [fill=gray!70, very thick]  \tri 3,7,6 \tri ;

  \draw[red, very thick] (2.29745,8.01573) coordinate(A3) node [above]{3}-- (5.69628,1.17731) coordinate (B3) node [below]{3};
 \draw[blue,very thick] (1.21191,5.3566) coordinate(A7) node [left]{7}-- (6.19674,1.88746) coordinate (B7) node [right]{7};
 \draw [very thick](0.753396,4.91808) coordinate(A4) node [left]{4}-- (6.41714,3.95712) coordinate (B4) node [right]{4};
 \draw [very thick](0.809918,3.01222) coordinate(A1) node [left]{1}-- (6.40458,4.2298) coordinate (B1) node [right]{1};
 \draw[red, very thick] (0.634861,1.60483) coordinate(A5) node [left]{5}-- (6.45257,7.3897) coordinate (B5) node [right]{5};
 \draw[blue, very thick] (2.02338,0.213689) coordinate(A2) node [below]{2}-- (6.3334,7.75733) coordinate (B2) node [right]{2};
 \draw [very thick] (4.6172,0.208444) coordinate(A6) node [below]{6}-- (6.20611,7.98551) coordinate (B6) node [above]{6};
 
  \end {tikzpicture}
  ~
  ~
 \begin {tikzpicture}[scale=1.1]
 \draw[red,very thick] (2.29745,8.01573) coordinate(A5) node [above]{5}-- (5.69628,1.17731) coordinate (B5) node [below]{5};
 \draw[blue,very thick] (1.21191,5.3566) coordinate(A2) node [left]{2}-- (6.19674,1.88746) coordinate (B2) node [right]{2};
 \draw [very thick](0.753396,4.91808) coordinate(A4) node [left]{4}-- (6.41714,3.95712) coordinate (B4) node [right]{4};
 \draw [very thick](0.809918,3.01222) coordinate(A1) node [left]{1}-- (6.40458,4.2298) coordinate (B1) node [right]{1};
 \draw[red, very thick] (0.634861,1.60483) coordinate(A3) node [left]{3}-- (6.45257,7.3897) coordinate (B3) node [right]{3};
 \draw[blue, very thick] (2.02338,0.213689) coordinate(A7) node [below]{7}-- (6.3334,7.75733) coordinate (B7) node [right]{7};
 \draw [very thick] (4.6172,0.208444) coordinate(A6) node [below]{6}-- (6.20611,7.98551) coordinate (B6) node [above]{6};


\filldraw [fill=gray!70, very thick]  \tri 1,3,2 \tri ;

\filldraw [fill=gray!70, very thick]  \tri 4,3,2 \tri ;

\filldraw [fill=gray!70, very thick]  \tri 4,5,3 \tri ;

\filldraw [fill=gray!70, very thick]  \tri 7,5,4 \tri ;

\filldraw [fill=gray!70, very thick]  \tri 1,7,5 \tri ;

\filldraw [fill=gray!70, very thick]  \tri 1,2,7 \tri ;

\filldraw [fill=gray!70, very thick]  \tri 7,3,6 \tri ;

\filldraw [fill=gray!70, very thick]  \tri 1,4,6 \tri ;

\filldraw [fill=gray!70, very thick]  \tri 5,2,6 \tri ;

 \draw[red,very thick] (2.29745,8.01573) coordinate(A5) node [above]{5}-- (5.69628,1.17731) coordinate (B5) node [below]{5};
 \draw[blue,very thick] (1.21191,5.3566) coordinate(A2) node [left]{2}-- (6.19674,1.88746) coordinate (B2) node [right]{2};
 \draw [very thick](0.753396,4.91808) coordinate(A4) node [left]{4}-- (6.41714,3.95712) coordinate (B4) node [right]{4};
 \draw [very thick](0.809918,3.01222) coordinate(A1) node [left]{1}-- (6.40458,4.2298) coordinate (B1) node [right]{1};
 \draw[red, very thick] (0.634861,1.60483) coordinate(A3) node [left]{3}-- (6.45257,7.3897) coordinate (B3) node [right]{3};
 \draw[blue, very thick] (2.02338,0.213689) coordinate(A7) node [below]{7}-- (6.3334,7.75733) coordinate (B7) node [right]{7};
 \draw [very thick] (4.6172,0.208444) coordinate(A6) node [below]{6}-- (6.20611,7.98551) coordinate (B6) node [above]{6};
 
  \end {tikzpicture}

  \caption{Arrangements of lines for two different $(3,7)$ GCOs related by the relabeling $2\leftrightarrow 7$ and  $3\leftrightarrow 5$ but with the same set of triangles.  It is interesting to note that if line six is removed, then the two arrangements coincide.
  \label{two37arrangements}}
\end{figure}
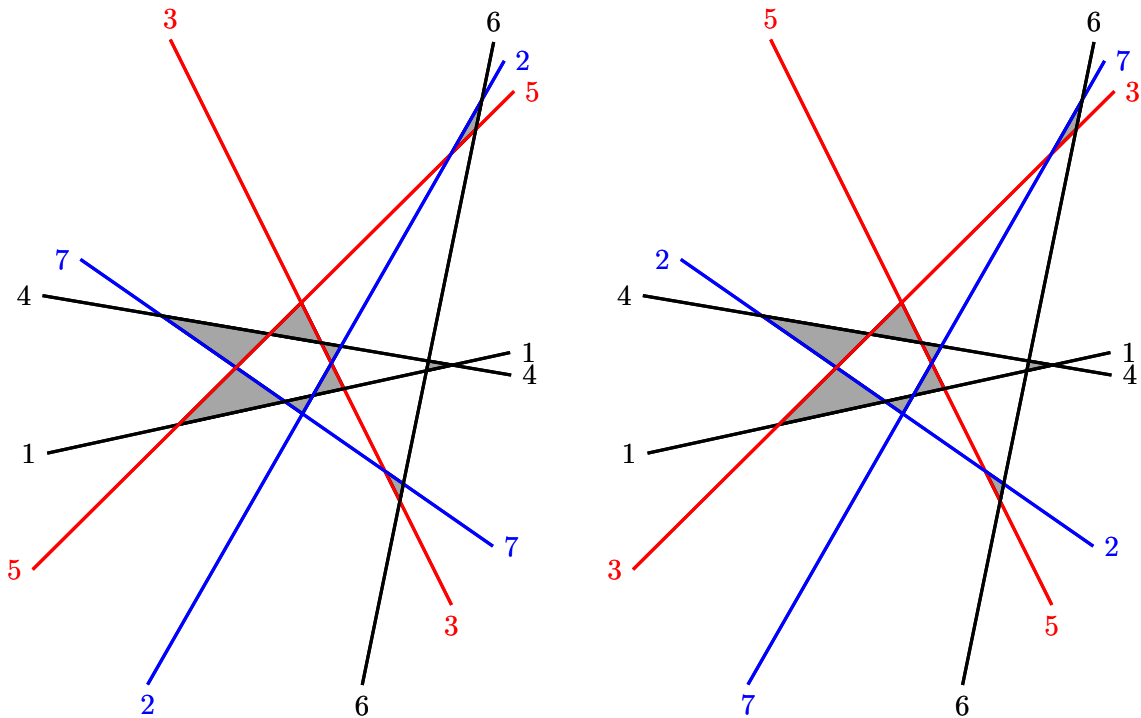

\subsubsection{$(3,9)$ Integrands\label{39integrand}}
In $(3,9)$, there are 4381 types of GCOs and one type of non-realizable pseudo-GCOs.  We focus on the triangle flips between GCOs first, match all residues using \eqref{secnodutilization} when applying Algorithm I, and deal with the subtleties relevant to the non-realizable pseudo-GCOs later. 

This way, though cumbersome,  we get $3582$ types of $(3,9)$ integrands without ambiguities and $798$ other types in terms of four free parameters. The GCO for the last missing type of integrands has 21 triangles, and we provide a more clever way to determine its integrand later in \cref{irre39sec}  by making use of irreducible decoupling identities \footnote{This integrand turns out to contain all of the four free parameters.}.
Among the $4381$ types of $(3,9)$ integrands, $369$ are of basic type split as $48$ with $p=0$ and $321$ with $p=1$. Here we present some examples with $p=2,3,4$,
\allowdisplaybreaks[1]
\begin{align}
&  {\cal I}^{(3,9)}(\Sigma_{369}) = \frac{\Delta _{2 3 9} \Delta _{3 5 9}}{\Delta _{1 2 3} \Delta _{1 2 9} \Delta _{1 4 5} \Delta _{2 3 4} \Delta _{2 7 8} \Delta _{3 5 6} \Delta _{3 7 9} \Delta
   _{3 8 9} \Delta _{4 6 7} \Delta _{5 6 9} \Delta _{5 8 9}},
\\   
&  {\cal I}^{(3,9)}(\Sigma_{1067}) = \frac{\Delta _{1 6 8} \Delta _{3 4 9}-\Delta _{1 3 4} \Delta _{6 8 9}}{\Delta _{1 2 3} \Delta _{1 2 9} \Delta _{1 6 7} \Delta _{1 8 9} \Delta _{2 3 4} \Delta
   _{3 4 5} \Delta _{3 7 8} \Delta _{4 5 6} \Delta _{4 7 9} \Delta _{5 6 8} \Delta _{6 8 9}},
\\
& {\cal I}^{(3,9)}(\Sigma_{1068}) = \frac{\Delta _{2 3 5} \Delta _{2 3 9} \Delta _{3 8 9}}{\Delta _{1 2 3} \Delta _{1 4 5} \Delta _{1 8 9} \Delta _{2 3 4} \Delta _{2 5 6} \Delta _{2 7 9} \Delta
   _{2 8 9} \Delta _{3 5 8} \Delta _{3 5 9} \Delta _{3 6 9} \Delta _{3 7 8} \Delta _{4 6 7}},
\\
& {\cal I}^{(3,9)}(\Sigma_{1890}) =  \frac{\Delta _{2 3 9} \Delta _{3 4 8} \Delta _{3 8 9} \Delta _{7 8 9}}{\Delta _{1 2 3} \Delta _{1 5 6} \Delta _{1 8 9} \Delta _{2 3 4} \Delta _{2 7 9} \Delta
   _{2 8 9} \Delta _{3 4 9} \Delta _{3 5 9} \Delta _{3 6 8} \Delta _{3 7 8} \Delta _{4 5 8} \Delta _{4 7 8} \Delta _{6 7 9}}.
\end{align}
\allowdisplaybreaks[0]
The following is an example of an integrand with one of the four undetermined coefficients, denoted $y$,
\be  
{\cal I}^{(3,9)}(\Sigma_{513}) = \frac{\Delta _{1 5 7} \Delta _{2 3 5}+y \  \Delta _{1 2 5} \Delta _{3 5 7}}{\Delta _{1 2 3} \Delta _{1 2 9} \Delta _{1 4 5} \Delta _{1 6 7} \Delta _{2 5 6} \Delta
   _{2 7 8} \Delta _{3 4 6} \Delta _{3 5 7} \Delta _{3 5 8} \Delta _{4 8 9} \Delta _{5 7 9}}.
\ee  

It is remarkable that the system of $4381$ types of integrands (for each type one must add integrands obtained by permutations of labels) can be determined up to only four coefficients. Once again, the list of representatives for all $4381$ types is provided in the ancillary file.  

Here we make some comments on non-realizable pseudo-GCOs defined below Definition \ref{generalpseudoGCO}.
In $(3,9)$, there is exactly one type of non-realizable pseudo-GCOs. A representative is given by 
\begin{align}\label{pseudoGCO39}
\big(
(24738695),(18536749),(14857296),& (17385629),(18273469),(15472938),
\nonumber\\
&(14853629),(14372596),(15247638)
\big)\,.
\end{align}
Combinatorially, a (3,9) GCO can be connected to a non-realizable pseudo-GCO via a triangle flip, but this cannot be realized geometrically. See more details in \cref{secnonreal39}. The reader might wonder what happens to the algorithm once a triangle flip of a target GCO leads to a pseudo-GCO.  We do not have a definite answer to this question and in our implementation, we have chosen to treat pseudo-GCOs as unknown. If we were to impose that the residue of an integrand vanishes when its triangle flip leads to a  non-realizable pseudo-GCO we would have been able to fix one of the four coefficients.
Since we expect that a direct computation of partial amplitudes using generalized Feynman diagrams will fix the value of all four coefficients,  we leave its value undetermined for now.

Let us comment on  another yet new phenomenon 
starting at $(3,9)$.
In $(3,7)$ we found two distinct GCOs with the same triangles. Those two were not related by a triangle flip \eqref{eq37two}. In $(3,9)$, there are two GCOs that are connected by a triangle flip and share the same set of triangles! We put some examples in \cref{39newphe}.

\subsubsection{$(4,8)$ Integrands \label{48integrands}}

In 
$(4,8)$, there are 2604 types of GCOs and 24 types of non-realizable pseudo-GCOs. 
Among the $2604$ types of $(4,8)$ integrands,  $279$ are of basic type split as $41$ with $p=0$ and $238$ with $p=1$.
Based on them, we further fix 1711 integrands completely and find expressions for 479 other integrands in terms of 24 parameters. This leaves $135$ integrands which we have yet not computed. 

Here we present some examples with $p=2,3,4$,
\allowdisplaybreaks[1]
\begin{align} 
{\cal I}^{(4,8)}(\Sigma_{279}) =&\frac{\Delta _{2678}^2}{\Delta _{1267} \Delta
   _{1278} \Delta _{1348} \Delta _{1456} \Delta _{2368}
   \Delta _{2378} \Delta _{2458} \Delta _{2567} \Delta
   _{3467} \Delta _{5678}},
   \\
 {\cal I}^{(4,8)}(\Sigma_{769})= & \frac{\Delta _{2348}^2 \Delta
   _{3678}}{\Delta _{1234} \Delta _{1238} \Delta _{1267}
   \Delta _{1458} \Delta _{2345} \Delta _{2378} \Delta
   _{2468} \Delta _{3468} \Delta _{3478} \Delta _{3567}
   \Delta _{5678}},
   \\
 {\cal I}^{(4,8)}(\Sigma_{1252}) =& \frac{\Delta _{2347} \Delta _{2348}
   \Delta _{2678} \Delta _{4678}}{\Delta _{1234} \Delta
   _{1237} \Delta _{1458} \Delta _{1678} \Delta _{2345}
   \Delta _{2368} \Delta _{2468} \Delta _{2478} \Delta
   _{2567} \Delta _{3467} \Delta _{3478} \Delta
   _{5678}} ,
\\ 
{\cal I}^{(4,8)}(\Sigma_{395}) =& \frac{\Delta _{2456} \Delta
   _{5678}+z \Delta _{2568} \Delta _{4567}}{\Delta _{1234} \Delta _{1267} \Delta _{1378} \Delta
   _{1456} \Delta _{2468} \Delta _{2568} \Delta _{2578} \Delta
   _{3458} \Delta _{3567} \Delta _{4567}} .
\end{align}
\allowdisplaybreaks[0]

Note that the last example
contains one of the 24 undetermined coefficients, which we denoted as $z$.
All $2469=279+1711+479$ integrands we obtained are included in the ancillary file. They are dual to themselves up to a relabelling, which is a strong consistency check of our algorithm.  We expect that the remaining 135 types of integrands can be worked out similarly but would require either large computing resources or improvements to the algorithm and we leave them for future research.

Note that 11 of 
 24 non-realizable pseudo-GCOs are connected to the 2469 integrands we obtained by combinatorial tetrahedron flips. Similar to the $(3,9)$ case, if we were to impose that the residues of integrands vanish when their tetrahedron flips lead to  
non-realizable pseudo-GCOs, we would have been able to fix nine of the 24 parameters. 

\section{Higher $k$ Irreducible Decoupling Identities \label{sec4}}

Decoupling identities in physics refer to linear relations among partial amplitudes which are derived by using that color-dressed amplitude vanish whenever the color of a particle commutes with that of all the others. The vanishing of the color-dressed amplitude is manifest from the Lagrangian viewpoint as interaction vertices are proportional to the structure constants of the color group. From the viewpoint of partial amplitudes, it is not obvious that they must satisfy identities. One of the advantages of the Witten-Roiban-Spradlin-Volovich (Witten-RSV) formulation \cite{Witten:2003nn,Roiban:2004yf} of ${\cal N}=4$ SYM amplitudes and of the CHY formulation of Yang-Mills and biadjoint scalar amplitudes is that decoupling identities can be understood directly from partial amplitudes. In this section, we start with a review of how decoupling identities are understood in the CHY formulation, i.e., when $k=2$, and then move on to the corresponding generalization for $k=3$. 

It turns out that $k=2$ decoupling identities are always irreducible, i.e., no proper subset of terms participating in a certain identity adds up to zero. In \cite{Cachazo:2022pnx}, we found that a naive extension of the $k=2$ procedure to higher $k$ leads to decoupling identities but they are reducible. Here we find a natural explanation for it and propose an algorithm for $k=3$ to generate various of irreducible decoupling identities. Moreover, in the process, we identify the analog of the $U(1)$ decoupling identities for higher $k$ in a procedure we call double extension, with section \ref{ref4B} devoted to them.

\subsection{$k=2$ Decoupling Identities
\label{sec4d1}}

In the Witten-RSV and CHY formulations of partial amplitudes, the part of the integrand that carries the color information is the Parke-Taylor function,
\be 
{\rm PT}(1,2,\ldots, n) := \frac{1}{\Delta_{12}\Delta_{23}\cdots \Delta_{n1}}.
\ee 
In our terminology, this is the integrand for type 0 color orderings. For $k=2$ this is the only type of color orderings present. 

Decoupling particle $n$ means that the position of the $n^{\rm th}$ label in an ordering is irrelevant and therefore one should group all partial amplitudes according to the $(2,n-1)$ ordering they map to under the projection, $\pi_n$, that is deleting $n$ from the ordering. This means that given a $(2,n-1)$ ordering, say $\tilde\sigma := (1,2,\ldots ,n-1)$, then  
\be\label{k2DI}
\sum_{\pi_n(\sigma)=\tilde\sigma} {\rm PT}(\sigma ) = 0.
\ee
Here the sum is over all $(2,n)$ orderings $\sigma$ such that when label $n$ is deleted they become $\tilde\sigma$.

The proof of \eqref{k2DI} is very simple. Consider the LHS of  \eqref{k2DI} as a rational function of the position of the $n^{\rm th}$ point in $\mathbb{CP}^1$. Denote its inhomogeneous coordinate by $z$. Clearly, the function is ${\cal O}(1/z^2)$ as $z\to \infty$ and can only have simple poles which are located where $\Delta_{in}=0$ for any $i$. By inspection, one can compute the residue at $\Delta_{in}=0$ noting that only two of the $n-1$ terms in \eqref{k2DI} contain the pole. Moreover, the residue becomes proportional to that of 
\be 
\frac{\Delta_{i-1,i}}{\Delta_{i-1,n}\Delta_{n,i}} +\frac{\Delta_{i,i+1}}{\Delta_{i,n}\Delta_{n,i+1}}.
\ee 
which clearly vanishes since 
\be
\Delta_{i-1,i}\Delta_{n,i+1}+\Delta_{i,i+1}\Delta_{i-1,n} = \Delta_{n,i}\Delta_{i-1,i+1}.
\ee 

This cancellation has a beautiful combinatorial interpretation. If we represent the $(2,n-1)$ color ordering as a circle with $n-1$ points on it, then placing point $n$ in any of the $n-1$ intervals corresponds to a point in $X(2,n)$. Each interval corresponds to a color ordering and crossing from one interval to another is an ``interval'' flip of the color ordering. The fact that the residues of the corresponding integrands agree up to a sign is built in the construction. Adding up overall integrands as we move along the circle guarantees that the resulting object has no poles. 

Moreover, each identity involving $n-1$ partial amplitudes is irreducible, i.e., no proper subset of the $n-1$ partial amplitudes can be added with constant coefficients to get zero. This fact is also clear from the picture using the circle.

\subsection{$k>2$ Decoupling Identities}

In \cite{Cachazo:2022pnx}, we generalized the definition of decoupling a particle to higher $k$.
\begin{defn}\label{k3decoupling}
For the $(k,n)$ color-dressed amplitude, ${\cal M}_n^{(k)}$, the operation that identifies (up to a sign) any two $(k,n)$-color orderings, $\Sigma$ and $\tilde\Sigma$, if their $i^{\rm th}$ projections are the same, i.e., $\pi_i (\Sigma ) = \pi_i (\tilde\Sigma)$, is called decoupling the $i^{\rm th}$-particle. 
\end{defn} 

Decoupling identities are obtained if one requires that the color-dressed amplitude ${\cal M}_n^{(k)}$ vanishes when a particle is decoupled.

This definition heavily relies on the projection operator defined in \eqref{projectionkkkk}, $\pi_i:CO_{k,n}\to CO_{k,n-1}$.

When $k>2$ one can straightforwardly define the notion of a decoupling set of GCOs as done in \cite{Cachazo:2022pnx},
but unlike the $k=2$ case, this does not lead to irreducible identities. The reason is that when $k>2$ one can have lower $k$ cancellations in the set! 

In this section, we present an algorithm for finding irreducible decoupling identities for $k=3$. The algorithm is purely combinatorial and only uses two of the basic operations on generalized color orderings introduced in \cite{Cachazo:2022pnx}: one is the projection $\pi_i$ and the other is the simplex flip reviewed at the end of \cref{sec3d1}.

Projecting, say particle $n$, induces a partition of the set of all $(k,n)$ GCOs in terms of the pre-image of $\pi_n$. In other words, given a $(k,n-1)$ GCO $\Sigma$, then $\pi^{-1}_n(\Sigma)\subset CO_{k,n}$ is called a {\it decoupling set}. For example, starting with $(3,6)$ one finds that the $372$ GCOs split into $12$ decoupling sets each with $31$ GCOs. The $12$ sets are labeled by the $12$ $(3,5)$ GCOs. When $(3,n>6)$, decoupling sets can have a different number of elements.

It turns out that proper subsets of $\pi^{-1}_n(\Sigma)$ satisfy identities when $n>5$ and $k=3$. In fact, for every subset that generates an identity, its complement also generates one. We find that there are two types of irreducible identities distinguished by whether their complement is also irreducible or not.  

\subsection{Constructing Irreducible Decoupling Identities\label{sec43}}

Let us introduce the algorithm for finding irreducible identities in the $(3,n)$ case and then provide examples. 

An irreducible identity can be characterized by a given $(3,n)$ GCO, $\Sigma$, and any pair of its triangles that share at least one label. We identify such a label as the one being decoupled.

\vspace{2mm}
\noindent \paragraph{Algorithm II: Building GCO Sets for Irreducible Decoupling Identities}  ~

{\it Input}: A $(3,n)$ GCO $\Sigma$, a particle to decouple, say $n$, and a choice of two distinct triangles of $\Sigma$ that contain $n$, say $\{i,j,n\}$ and $\{p,q,n\}$.

{\it Output}: A set of $(3,n)$ GCOs belonging to the $n$ decoupling set that contains $\Sigma$ and which combine to produce an irredubible decoupling identity.

\begin{enumerate}
    \item Construct a list $L$ which starts with $\Sigma$ as its only element.
    \item For every GCO in $L$, determine all its triangles that contain label $n$ and discard any that coincide with either $\{i,j,n\}$ or $\{p,q,n\}$. 
    \item Apply all possible triangle flips for the triangles left over in step 2 and only add new GCOs to $L$
    \footnote{There are subtleties if a triangle flip leads to a non-realizable pseudo-GCO. We postpone the discussion of these cases to \cref{irre39sec}.}.
    \item Repeat steps 2-3 until no new GCOs are produced.
    \item Return list $L$ with the GCOs in the identity.
\end{enumerate}

In the following, we call $n$ the decoupling particle, $\Sigma$ the original GCO that generates the identity, and $\{i,j\}, \{p,q\}$ the ``walls" as they constitute the forbidden triangles that cannot be flipped (or crossed).

\begin{prop}\label{irreProof}
   Any set of GCOs generated by using Algorithm II gives rise to an irreducible decoupling identity.   
\end{prop}

The proof is similar to that for $k=2$ explained in \cref{sec4d1} and uses a residue theorem. The only new feature is that there are cancellations inside the region and on the boundary. The proof is presented in \cref{appproof}.  

\subsection{Geometric Intuition and Examples}

Let us illustrate the use of algorithm II in several $(3,n)$ cases. When $n<7$ it is easy to develop a geometric picture for the decoupling identities since the real model of $X(3,n)$ captures all the relevant structures. 

The idea is to consider a representation of the fibration $X(3,n)\to X(3,n-1)$. One starts by selecting a particular point in $X(3,n-1)$. The choice is equivalent to that of a decoupling set\footnote{More formally, we first fix a rank 3 oriented uniform matroid on a set of size $n-1$.  The second step is to calculate all one-element extensions, see \cite{bjorner1999oriented} for details.}. Having $n-1$ points on $\mathbb{RP}^2$ one can draw all $\binom{n-1}{2}$ lines defined by any pair of such points. So, a line is labeled by the pair that it contains, e.g. $L_{24}$ is the line that passes through points $2$ and $4$. The $\binom{n-1}{2}$ lines intersect partitioning $\mathbb{RP}^2$ into polygons known as {\it chambers}.

Before continuing, let us explain the connection of this picture to that of the arrangement of lines used to describe GCOs. As it turns out, the two pictures are dual to each other. Lines in one picture correspond to points in the other and vice-versa. 

Now, the fibration of $X(3,n)$ over $X(3,n-1)$ is realized by placing point $n$ on the figure with the other $n-1$ points fixed. Point $n$ must be placed on one of the chambers and once it is in there, it cannot be moved to another chamber without making it become collinear with two others. For example, if point $n$ is placed on a chamber bounded by lines $L_{ab},L_{cd},L_{ef}$, then in order to make it move to another chamber one must have that at least one of $\Delta_{abn}$, $\Delta_{cdn}$, or $\Delta_{efn}$ vanishes.

This means that a chamber is dual to a particular GCO in a decoupling set and crossing into another chamber is equivalent to performing a triangle flip on the GCO.

Now all the elements in algorithm II come to life:

\begin{itemize}
    \item Selecting a label to decouple, say $n$, is equivalent to selecting a fibration, $X(3,n)\to X(3,n-1)$.
    \item Selecting a GCO, say $\Sigma$, corresponds to selecting a point in $X(3,n-1)$, i.e. a configuration of $n-1$ points in $\mathbb{RP}^2$, via the projection $\pi_n(\Sigma)$ and the duality between lines and points. Moreover, $\Sigma$ also picks a particular chamber in the fiber.
    \item The choice of two triangles of $\Sigma$, $\{ i,j,n \}$ and $\{ p,q,n \}$ corresponds to selecting the two boundaries of the chamber corresponding to $\Sigma$ and defined by the lines $L_{ij}$ and $L_{pq}$. 
    \item Performing triangle moves which are not associated with either $\{ i,j,n \}$ or $\{ p,q,n \}$ means that we cover chambers in $\mathbb{RP}^2$ which are within lines $L_{ij}$ and $L_{pq}$. Note that two lines in $\mathbb{RP}^2$ partition the space into only two regions. We are choosing chambers in the region that already contains that associated with $\Sigma$.
\end{itemize}

Let us mention some of the consequences of the geometric perspective. Even though we started with a GCO, the picture reveals that the decoupling identity is more naturally associated with the pair of lines $L_{ij}$ and $L_{pq}$ and to the choice of ``inside'' and ``outside''. 

As we will see in the examples, if $L_{ij}\cap L_{pq}$ is not one of the $n-1$ special points, the decoupling identities are always irreducible as one can show that there exists a chamber for which algorithm II leads to them. If $L_{ij}\cap L_{pq}$ is one of the special points,  the identity is not necessarily irreducible. In fact, there are $n-2$ lines passing through any given special point. If a circle is drawn centered around the point, an identity is irreducible if the two lines intersect the circle at points that are adjacent.

\subsubsection{Types of Irreducible Decoupling Identities and $(3,5)$ Examples
\label{sec4d4d1}}

We mentioned that there are two types of irreducible identities distinguished by whether their complement in a decoupling set is also an irreducible identity or not. Let us call the former a {\it partitioning} identity and the latter a non-partitioning one. In section \ref{ref4B}, we show that the non-partitioning ones are, in fact, the analog of the $U(1)$ decoupling identities in $k=2$ and that they can be thought of as ``fundamental", in the sense that, conjecturally, all other identities are linear combinations of them.

There is a simple combinatorial characterization to distinguish them. An irreducible identity constructed using a GCO $\Sigma$ and triangles $\{i,j,n\}$ and $\{p,q,n\}$ is partitioning if and only if $\{i,j\}\cap \{p,q\} =\emptyset$. 

Let us illustrate this concept with two $(3,5)$ identities. 
\begin{figure}[h!]
	\centering
   \includegraphics
   [width=1\linewidth]
{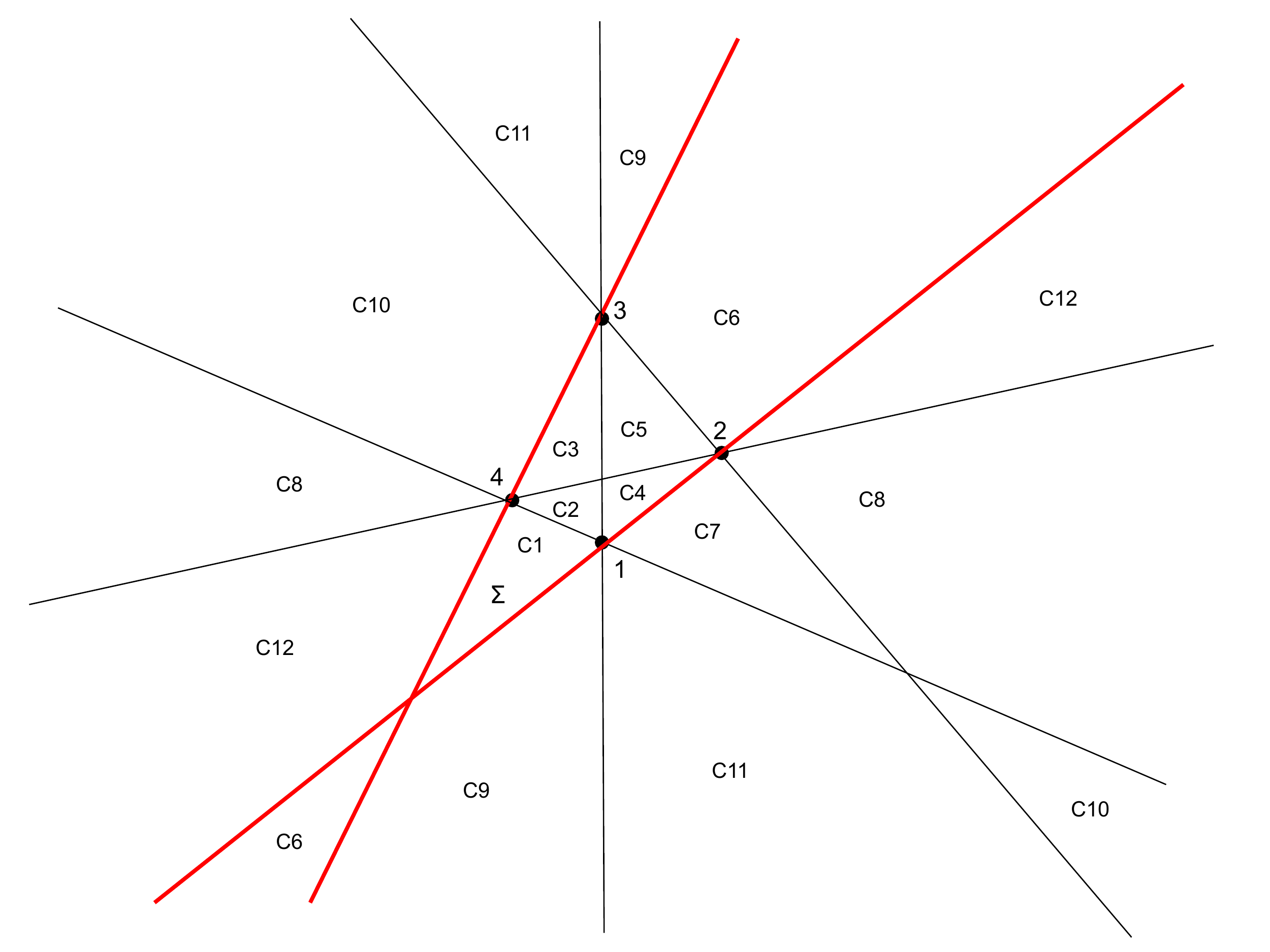}
 \caption{Partitioning identity obtained by considering the GCO $\Sigma$ associated with chamber 1 and using triangles $\{3,4,5\}$ and $\{5,1,2\}$. The red lines, $L_{34}$ and $L_{12}$, bound the chambers that participate in the identity. These are $C_1,C_2,C_3,C_4,C_5,C_6$. The remaining six chambers also generate an irreducible identity.   \label{PartiID35} }
\end{figure}

In \cref{PartiID35}, we start with a fibration of $X(3,5)$, thought of as the configuration space of five points on $\mathbb{RP}^2$, over $X(3,4)$. The figure represents the base $X(3,4)$ and shows the lines where if point $5$ were located, then it would be collinear with two other points. The regions bounded by lines are chambers and there are $12$ in this case. Each chamber is associated with a GCO and $\Sigma$ is chosen to be in chamber $C_1$.

\begin{figure}[h!]
	\centering
   \includegraphics
   [width=1\linewidth]
{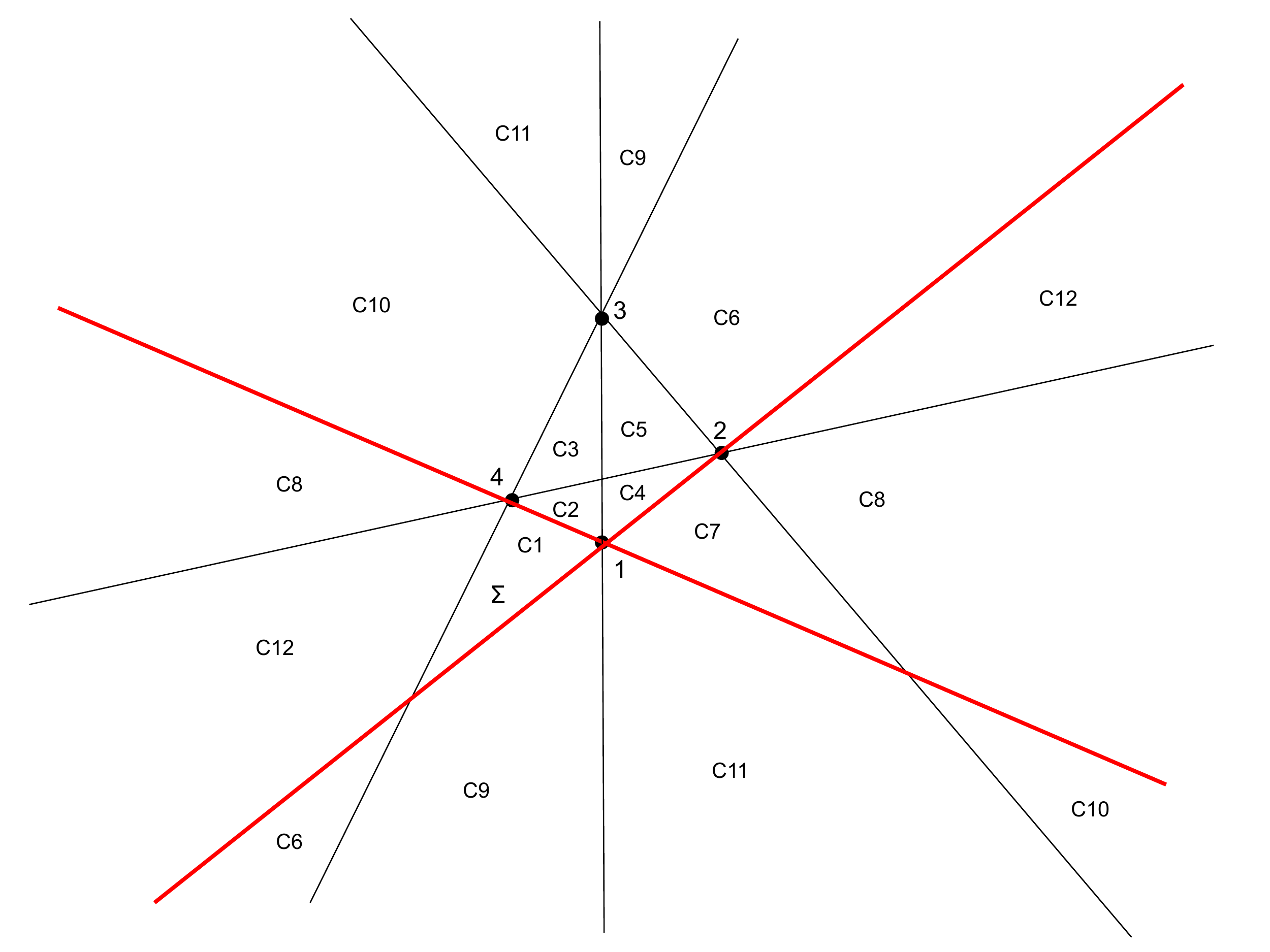}
 \caption{Non-partitioning identity obtained by considering the GCO $\Sigma$ associated with chamber 1 and using triangles $\{4,5,1\}$ and $\{5,1,2\}$. The red lines, $L_{14}$ and $L_{12}$, bound the chambers that participate in the identity. These are $C_1,C_7,C_8,C_{12}$. The remaining eight chambers do not produce an irreducible identity.   \label{NonPartiID35} }
\end{figure}

There are three triangles, i.e., boundaries that can be chosen to create the identity. These are $\{3,4,5\}$, $\{4,5,1\}$, and $\{5,1,2 \}$. In \cref{PartiID35} we chose $\{3,4,5\}$ and $\{5,1,2\}$. Since $\{3,4\}\cap \{1,2\}=\emptyset$, we expect a partitioning identity. The algorithm for finding the corresponding identity boils down to selecting the set of chambers bounded by the two red lines. Note that there are only two regions divided by two lines in $\mathbb{RP}^2$. Here we consider the region which contains $\Sigma$ to get the identity
\be\label{shuffle1}
\{ C_1,C_2,C_3,C_4,C_5,C_6 \}.
\ee
Since this is a partitioning identity, we also have that 
\be\label{shuffle2}
\{ C_7,C_8,C_9,C_{10},C_{11},C_{12} \}.
\ee
give rise to an irreducible identity. 

Here we list the 
 integrands associated with the twelve chambers:
\begin{align}
\label{12PT5pt}
 &\text{C}_1\,\, \text{PT}^{\text{(3)}}(1,2,3,4,5); ~~
 \text{C}_2 \, \,\text{PT}^{\text{(3)}}(1,3,2,4,5); ~~\,\,
 \text{C}_3\, \,\text{PT}^{\text{(3)}}(1,2,4,5,3); ~~\,\,
 \text{C}_4 \, \,\text{PT}^{\text{(3)}}(1,3,4,2,5) ; \nonumber \\
& \text{C}_5 \, \,\text{PT}^{\text{(3)}}(1,3,5,2,4) ; ~~
 \text{C}_6 \, \,\text{PT}^{\text{(3)}}(1,2,5,3,4) ; ~~\,\,
 \text{C}_7 \, \,\text{PT}^{\text{(3)}}(1,4,3,2,5) ; ~~\,\,
 \text{C}_8 \, \,\text{PT}^{\text{(3)}}(1,3,2,5,4) ; \nonumber \\
 &\text{C}_9 \, \,\text{PT}^{\text{(3)}}(1,2,4,3,5) ; ~~
 \text{C}_{10} \, \,\text{PT}^{\text{(3)}}(1,2,3,5,4) ; ~~
 \text{C}_{11} \, \,\text{PT}^{\text{(3)}}(1,4,2,3,5) ; ~
 \text{C}_{12} \, \,\text{PT}^{\text{(3)}}(1,2,5,4,3) .
\end{align}
Here $\text{PT}^{\text{(3)}}(1,2,3,4,5)= 1/(\Delta _{123} \Delta _{125} \Delta _{145} \Delta _{234} \Delta _{345}) $ as defined in \eqref{PTform} and it corresponds to the $(3,5)$ GCO $((2345),(1345),(1245),(1235),(1234))$.
%

In  \cref{NonPartiID35}, we choose as triangles for the identity $\{4,5,1\}$ and $\{5,1,2\}$. Now $\{4,1\}\cap \{1,2\}={1}$ and so the corresponding irreducible identity is not partitioning. Once again, reading the labels of the chambers in the same region as $\Sigma$ one finds a 4-term identity
\be 
\label{35-4termidcc}
\{ C_1,C_7,C_8,C_{12} \} .
\ee  
The reader familiar with decoupling identities in the dual $(2,5)$ would recognize this as the standard four-term $U(1)$ decoupling identity,
\begin{align}
\label{35-4termid}
    \text{PT}^{\text{(3)}}(1,2,3,4,5)+\text{PT}^{\text{(3)}}(1,2,5,4,3)-\text{PT}^{\text{(3)}}(1,3,2,5,4)+\text{PT}^{\text{(3)}}(1,4,3,2,5)=0\,,
\end{align}
where label 1 moves around.

At the same time, the six-term decoupling identities found in the partitioning case, are non-standard in the $(2,5)$ framework. 
In this case, one finds that \eqref{shuffle1} corresponds to 
%
\begin{align}
\label{35-6termid}
   & \text{PT}^{\text{(3)}}(1,2,3,4,5)+\text{PT}^{\text{(3)}}(1,2,4,5,3)+\text{PT}^{\text{(3)}}(1,2,5,3,4)-
   \\
   &\text{PT}^{\text{(3)}}(1,3,2,4,5)+\text{PT}^{\text{(3)}}(1,3,4,2,5)+\text{PT}^{\text{(3)}}(1,3,5,2,4)\equiv \text{PT}^{\text{(3)}}(\{1,2\}\shuffle\{3,4\},5) =0\,,
   \nonumber
\end{align}
while \eqref{shuffle2} corresponds to $\text{PT}^{\text{(3)}}(\{1,2\}\shuffle\{4,3\},5) =0$.

\setlength{\abovecaptionskip}{-5pt}
\begin{figure}[h!]
	\centering
   \includegraphics
   [width=1\linewidth]
{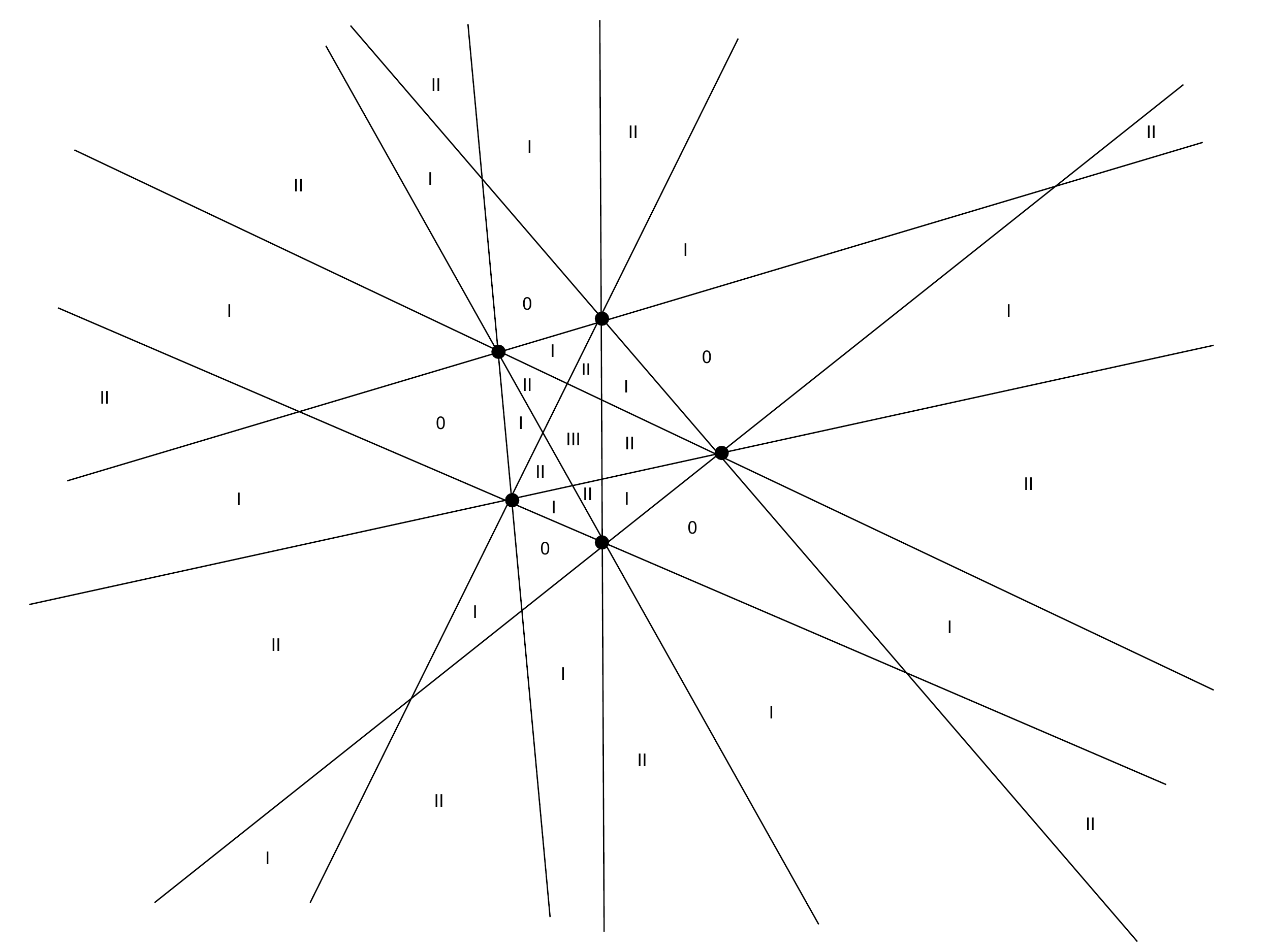}
 \caption{ A point in $X(3,5)$. The five points on $\mathbb{RP}^2$ are represented by the solid circles. Also shown are the $10$ lines containing pairs of such points. The $10$ lines intersect partitioning $\mathbb{RP}^2$ into $31$ regions (or chambers). Once point six is placed on a chamber, it cannot visit another chamber unless it becomes collinear with two other points, i.e. it must cross a line. Each chamber corresponds to a GCO and the corresponding types are shown. There are five type 0, fifteen type I, ten type II, and a single type III GCO in this decoupling set. \label{Types36} }
\end{figure}
\setlength{\abovecaptionskip}{0pt}

\subsubsection{$(3,6)$ Irreducible Decoupling  Identities}


Let us start again by fibering $X(3,6)$ over $X(3,5)$ with fiber in the position of point $6$. This means that the base is given by the configuration space of five points. In \cref{Types36}, we show a point in $X(3,5)$. We also draw all lines connecting any pair of the five points as they are the boundaries of the connected chambers where point 6 can live. There are $31$ chambers and we have indicated the type of the GCO associated with each of them. Note that there are five type 0 GCOs as expected from the fact that this is a decoupling set and exactly five $(2,6)$ color factors get identified after removing point $6$.

There are a total of $9$ different types of irreducible identities that can be derived by selecting a pair of lines in \cref{Types36}.

There is one 7-term, one 8-term, one 9-term, two 12-term, one 14-term, one 17-term, and two 19-term identities.  

\setlength{\abovecaptionskip}{-5pt}
\begin{figure}[h!]
	\centering
   \includegraphics
   [width=1\linewidth]
{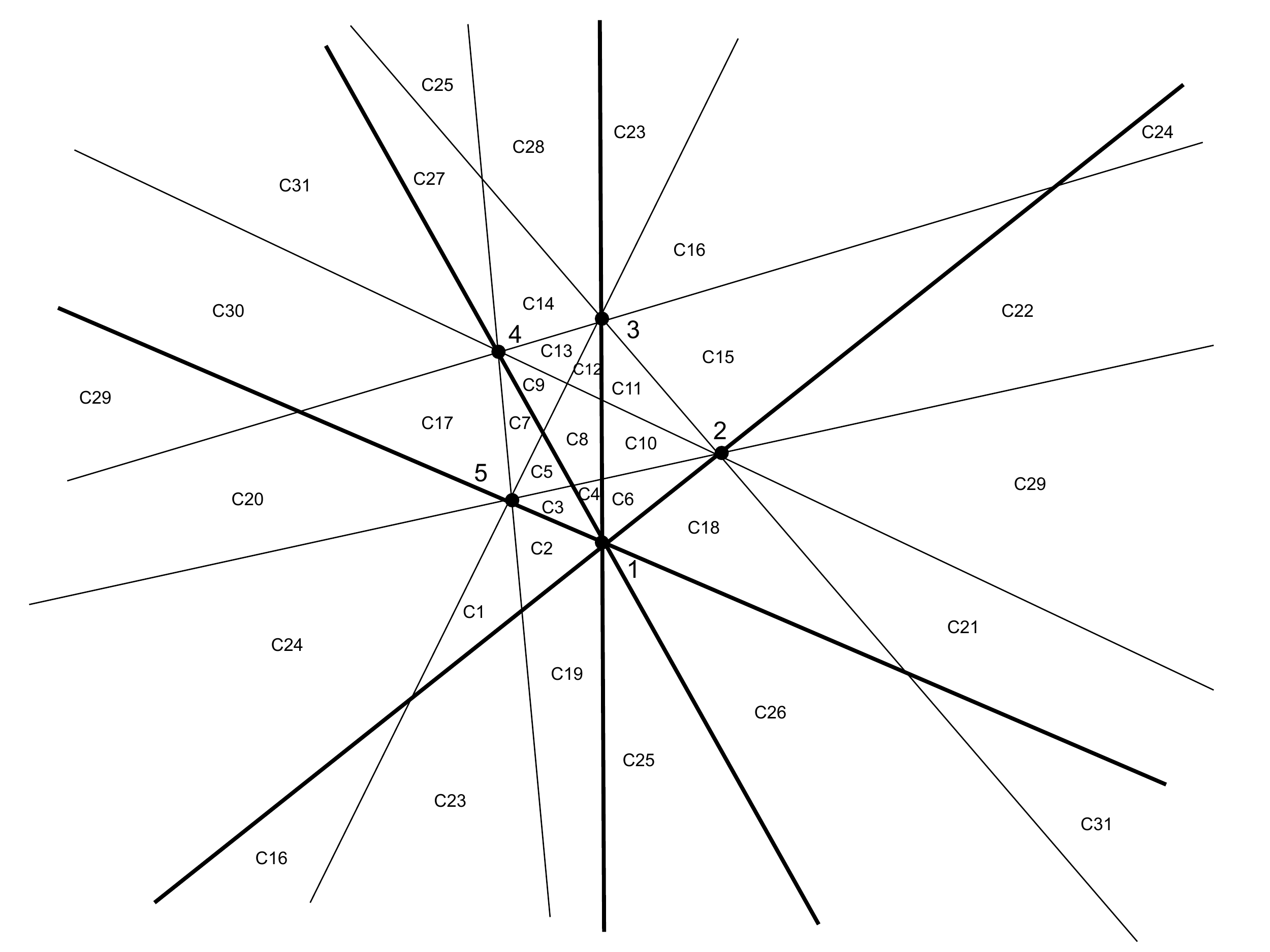}
\caption{Identities obtained by considering all lines that contain point $1$, i.e., lines $L_{12}$, $L_{13}$, $L_{14}$ and $L_{15}$. The regions bounded by 
$\{L_{12}, L_{13}\}$ or $\{L_{14}, L_{15}\}$ both contain 7 GCOs, which are related by a relabelling.  The region bounded by 
$\{L_{13}, L_{14}\}$ contains 9 GCOs while that of $\{L_{15}, L_{12}\}$ contains 8 GCOs. By the way, the associated GCOs for the chambers $C_2,C_{19},C_5,C_8$ are the type 0, I, II and III representatives shown in \eqref{type036GCO}-\eqref{typeIII36GCO} respectively whose arrangements of lines are given in \cref{triangles36}.
\label{P36point1Text}}
\end{figure}
\setlength{\abovecaptionskip}{0pt}

In order to explicitly construct the identities we have to label the points so that a pair of lines can be selected. In \cref{P36point1Text}, we show the five points with labels and the four lines that contain point 1. The four lines split $\mathbb{RP}^2$ into four regions. The chambers in each of the four regions define an irreducible identity. The explicit form of the GCOs and integrands associated with each chamber can be found in  \cref{appb}. Here we simply point out that there are one 8-term, one 9-term identity, and two 7-term identities which are related by a relabeling. The seven-term identity bounded by $L_{14},L_{15}$ is explicitly given by
\begin{align}
\label{sevenid}
   & \{ C_3,C_5,C_7,C_{17},C_{30},C_{31},C_{26}  \} \Rightarrow
 \\
&-\frac{1}{\Delta _{123} \Delta _{146} \Delta _{156} \Delta _{234} \Delta _{256} \Delta _{345}}
-\frac{\Delta _{235}}{\Delta _{123} \Delta _{125} \Delta _{146} \Delta _{234} \Delta _{256} \Delta _{345} \Delta _{356}}
 \nonumber\\
&-\frac{1}{\Delta _{123} \Delta _{125} \Delta _{146} \Delta _{234} \Delta _{356} \Delta _{456}}
+
\frac{1}{\Delta _{123} \Delta _{125} \Delta _{156} \Delta _{234} \Delta _{346} \Delta _{456}}
 \nonumber \\
&+
\frac{1}{\Delta _{123} \Delta _{125} \Delta _{156} \Delta _{246} \Delta _{345} \Delta _{346}}
+
\frac{\Delta _{126}}{\Delta _{123} \Delta _{125} \Delta _{146} \Delta _{156} \Delta _{236} \Delta _{246} \Delta _{345}}
 \nonumber \\
&+\frac{1}{\Delta _{125} \Delta _{146} \Delta _{156} \Delta _{234} \Delta _{236} \Delta _{345}} =0\,.
 \nonumber
\end{align}
We have also included figures of the other kinds of identities in \cref{appc}. 

For the reader's convenience, we also present the 9 kinds of GCO sets for irreducible identities in the ancillary file {\tt irreID_36_9.m}. 
For example, in the ancillary file, the GCO sets for the seven-term identity in \eqref{sevenid} are given by 
\begin{verbatim}
irreID[
{type[1], perm[1, 5, 4, 3, 2, 6]}, {type[2], perm[5, 1, 2, 3, 6, 4]}, 
{type[1], perm[1, 2, 6, 4, 3, 5]}, {type[0], perm[1, 2, 3, 4, 6, 5]}, 
{type[1], perm[2, 1, 4, 6, 5, 3]}, {type[2], perm[1, 2, 3, 4, 5, 6]}, 
{type[1], perm[3, 4, 1, 6, 5, 2]},
{decouple[6], walls[{1, 4}, {1, 5}]}].
\end{verbatim}
 Here ${\tt 
\{type[1], perm[1, 5, 4, 3, 2, 6]\}
}$ means a type I GCO using the seed given in \eqref{typeI36GCO} under a relabelling, $\eqref{typeI36GCO}|_{\{1,2,3,4,5,6\}\to \{1, 5, 4, 3, 2, 6\}}$. 
We have also shown that the decoupling particle is 6 and the walls are $\{1, 4\}, \{1, 5\}$.  
The first GCO in ${\tt irreID[\bullet]}$ is the original GCO in Algorithm II.

It is clear to see there is one type 0, four type I, and two type II GCOs in this identity.

\subsubsection{$(3,7)$,$(3,8)$ and $(3,9)$ Irreducible 
 Decoupling Identities \label{irre39sec}}

As $n$ increases the number of different kinds of identities grows very fast. We have implemented algorithm II and made an exhaustive search for $n=7,8,9$. They have $78$, $1432$, and $52\, 444$ classes of GCO sets for irreducible decoupling identities respectively, among which 18, 234, and 6193
classes respectively are for non-partitioning identities. 

We store all of the classes of  $(3,7)$ and  $(3,8)$ GCO sets in the ancillary file while we just put the  
6193
classes of (3,9) GCO sets for non-partitioning identities there to save the size of the file. As will be argued in \cref{ref4B}, their combinations are supposed to 
express all other irreducible decoupling identities. So they should be enough.

Here is an example of $(3,7)$ with 11 GCOs,
\begin{verbatim}
irreID[
{type[3], perm[1, 2, 3, 4, 5, 6, 7]}, {type[10], perm[3, 2, 1, 7, 6, 4, 5]}, 
{type[10], perm[1, 2, 3, 4, 6, 7, 5]}, {type[3], perm[3, 2, 1, 7, 5, 6, 4]}, 
{type[1], perm[4, 6, 5, 7, 3, 2, 1]}, {type[5], perm[1, 4, 6, 5, 7, 2, 3]}, 
{type[8], perm[1, 4, 6, 5, 7, 2, 3]}, {type[2], perm[3, 2, 1, 5, 7, 6, 4]}, 
{type[10], perm[2, 1, 5, 3, 4, 7, 6]}, {type[10], perm[6, 4, 5, 7, 3, 1, 2]}, 
{type[3], perm[1, 2, 3, 5, 4, 6, 7]}, 
{decouple[5], walls[{4, 6}, {4, 7}]}].
\end{verbatim}
Remind that  ${\tt \{type[\bullet], perm[\bullet]\} }$
means  GCOs whose seeds are given in \cref{37gco}. The integrands associated with the seeds are given in \eqref{37integrand0}-\eqref{37integrand10},  and their relabelling participates in the irreducible decoupling identity. 

Note that the two walls share a label of $4$. Hence this GCO set corresponds to a non-partitioning irreducible decoupling identity.

The GCO sets for irreducible decoupling identities have at least 11, 16, and 23 GCOs for $(3,7), (3,8)$ and $(3,9)$ respectively.

 Starting at $(3,7)$, there are different kinds of GCO sets for irreducible decoupling identities that have the same set of types of GCOs. For example, the $47^{\rm th}$ and $48^{\rm th}$ GCO sets in the ancillary file {\tt irreID_37_78.m} both have 
four type I GCOs, four type II GCOs, etc.,
\begin{align}
\{&\{{\rm type}[1],4\},\{{\rm type}[2],4\},
\{{\rm type}[3],6\},\{{\rm type}[4],2\},
\{{\rm type}[5],4\},
\nonumber
\\&
\{{\rm type}[6],2\},\{{\rm type}[8],4\},\{{\rm type}[9],2\},\{{\rm type}[10],8\}\}
\end{align}
but the two GCO sets are not related by a relabelling.

In $(3,9)$, there are GCOs sets for irreducible identities whose integrands contain some of the four free parameters. However, their identities hold no matter what the free parameters are.

Note that the irreducible decoupling identities express one integrand by others. If other integrands are already known, we get one integrand almost for free. One can express the unknown integrand as a linear combination of others with few undetermined signs first but these signs can be further fixed by matching the residues of this unknown integrand with its neighbors using \eqref{secnodutilization}.  This also provides a way to get new integrands from known integrands. In particular, it produces the last integrand of $(3,9)$ with 21 triangles 
 mentioned in \cref{39integrand}.

As we already know, there is a type of non-realizable pseudo-GCOs in $(3,9)$.  So there might be a problem when we apply Algorithm II as a combinatorial triangle flip of a GCO may lead to a non-realizable pseudo-GCO. In this case, we have to relax the requirement in Algorithm II a little bit and also allow us to put pseudo-GCOs in the set. This way, we 
may get a set of pseudo-GCOs including non-realizable ones even if we start with a GCO to generate this set. If we set the free parameter $y_0$ mentioned in \eqref{pseudoGCO39neiinte} as 0 and set the ``integrand'' associated with the non-realizable pseudo-GCOs also as 0, we still get an irreducible decoupling identity in terms of integrands associated with the remaining GCOs in that set.

There are 558 kinds of such sets of pseudo-GCOs and we also put them in the ancillary file.

\subsection{General $(k,n)$ Irreducible Decoupling Identities \label{sec4d5}}

The algorithm presented in this section for generating irreducible $(3,n)$ identities can, in principle, be extended to higher values of $k$. The practical obstacle is deciding which sets of simplices should not be used for simplex flips. The naive generalization would be to choose any set of $(k-1)$ simplices of the original GCO that contain the decoupling particle and use them to act as walls to bound a region in ${\mathbb {RP}}^{k-1}$. This procedure always produces a set of GCOs whose associated integrands constitute an identity but it is not guaranteed to be irreducible due to possible lower $k$ (but larger than 2) reducible identities on the boundary of the region. 
The higher dimensionality of the space makes it hard to develop a simple rule to generate the set of simplices needed to guarantee the irreducibility of the identities as it was done for $k<4$. 
  We leave its systematic study to future work.

Of course, when $k=n-2$, one can develop the dual of Algorithm II and apply it to get identities.
\subsubsection{$(n-2,n)$ irreducible decoupling identities}

Due to the duality with $(2,n)$, we know that there is a single type of $(n-2,n)$ GCOs, i.e. type 0. Let us consider the canonical type 0 GCO, i.e., the one which is a descendant of the canonical color ordering $(1,2,\ldots ,n)$. Choosing $n$ to be the label to decouple, there are $n-2$ simplices that contain $n$. They are given by $\{3,4,\ldots ,n\}, \{4,5,\ldots ,n,1\}$, etc.  Let us introduce a short-hand notation and denote 
\begin{align}
   S_i:=[n]\setminus\{i,i+1\}=\{i+2,i+3,\ldots ,n,1,\ldots ,i-1\}\,, \qquad {\rm for}~ i=1,2,\cdots,n-2.
\end{align}

One can get an irreducible identity by selecting $n-3$ simplices to be the forbidden ones for simplex flips in a higher k analog of algorithm II. This is equivalent to the selection of one simplex to be left out of the list, denoted by ${\hat S}_i$.  
It turns out that such kind of selection leads to the following irreducible decoupling identity,
\begin{align}
\label{degeKK}
 & {\text {Forbidden  simplices :}}~ \{S_1,S_2,\cdots, {\hat S}_i, \cdots, S_{n-2} \}
  \nonumber
  \\
   & \qquad\qquad \Rightarrow
\sum_{\shuffle}  {\rm PT}^{(k=n-2)} ( \{1,2,\cdots,i\} \shuffle \{i+1,i+2,\cdots,n-1\},n) =0\,,
\end{align}
where $\sum_{\shuffle}$ 
 is over the shuffle product  ${\pmb \rho} \shuffle {\pmb \omega}  $ of two ordered sets ${\pmb \rho}$ and ${\pmb \omega}$ (i.e., all the permutations of
${\pmb \rho} \cup {\pmb \omega} $ that preserving the ordering of ${\pmb \rho} $ and $ {\pmb \omega} $ respectively). The higher $k$ Parke-Taylor factors are defined in \eqref{PTform}, which are now dual to the standard ones,  ${\rm PT}^{(n-2)}(\alpha)\sim {\rm PT}^{(2)}(\alpha)$ by sending, e.g., $\Delta_{12\cdots n-2} \to {\tilde \Delta}_{n-1,n}$. 

It is clear to see the choice of ${\hat S}_i$ leads to an ${n-1 \choose i}$-term irreducible identity and  
it is isomorphic to the one obtained by selecting  ${\hat S}_{n-1-i}$ for any $i\in \{1,\cdots,n-2\}$. 

We have seen some examples in \eqref{35-4termid} and \eqref{35-6termid} which correspond to the selection ${\hat S}_1$ and ${\hat S}_2$ respectively and here we generalize them to any $n$.

Eq. \eqref{degeKK} looks very like the Kleiss-Kuijf  (KK) relations 
\footnote{The KK relations were first found for tree-level color ordered YM amplitudes  \cite{Kleiss:1988ne} but were made manifest at the level of integrands in the CHY formula \cite{Cachazo:2013hca}, which are nothing but some algebraic properties of the standard Parke-Taylor factors.} and the latter can indeed help us to understand \eqref{degeKK}. 

Let us present a proof using well-known results in the literature. Using the dual version of \eqref{degeKK} and introducing an auxiliary label $n+1$, we have a standard KK relation,
\begin{align}
\label{standKK}
&\sum_{\shuffle}  {\rm PT}^{(2)} (n+1, \{1,\cdots,i\} \shuffle \{i+1,\cdots,n-1\},n)
\nonumber
\\
&\qquad\qquad
= (-1)^{n-i+1}\,
 {\rm PT}^{(2)} (n+1,1,\cdots,i,n, n-1,\cdots,i+1)
 \,.
\end{align}
Remarkably, \eqref{degeKK} can be derived by just 
taking the residues of both sides of \eqref{standKK} at ${\tilde \Delta}_{n,n+1}=0$. In this sense, \eqref{degeKK} could also be called a degenerated KK relation.

We have verified up to $n=20$ that any selection ${\hat S}_i$ would lead to a particular combination of ${\rm PT}^{(k)}(\cdots)$ given in \eqref{degeKK}, which constitutes a degenerated KK relation. A straightforward combinatorial proof to get this combination for any $n$ seems to be within reach.

\section{Connecting Different Descriptions of GCOs
\label{sec:translation}}

The attentive reader might be puzzled by the fact that CEGM integrals \eqref{fullInt} are defined in the configuration space of $n$ points in ${\mathbb {CP}}^{k-1}$, $X(k,n)$, while generalized color orderings (GCOs) are given in terms of arrangements of $n$ hyperplanes in ${\mathbb {RP}}^{k-1}$ (see \eqref{defineGCO}). 

It turns out that GCOs have several equivalent representations. In this section, we discuss three of them and their relations. The first is the one we have used so far, i.e., as arrangements of $n$ hyperplanes in ${\mathbb {RP}}^{k-1}$. The second one is as a configuration of $n$ points in ${\mathbb {RP}}^{k-1}$, i.e., the real model of the space over which the CEGM integral is defined. The third is as reorientation classes of realizable oriented matroids (or chirotopes) \cite{bjorner1999oriented}. Due to the realizability constraint, we can think of chirotopes simply as a vector containing the signs of all Plucker variables of the $n$ points in ${\mathbb {RP}}^{k-1}$.

\usetikzlibrary{angles,quotes}
 \usetikzlibrary{calc}	 
 
 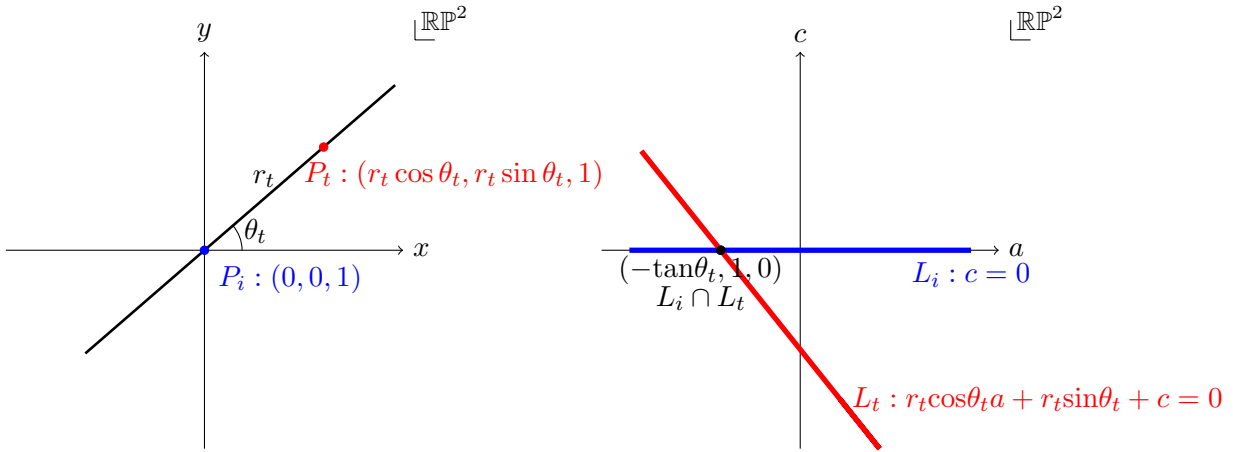
\begin{figure}[h!]

 	\centering
 	\hspace{0.1in}
  \begin{tikzpicture}[scale=.88]
 \begin{scope}[xshift=0cm,yshift=0cm, scale=0.6]
 \draw[->] (-5,0)--(5,0) node[right]{$x$};
 \draw[->] (0,-5)--(0,5) node[above]{$y$};

  \draw (5.3,5.3)--+(0,.5) (5.3,5.3)--+(0.5,0) ;
  \node at (6.05,5.85) {\small ${\mathbb {RP}}^2$};
 \draw (.02,0) coordinate (A) -- (0,0) coordinate (B)
          -- (3,2.6) coordinate (C) 
   pic ["$\theta_t$", draw,
    angle eccentricity=1.45] {angle} ;
     \draw [line width=1pt] ($ (B)! -1/1 ! (C) $) -- ($ (C)! -.6/1 ! (B) $) ;
    \filldraw [blue] (B) circle(3pt)  node[below right=1.5pt]{$P_i: (0,0,1)$};
     \filldraw [red] (C) circle(3pt)  node[below=1pt]{$\qquad \qquad \qquad \qquad \quad  P_t: ( r_t \cos \theta_t, r_t \sin \theta_t,1 )$}; 
 \node at ($ (B)! 1/2 ! (C) +(0,.5)$)  {$r_t$};
  \end{scope}
  \begin{scope}[xshift=9cm,yshift=0cm, scale=0.6]
 \draw[->] (-5,0)--(5,0) node[right]{$a$};
 \draw[->] (0,-5)--(0,5) node[above]{$c$};
   \draw (5.3,5.3)--+(0,.5) (5.3,5.3)--+(0.5,0) ;
  \node at (6.05,5.85) {\small ${\mathbb {RP}}^2$};
 \draw[blue, line width=2pt] (-4.3,0)--(4.3,0) node[below]{$L_i: c=0$};
  \filldraw (-2,0) coordinate (A) circle(3pt) ;
  \coordinate (B) at  (0,-2.5);
  \draw [red, line width=2pt] ($ (A)! -1/1 ! (B) $) -- ($ (B)! -1/1 ! (A) $) -- ($ (B)! -.5/1 ! (A) $) node [right] {$L_t:  r_t {\cos} \theta_t a +  r_t {\sin} \theta_t +c =0  $} ;
  \node at (-2.5,-.5) {$ (-{\tan}\theta_t,1,0)$};
   \filldraw (-2,0) circle(3pt) ;
  \node at (-2.5,-1.2) {$L_i\cap L_t$};
  \end{scope}
 \end{tikzpicture}	
 %
 	\caption{
  Left: A presentation of a configuration of points in $\mathbb{RP}^2$ by setting $z=1$.
 	  Right: A presentation of an arrangement of lines in ${\mathbb {RP}}^{2}$ by setting $b=1$.   A point $(x,y,z)\in \mathbb{RP}^2$ in one projective space is mapped to the line $L=\{ (a,b,c)\in \mathbb{RP}^2:\, ax+by+cz=0  
 	  \} $ in the other projective space and vice versa. More explicitly, the points $P_i, P_t$  are mapped to the lines $L_i, L_t$ in the arrangement of lines and the line $P_{i,t}$ is mapped to the point $L_i \cap L_t = (-{\tan}\theta_t,1,0)$, which only depends on the angle $\theta_t$. As one can imagine, as $\theta_t$ increases from 0 to $\pi$,  the point $L_i \cap L_t $ moves along axis $a$ from the origin  to $(-\infty, 0)$ and then comes back from  the other direction.
    }
 	\label{dualgraph}
 \end{figure}

 \def \sca {1.2}
 
 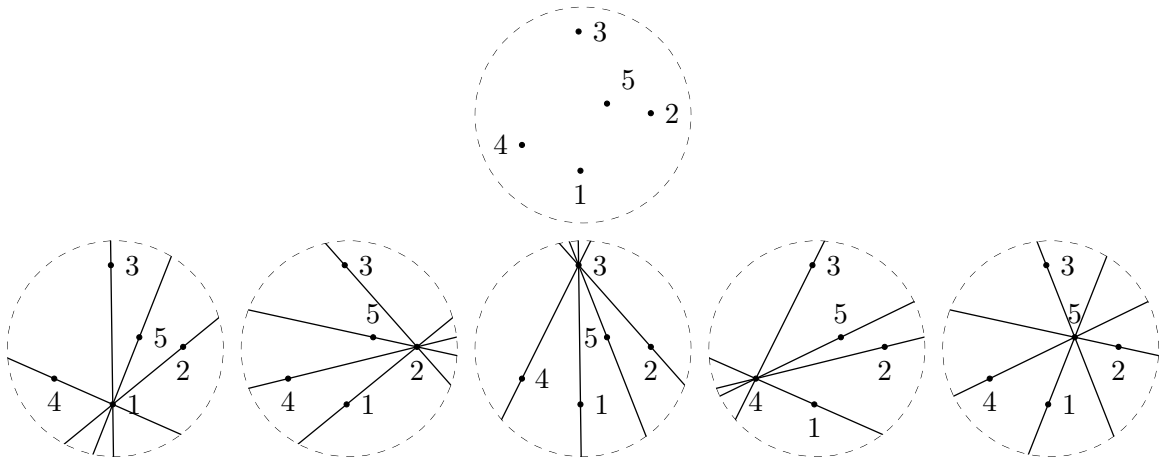
\begin{figure}

 \begin{tikzpicture}[scale=1]
   
   \node at (.4\linewidth, .2\linewidth) {
  \begin{tikzpicture}[scale=\sca ]
  
 \filldraw  (1.22418,0.582795)  coordinate (A1)  circle(.8 pt)  node[below =1.5pt]{$1$};
  
 \filldraw  (2.00383,1.22059)  coordinate (A2)  circle(.8pt)  node[ right=1.5pt]{$2$};
 
  \filldraw  (1.2035,2.12622)  coordinate (A3)  circle(.8pt)  node[ right=1.5pt]{$3$};
 
  \filldraw  (0.576346,0.868565)  coordinate (A4)  circle(.8pt)  node[left=1.5pt]{$4$};

    \filldraw  ($(intersection of  A1--A3 and A2--A4)+(.3,.3)$)  coordinate (A5)  circle(.8pt)  node[ above right =1.5pt]{$5$};
    
 \coordinate   (AA)  at ( $  (A1)+(A2)+(A3)+(A4)  $) ;
  \coordinate   (AA4)  at ($ (0,0)! 1/4 ! (AA) $)  ;
     \clip  (AA4) circle (1.2cm);
      \draw[dashed]  (AA4) circle (1.2cm); 
 
\end{tikzpicture} };

   \node at (0\linewidth, 0\linewidth) {
  \begin{tikzpicture}[scale=\sca]
 \filldraw  (1.22418,0.582795)  coordinate (A1)  circle(.8pt)  node[right =1.5pt]{$1$};
  
 \filldraw  (2.00383,1.22059)  coordinate (A2)  circle(.8pt)  node[ below=1.5pt]{$2$};
 
  \filldraw  (1.2035,2.12622)  coordinate (A3)  circle(.8pt)  node[ right=1.5pt]{$3$};
 
  \filldraw  (0.576346,0.868565)  coordinate (A4)  circle(.8pt)  node[below=1.5pt]{$4$};
  
    \filldraw  ($(intersection of  A1--A3 and A2--A4)+(.3,.3)$)  coordinate (A5)  circle(.8pt)  node[  right =1.5pt]{$5$};
    
        \begin{scope}
    
 \coordinate   (AA)  at ( $  (A1)+(A2)+(A3)+(A4)  $) ;
  \coordinate   (AA4)  at ($ (0,0)! 1/4 ! (AA) $)  ;
     \clip  (AA4) circle (1.2cm);
      \draw[dashed]  (AA4) circle (1.2cm); 
         
  \draw [line width=.5pt] ($ (A2)! -2/1 ! (A1) $) -- ($ (A1)! -2/1 ! (A2) $) ;   
   \draw [line width=.5pt] ($ (A3)! -2/1 ! (A1) $) -- ($ (A1)! -2/1 ! (A3) $) ;   
   \draw [line width=.5pt] ($ (A4)! -2/1 ! (A1) $) -- ($ (A1)! -2/1 ! (A4) $) ;   
   \draw [line width=.5pt] ($ (A5)! -2/1 ! (A1) $) -- ($ (A1)! -2/1 ! (A5) $) ;   
   \end{scope}           
\end{tikzpicture} };

   \node at (0.2\linewidth, 0\linewidth) {
  \begin{tikzpicture}[scale=\sca]
 \filldraw  (1.22418,0.582795)  coordinate (A1)  circle(.8pt)  node[right =1.5pt]{$1$};
  
 \filldraw  (2.00383,1.22059)  coordinate (A2)  circle(.8pt)  node[ below=1.5pt]{$2$};
 
  \filldraw  (1.2035,2.12622)  coordinate (A3)  circle(.8pt)  node[ right=1.5pt]{$3$};
 
  \filldraw  (0.576346,0.868565)  coordinate (A4)  circle(.8pt)  node[below=1.5pt]{$4$};
  
    \filldraw  ($(intersection of  A1--A3 and A2--A4)+(.3,.3)$)  coordinate (A5)  circle(.8pt)  node[  above =1.5pt]{$5$};
    
        \begin{scope}
    
 \coordinate   (AA)  at ( $  (A1)+(A2)+(A3)+(A4)  $) ;
  \coordinate   (AA4)  at ($ (0,0)! 1/4 ! (AA) $)  ;
     \clip  (AA4) circle (1.2cm);
      \draw[dashed]  (AA4) circle (1.2cm); 
      
  \draw [line width=.5pt] ($ (A1)! -4/1 ! (A2) $) -- ($ (A2)! -4/1 ! (A1) $) ;   
   \draw [line width=.5pt] ($ (A3)! -4/1 ! (A2) $) -- ($ (A2)! -4/1 ! (A3) $) ;   
   \draw [line width=.5pt] ($ (A4)! -4/1 ! (A2) $) -- ($ (A2)! -4/1 ! (A4) $) ;   
   \draw [line width=.5pt] ($ (A5)! -4/1 ! (A2) $) -- ($ (A2)! -4/1 ! (A5) $) ;   
              
\end{scope}

\end{tikzpicture} };

   \node at (0.4\linewidth, 0\linewidth) {
  \begin{tikzpicture}[scale=\sca]
  
 \filldraw  (1.22418,0.582795)  coordinate (A1)  circle(.8pt)  node[right =1.5pt]{$1$};
  
 \filldraw  (2.00383,1.22059)  coordinate (A2)  circle(.8pt)  node[ below=1.5pt]{$2$};
 
  \filldraw  (1.2035,2.12622)  coordinate (A3)  circle(.8pt)  node[ right=1.5pt]{$3$};
 
  \filldraw  (0.576346,0.868565)  coordinate (A4)  circle(.8pt)  node[right=1pt]{$4$};
  
    \filldraw  ($(intersection of  A1--A3 and A2--A4)+(.3,.3)$)  coordinate (A5)  circle(.8pt)  node[   left=-.5pt]{$5$};
    
        \begin{scope}

 \coordinate   (AA)  at ( $  (A1)+(A2)+(A3)+(A4)  $) ;
  \coordinate   (AA4)  at ($ (0,0)! 1/4 ! (AA) $)  ;
     \clip  (AA4) circle (1.2cm);
      \draw[dashed]  (AA4) circle (1.2cm); 
      
  \draw [line width=.5pt] ($ (A1)! -4/1 ! (A3) $) -- ($ (A3)! -4/1 ! (A1) $) ;   
   \draw [line width=.5pt] ($ (A2)! -4/1 ! (A3) $) -- ($ (A3)! -4/1 ! (A2) $) ;   
   \draw [line width=.5pt] ($ (A4)! -4/1 ! (A3) $) -- ($ (A3)! -4/1 ! (A4) $) ;   
   \draw [line width=.5pt] ($ (A5)! -4/1 ! (A3) $) -- ($ (A3)! -4/1 ! (A5) $) ;   
              
\end{scope}

\end{tikzpicture} };

   \node at (0.6\linewidth, 0\linewidth) {
  \begin{tikzpicture}[scale=\sca]
  
 \filldraw  (1.22418,0.582795)  coordinate (A1)  circle(.8pt)  node[below =1.5pt]{$1$};
  
 \filldraw  (2.00383,1.22059)  coordinate (A2)  circle(.8pt)  node[ below=1.5pt]{$2$};
 
  \filldraw  (1.2035,2.12622)  coordinate (A3)  circle(.8pt)  node[ right=1.5pt]{$3$};
 
  \filldraw  (0.576346,0.868565)  coordinate (A4)  circle(.8pt)  node[below =1.5pt]{$4$};
  
    \filldraw  ($(intersection of  A1--A3 and A2--A4)+(.3,.3)$)  coordinate (A5)  circle(.8pt)  node[  above=1.5pt]{$5$};
    
        \begin{scope}
    
 \coordinate   (AA)  at ( $  (A1)+(A2)+(A3)+(A4)  $) ;
  \coordinate   (AA4)  at ($ (0,0)! 1/4 ! (AA) $)  ;
     \clip  (AA4) circle (1.2cm);
      \draw[dashed]  (AA4) circle (1.2cm);

  \draw [line width=.5pt] ($ (A1)! -4/1 ! (A4) $) -- ($ (A4)! -4/1 ! (A1) $) ;   
   \draw [line width=.5pt] ($ (A2)! -4/1 ! (A4) $) -- ($ (A4)! -4/1 ! (A2) $) ;   
   \draw [line width=.5pt] ($ (A3)! -4/1 ! (A4) $) -- ($ (A4)! -4/1 ! (A3) $) ;   
   \draw [line width=.5pt] ($ (A5)! -4/1 ! (A4) $) -- ($ (A4)! -4/1 ! (A5) $) ;   
              
\end{scope}

\end{tikzpicture} };

   \node at (0.8\linewidth, 0\linewidth) {
  \begin{tikzpicture}[scale=\sca]
  
 \filldraw  (1.22418,0.582795)  coordinate (A1)  circle(.8pt)  node[ right =1.5pt]{$1$};
  
 \filldraw  (2.00383,1.22059)  coordinate (A2)  circle(.8pt)  node[ below=1.5pt]{$2$};
 
  \filldraw  (1.2035,2.12622)  coordinate (A3)  circle(.8pt)  node[ right=1.5pt]{$3$};
 
  \filldraw  (0.576346,0.868565)  coordinate (A4)  circle(.8pt)  node[below =1.5pt]{$4$};
  
    \filldraw  ($(intersection of  A1--A3 and A2--A4)+(.3,.3)$)  coordinate (A5)  circle(.8pt)  node[  above=1.5pt]{$5$};
    
        \begin{scope}
    
 \coordinate   (AA)  at ( $  (A1)+(A2)+(A3)+(A4)  $) ;
  \coordinate   (AA4)  at ($ (0,0)! 1/4 ! (AA) $)  ;
     \clip  (AA4) circle (1.2cm);
      \draw[dashed]  (AA4) circle (1.2cm);

  \draw [line width=.5pt] ($ (A1)! -4/1 ! (A5) $) -- ($ (A5)! -4/1 ! (A1) $) ;   
   \draw [line width=.5pt] ($ (A2)! -4/1 ! (A5) $) -- ($ (A5)! -4/1 ! (A2) $) ;   
   \draw [line width=.5pt] ($ (A3)! -4/1 ! (A5) $) -- ($ (A5)! -4/1 ! (A3) $) ;   
   \draw [line width=.5pt] ($ (A4)! -4/1 ! (A5) $) -- ($ (A5)! -4/1 ! (A4) $) ;   
              
\end{scope}

\end{tikzpicture} };

\end{tikzpicture}

\caption{Top: A configuration of five points in ${\mathbb {RP}}^{2}$. Bottom: For every point, draw a line crossing that point and every one else. This way, we can read the GCO just from the configuration of points. For example, in the last graph, imagining that we start with the line crossing point $P_1$ and $P_5$, rotate it around $P_5$, it will cross $P_3$,$P_2$ and $P_4$ consecutively until it crosses point $P_1$ again. According to the argument in figure \ref{dualgraph}, in the dual graph, i.e., the arrangement of lines, there will be intersection points, $L_5\cap L_1, L_5\cap L_3,L_5\cap L_2,L_5\cap L_4$ consecutively on the line $L_5$, resulting in a standard color ordering $(1324)$ there. So do other graphs in the bottom of this figure, resulting in a $(3,5)$ GCO $((2435), (1354), (1425),(1253), (1324))$, which is dual to $(13524)$. Compared with \cref{PartiID35}, the present $P_5$ is put in the chamber $C_5$ there  and the present resulting GCO is consistent with the Parke-Taylor factor associated with  $C_5$ given in (\ref{12PT5pt}).
}
	\label{fiveconfiguration}

\end{figure}

 \usetikzlibrary{angles,quotes}
 \usetikzlibrary{calc}	 
 
 \begin{figure}[h!]

 	\centering
 	\hspace{0.1in}
  \begin{tikzpicture}[scale=.88]
 \begin{scope}[xshift=0cm,yshift=0cm, scale=0.6]
 \draw[->] (-5,0)--(5,0) node[right]{$x$};
 \draw[->] (0,-5)--(0,5) node[above]{$y$};
   \draw (5.3,5.3)--+(0,.5) (5.3,5.3)--+(0.5,0) ;
  \node at (6.40,5.85) {\small ${\mathbb {RP}}^{k-1}$};
 \draw (.02,0) coordinate (A) -- (0,0) coordinate (B)
          -- (3,2.6) coordinate (C) 
   pic ["$\theta_t$", draw,
    angle eccentricity=1.45] {angle} ;
     \draw [line width=1pt] ($ (B)! -1/1 ! (C) $) -- ($ (C)! -.6/1 ! (B) $) ;
    \filldraw [blue] (B) circle(3pt)  node[below right=1.5pt]{$P_{i_1,i_2,\cdots,i_{k-2}}$};
 \filldraw [blue] (B) circle(3pt)  node[below right=15.5pt]{$\!\!\!\!\!\!\! x=0 \cap y=0 $};
     \filldraw [red] (C) circle(3pt)  node[below=1pt]{$\qquad \qquad \qquad \qquad \quad  P_t: ( r_t \cos \theta_t, r_t \sin \theta_t,1, $}
     node[below=15pt]{$\qquad \qquad \qquad \qquad \qquad\qquad  x_4,x_5,\cdots,x_k )$};
   %
 %
  \end{scope}
  \begin{scope}[xshift=9cm,yshift=0cm, scale=0.6]
 \draw[->] (-5,0)--(5,0) node[right]{$a$};
 \draw[->] (0,-5)--(0,5) node[above]{$c$};
   \draw (5.3,5.3)--+(0,.5) (5.3,5.3)--+(0.5,0) ;
  \node at (6.40,5.85) {\small ${\mathbb {RP}}^{k-1}$};
 \draw[blue, line width=2pt] (-4.3,0)--(4.3,0) node[below]{$H_{i_1}\cap H_{i_2}  \cdots \cap H_{i_{k-2}} $} 
 node[below=14pt]{$ c=0 \cap_{j=4}^{k} a_j=0$};
  \filldraw (-2,0) coordinate (A) circle(3pt) ;
  \coordinate (B) at  (0,-2.5);
  \draw [red, line width=2pt] ($ (A)! -1/1 ! (B) $) -- ($ (B)! -1/1 ! (A) $) -- ($ (B)! -.5/1 ! (A) $) node [right] {$H_t:  r_t {\cos} \theta_t a +  r_t {\sin} \theta_t +c $}
  node [below right=6pt] {$\qquad\qquad + \sum_{j=4}^k x_j a_j =0  $};
  \node at (-3.2,-.6) {$ (-{\tan}\theta_t,1,0,\cdots,0)$};
   \filldraw (-2,0) circle(3pt) ;
  \node at (-3.2,-1.6) {$H_t\cap H_{i_1} \cdots \cap H_{i_{k-2}} $};
  \end{scope}
 \end{tikzpicture}	
 %
 	\caption{
  Left: A presentation of a configuration of points in $\mathbb{RP}^{k-1}$ by setting $z=1$.
 	  Right: A presentation of an arrangement of  hyperplanes in ${\mathbb {RP}}^{k-1}$ by setting $b=1$.   A point $(x,y,z,x_4,x_5,\cdots,x_k)\in \mathbb{RP}^{k-1}$ in one projective space is mapped to the codim-1 hyperplane $H=\{ (a,b,c,a_4,a_5,\cdots,a_k)\in \mathbb{RP}^{k-1}:\, ax+by+cz+\sum_{j=4}^k a_j x_j=0  
 	   \}$ in the other projective space and vice versa. More explicitly, the unique codim-2 hyperplane denoted as  $P_{i_1,i_2,\cdots,i_{k-2}}$ that crosses points $P_{i_1},P_{i_2},\cdots, P_{i_{k-2}}$ is mapped to the line $H_{i_1}\cap H_{i_1}  \cdots \cap H_{i_{k-2}} $  in the arrangement of hyperplanes and the codim-1 hyperplane  $P_{t,i_1,i_2,\cdots,i_{k-2}}$ is mapped to the point $H_t\cap H_{i_1}  \cdots \cap H_{i_{k-2}} $.
   The  coordinate system is chosen such that the codim-2 hyperplane $P_{i_1,i_2,\cdots,i_{k-2}}$ is described by $\{ (x,y,z,x_4,x_5,\cdots,x_k)\in \mathbb{RP}^{k-1}:\, x=0 \,\& \, y=0 
 	   \}$ but different choices of 
     coordinate system will not affect the intersection relations of hyperplanes in the dual space. Under the current coordinate system, we see the position of the point $H_t\cap H_{i_1}  \cdots \cap H_{i_{k-2}} $ only depends on the angle $\theta_t$ and 
      one can imagine that as $\theta_t$ increases from 0 to $\pi$,  the point $H_t\cap H_{i_1}  \cdots \cap H_{i_{k-2}} $ moves along axis $a$ from the origin  to $(-\infty, 0)$ and then comes back from  the other direction. 
      }
 	\label{dualgraphhigherk}
 \end{figure}
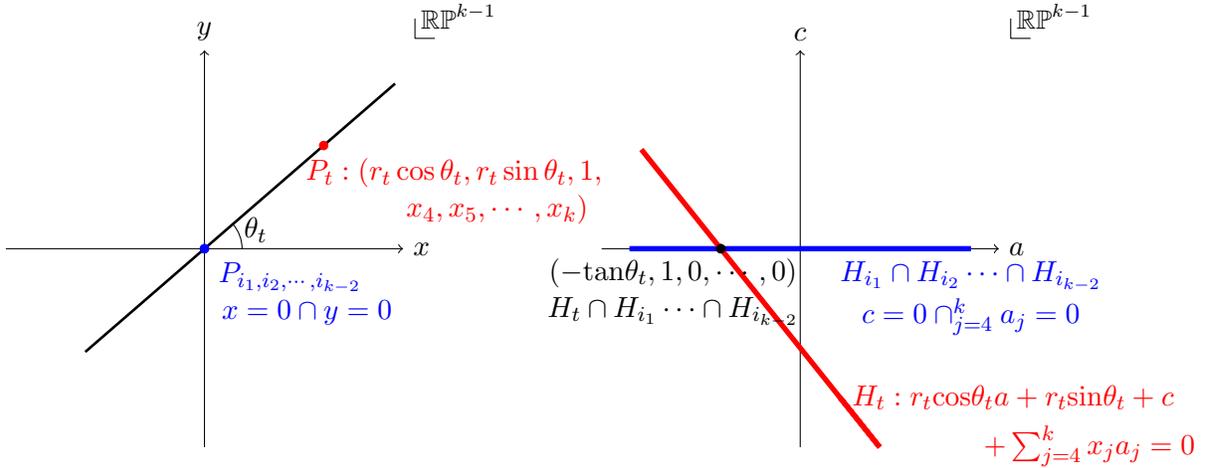

\subsection{Reading a GCO from a Configuration of Points}

In \cref{sec4}, we found that relating configurations of points and GCOs was very useful in geometrizing the decoupling identities. Here we explain the connection in detail 
and show how to read GCOs directly from the configuration of points. This will allow us to find a geometric understanding of the ``basic" decoupling identities explained in the next section which are the natural analog of the $U(1)$ decoupling identities.

Let us start with $k=3$. A $(3,n)$ GCO, which is encoded in the intersection pattern of lines $\{L_1,L_2\ldots ,L_n\}$ in $\mathbb{RP}^2$, can also be represented in the dual configuration given by points $\{ P_1,P_2,\ldots ,P_n \}$ in $\mathbb{RP}^2$ using the following procedure. For each point, $P_i$, draw lines connecting it to other $n-1$ points. Starting on any of the lines, circle around $P_i$ recording the labels of the lines as they are crossed. The resulting list  is the $(2,n-1)$ color ordering in the $i^{\rm th}$ entry of the GCO.

Let us denote the resulting list as $\tilde\sigma^{(i)}$ while the one on the line $L_i$ as  $\sigma^{(i)}$.
To prove $\tilde\sigma^{(i)}=\sigma^{(i)}$, one just needs to make use of the explicit duality relation which maps a point $(a,b,c)\in \mathbb{RP}^2$ to the line $L=\{ (x,y,z)\in \mathbb{RP}^2:\, ax+by+cz=0 \}$. As shown in \cref{dualgraph}, the point $P_i$ is mapped to $L_i$ and the line $P_{i,t}$ that crosses $P_i$ and $P_t$ is mapped to the point $L_i\cap L_t$ whose position only depends on $\theta_t$.  In a configuration of $n$ points in $X(3,n)$, $P_t$ could be any points except $P_i$. As we read the lines $P_{i,t}$ clockwise or anti-clockwise, it gives the ordering $\tilde\sigma^{(i)}$; in the meantime, their dual points $L_i\cap L_t$  sit consecutively on the line $L_i$. While the ordering of the latter is just $\sigma^{(i)}$.
Hence, $\tilde\sigma^{(i)}=\sigma^{(i)}$.

Reading the $(2,n-1)$ color ordering for every point $P_i$ with $i\in [n]$, we get the whole $(3,n)$ GCO for the configuration of points in $\mathbb{RP}^{2}$. See an example in \cref{fiveconfiguration}.

The generalization to any $k\geq 3$ is straightforward.
\begin{prop}\label{readingGCO}
Given a configuration of $n$ points in $\mathbb{RP}^{k-1}$, 
for any $k-2$ points  $\{P_{i_1},P_{i_2},$ $\cdots,$ $P_{i_{k-2}}\}$, find the unique codim-2 hyperplane denoted as $P_{i_1,\cdots, i_{k-2}}$ that crosses all of the $k-2$ points.  Start with any codim-1 hyper-plane that crosses the  codim-2 one and rotate it around the codim-2 one, it will cross the remaining points consecutively, resulting in a list of $n-k+2$ labels.  As proved in \cref{dualgraphhigherk}, the resulting list is the standard color ordering in the entry of the GCO with labels $\{i_1,\cdots, i_{k-2}\}$ removed. This way, we get a $(k,n)$ GCO directly from  the configuration of $n$ points in ${\mathbb {RP}}^{k-1}$.
\end{prop}

\subsection{Construction of Chirotopes from GCOs \label{gcotochirotope}}

Denote the sign of a Plucker variable at a generic point in the Grassmannian $G(k,n)$ by
\be 
\chi_{i_1,i_2,\cdots, i_k} \equiv  {\rm sgn} \,\Delta_{i_1,i_2,\cdots, i_k},\qquad \forall \, \{i_1,i_2,\cdots, i_k\} \subset [n].
\ee
We shall assume that all Plucker coordinates are nonzero, in which case we obtain a point $(\chi) \in \{\pm1\}^{\binom{n}{k}}$
known as a (realizable) simplicial chirotope, or equivalently an oriented uniform matroid \cite{bjorner1999oriented}. 

The positive Grassmannian is the subset of the real Grassmannian where all Plucker coordinates are positive, $\Delta_{j_1,\ldots, j_k} > 0$.  
In \cite{Arkani-Hamed:2012zlh,Arkani-Hamed:2009ljj}, it is argued that the Parke-Taylor factor $\text{PT}^{(k)}(\mathbb{I})$ in Equation \eqref{PTform}, multiplied by the canonical measure of the positive Grassmannian is its canonical form, while the Parke-Taylor factor itself is sometimes called the canonical function.  See \cite{Arkani-Hamed:2017tmz} for details about positive geometry and canonical forms \footnote{While the general definition of canonical form is defined in \cite{Arkani-Hamed:2017tmz}, here we give some intuition. The simplest canonical form is the one for an interval $x\in [0,1]$ given by $\frac{d x}{x(1-x)}$ which has two unit residues at the boundaries of the interval. In general, the canonical form of a positive geometry is required to have a residue on each stratum that is again a canonical form, for the new semi-algebraic set.}.

Now, given any generic point in the real Grassmannian, after modding out by the torus $(\mathbb{R}_{\not=0})^n$ we can fix an affine chart such that a certain collection of $n$ Plucker coordinates are all positive. Then
the hypersurfaces $\Delta_{i_1,i_2,\cdots, i_k}=0$ in the real Grassmannian induce a decomposition of the configuration space $X(k,n)$ into chambers.  These open chambers are characterized by so-called reorientation classes of simplicial chirotopes, that is, a vector of signs $\pm1$ of length $\binom{n}{k}$, modulo the scaling action of the torus $(\mathbb{Z}\slash 2)^n$.  Among all simplicial chirotopes, one has been particularly well-studied: the chirotope with all positive signs $\chi_J=1$.  

As explained in \cite[eq 6.8]{Arkani-Hamed:2019mrd}, modding out by the torus $T^+ = (\mathbb{R}_{>0})^n$, one obtains the positive configuration space $X^+(k,n) = G^+(k,n)/T^+$; moreover, the PT function can also be used to compute the canonical form of $G^+(k,n)/T^+$. In the same way, it was proposed in \cite{unpublishedYong} that the 372 (3,6) forms provide an analog for spaces that are supposed to be closely related to the chambers just mentioned where the Plucker variables are allowed to have other signs.

Let us now formulate a dictionary between GCOs and chirotopes\footnote{Compare to for example \cite{bokowski2001folkman} for a rank three construction.}.  More precisely, for each GCO we construct an equivalence class of chirotopes modulo the torus action\footnote{According to the torus action,   we can multiply the $i^{\rm th}$ column of the $M$ matrix shown in \eqref{cegmPot} by $-1$, which turn every $\chi_{ib\dots c}$ of a chirotope to be its opposite. We say the resulting chirotope and the original one are in the equivalent class. }
: we select a representative chirotope by fixing a projective frame.  The nontrivial direction is to map GCO to chirotope; we state the full solution.

Fix a GCO $\Sigma = (\sigma^{(12\cdots k-2)},\ldots,\sigma^{(n-k+3,\dots,n)})$ as in \eqref{defineGCO}, where $\sigma^{(L)}$ is a cyclic order on $\{1,\ldots, n\}\setminus L$ with $L = \{\ell_1,\ldots, \ell_{k-1}\}$.  Here each element $(\chi_{J}) \in \{\pm1\}^{\binom{n}{k}}$ is antisymmetric in its indices.  For example, when $k=3$ we have $\chi_{abc} = -\chi_{bac} = \chi_{bca}$.  We first fix a frame by putting $\chi_{I} = 1$ for all $ I = \{i_1,\ldots, i_k\} \subset \lbrack 1,k+1\rbrack$ and $\chi_{1,2,\cdots,k-1,j}=1$ for $j\in [k+2,n]$.  We have only to solve a system of $\binom{n}{k-2}\cdot \left(\binom{n-(k-2)}{k-1} - (n-(k-2)) \right)$ monomial equations in $\binom{n}{k}-n$ unknowns.  For concreteness and to simplify notation, let us specialize to the case $k=3$.

Then we have a system of $n\left(\binom{n-1}{2}-(n-1)\right)$ monomial equations in $\binom{n}{3}-n$ unknowns,
	\begin{eqnarray}\label{eqn: crossRatios GCO to Chirotope}
		\frac{\chi_{\ell j_1j_4}\chi_{\ell j_2,j_3}}{\chi_{\ell j_1j_3}\chi_{\ell j_2j_4}} & = & 1
	\end{eqnarray}	
	for each $\ell=1,\ldots, n$ and each cyclically-ordered 4-element subset $(j_1,j_2,j_3,j_4)$ of $\{1,\ldots, n\}\setminus \{\ell\}$ of the form
	$$\sigma^{(\ell)} = (\ldots, j_1,j_2,\ldots, j_3,j_4,\ldots).$$

	In general, given a pseudo GCO $\Sigma$, then $\chi_\Sigma$ can still be found but it may be non-realizable, in the sense that there would not exist a point in the Grassmannian whose Plucker coordinates $\Delta_{abc}$ have sign $\chi_{abc}$.  

For example, using \eqref{eqn: crossRatios GCO to Chirotope}, we find that the four types of $(3,6)$ GCOs, with representatives  
 given in \eqref{type036GCO}-\eqref{typeIII36GCO}, that is 
 \begin{align*}
& \Sigma_{0}   =((2 3 4 5 6),(1 3 4 5 6),(1 2 4 5 6),(1 2 3 5 6),(1 2 3 4 6),(1 2 3 4 5)),  \\
&\Sigma_{I}   = ((2 5 4 3 6),(1 5 4 3 6),(1 2 4 5 6),(1 2 3 5 6),(1 2 3 4 6),(1 2 5 4 3)),\\
&\Sigma_{II}   =((2  3465 ),(1 3465 ),(1 2 4 5 6),(1 2 3 5 6),(1 2 6 3 4),(1 2 5 3 4)), \\
& \Sigma_{III}   =((2  3645 ),(1 3465 ),(1 2 4 5 6),(1 5 3 2 6),(1 2 6 3 4),(1 3524 )),
\end{align*}
evaluate to the four rows of the matrix
\begin{align}
\begin{blockarray}{ccccccccccccccccccccc}
\matindex{1} &\matindex{1} &\matindex{1} &\matindex{1} &\matindex{1} &\matindex{1} &\matindex{1} &\matindex{1} &\matindex{1} &\matindex{1} &\matindex{2} &\matindex{2} &\matindex{2} &\matindex{2} &\matindex{2} &\matindex{2} &\matindex{3} &\matindex{3} &\matindex{3} &\matindex{4} &\\[-2mm]
\matindex{2} &\matindex{2} &\matindex{2} &\matindex{2} &\matindex{3} &\matindex{3} &\matindex{3} &\matindex{4} &\matindex{4} &\matindex{5} &\matindex{3} &\matindex{3} &\matindex{3} &\matindex{4} &\matindex{4} &\matindex{5} &\matindex{4} &\matindex{4} &\matindex{5} &\matindex{5} &\\[-2mm]
\matindex{3} &\matindex{4} &\matindex{5} &\matindex{6} &\matindex{4} &\matindex{5} &\matindex{6} &\matindex{5} &\matindex{6} &\matindex{6} &\matindex{4} &\matindex{5} &\matindex{6} &\matindex{5} &\matindex{6} &\matindex{6} &\matindex{5} &\matindex{6} &\matindex{6} &\matindex{6} & \\
    \begin{block}{(cccccccccccccccccccc)c}
 1 & 1 & 1 & 1 & 1 & 1 & 1 & 1 & 1 & 1 & 1 & 1 & 1 & 1 & 1 & 1 & 1 & 1 & 1 & 1 & \\
  1 & 1 & 1 & -1 & 1 & 1 & 1 & 1 & 1 & 1 & 1 & 1 & 1 & 1 & 1 & 1 & 1 & 1 & 1 & 1 \\
 1 & 1 & 1 & 1 & 1 & 1 & 1 & 1 & 1 & -1 & 1 & 1 & 1 & 1 & 1 & -1 & 1 & 1 & 1 & 1& \\
 1 & 1 & 1 & 1 & 1 & 1 & 1 & 1 & -1 & -1 & 1 & 1 & 1 & 1 & 1 & -1 & 1 & 1 & 1 & 1 &\\
    \end{block}
  \end{blockarray}
\end{align}
respectively.\footnote{
We emphasize here that for clarity we have used the torus action to flip the sign of each $\chi_{ij6}$, in order to bring the chirotope obtained by solving equations \eqref{eqn: crossRatios GCO to Chirotope},  $\{1 , 1 , 1 , 1 , 1 , 1 , -1 , 1 , -1 , -1 , 1 , 1 , -1 , 1 , -1 , -1 , 1 , -1 , -1 , -1 \}$, to a particularly simple form.

The four resulting  chirotopes are consistent with those in the literature after relabeling of coordinates. See \url{https://finschi.com/math/om/?p=catom&card=6&rank=3&filter=nondeg} for an example.}
The entries of each row are lexicographically ordered $(\chi _{123},\cdots, \chi _{456})$ as shown above the matrix.


We comment that the type 0 and I GCOs \eqref{type036GCO} and \eqref{typeI36GCO} are related by a triangle flip via $\{1,2,6\}$, as shown in their arrangements of lines in \cref{triangles36} or in the configurations of points in ${\mathbb {RP}}^2$ in \cref{P36point1Text} and here we see their chirotopes are related by flipping $\chi_{126}$.

The explicit dictionaries between GCOs and chirotopes
for $(3,7),(3,8), (3,9)$ are put in the ancillary file.

\section{Double Extensions and Fundamental Identities  \label{ref4B}  }

According to Definition \ref{k3decoupling}, the notion of decoupling for general $k$ given in \cite{Cachazo:2022pnx} is closely related to the projection operator defined in \eqref{projectionkkkk}, $\pi_i:CO_{k,n}\to CO_{k,n-1}$. 

GCOs are grouped together according to the result from this projection. In other  words, one defines equivalence classes 
\begin{align}
\label{decouplingset}
[ \Sigma^{[k]}]_i  :=  \{
    {\tilde\Sigma}^{[k]} \in CO_{k,n} | 
    \pi_i\big(
   { \tilde\Sigma}^{[k]}
    \big)   = \pi_i\big(
  \Sigma^{[k]}
    \big)   \}\,,
    \qquad \forall\,  {\Sigma}^{[k]} \in CO_{k,n}\, {\rm and}~ i\in [n].
\end{align}

The GCOs in each equivalence class
produce identities that are in general reducible, i.e. the set of integrands associated with the GCOs satisfy more than one linear relation. There are two ways to think about this.  Either the correct analog of $U(1)$ decoupling for $k\ge 3$ would necessarily yield reducible identities, or else there is a finer, deeper structure, which selects more elementary collections of GCOs whose integrands satisfy a unique relation. 

Clearly, one would like to find a deeper structure that results in irreducible identities; therefore we need to refine Definition \ref{k3decoupling}. In section \ref{sec4d1}, we discussed techniques for constructing irreducible identities for $k=3$ and found two classes. In this section, we show that the non-partitioning identities can be obtained using such a finer notion of decoupling. 

Our proposal for the refinement involves constraints imposed simultaneously on $CO_{k,n}$ and its dual $CO_{n-k,n}$; consequently, any identities or associated algebraic structures that may arise would be manifestly compatible with the duality on configuration spaces $X(k,n)\leftrightarrow X(n-k,n)$.   Actually, we are going to define a GCO set by using the intersection of one decoupling set in   $X(k,n)$ and the dual of the other in $X(n-k,n)$  obtained by decoupling different particles. 

According to \eqref{component}, every GCO $\Sigma^{[k]}$ of type $(k,n)$ is a set of $n$ GCOs $\Sigma^{(i),[k-1]}$ of type $(k-1,n-1)$. 
 We say that $\pi_i\big(\Sigma^{[k]}\big)$ is the $k$-\textit{preserving} projection of $\Sigma^{[k]}$ in the direction $i$.  On the level of the configuration space, the $k$-preserving projection in the $j^\text{th}$ direction acts by 
$$(H_1,\ldots, H_n) \mapsto (H_1,\ldots, \widehat{H_j},\ldots, H_n).$$
On the contrary, we denote by $\pi_{(j)}$ the $k$-\textit{decreasing} projection 
\be
\label{componentpro}
\pi_{(j)}\big(\Sigma^{[k]}\big)\equiv \Sigma^{(j),[k-1]} =  \{ \sigma^{(j,i_2,\cdots,i_{k-2})}| \{i_2,\cdots, i_{k-2}\} \subset [n] \setminus \{j\} \} . 
\ee
As above, again on the level of the configuration space, the $k$-\textit{decreasing} projection in the $j^\text{th}$ direction acts by 
$$(H_1,\ldots, H_n) \mapsto (H_1\cap H_j,\ldots, \widehat{H_j\cap H_j},\ldots, H_n\cap H_j\}.$$

 As explained in \cite{Cachazo:2022pnx}, the  dual of a $k$-decreasing projection of a GCO is a $k$-preserving projection of the dual GCO.

Now we prepare to introduce double extensions and fundamental decoupling identities.

\begin{defn}
\label{tripleT}
    A triple $\mathcal{T} = (\Sigma_0,\Sigma_1,\Sigma_2)$, with  
    	$$\Sigma_0 \in CO_{k-1,n-2},\ \Sigma_1 \in CO_{k-1,n-1},\ \Sigma_2 \in CO_{k,n-1},$$
     is said to be fundamental if $\pi_i(\Sigma_1) = \pi_{(j)}(\Sigma_2) = \Sigma_0$.

     Further, a GCO $\Sigma \in CO_{k,n}$ is said to be a double $\mathcal{T}$-extension of $\Sigma_0$ provided that $\pi_{(i)}(\Sigma) = \Sigma_1$ and $\pi_{j}(\Sigma) = \Sigma_2$ for some $i\not=j$.
\end{defn}


\begin{conjecture}\label{conjecture: fundamental decouplings}
		Fix a fundamental triple $\mathcal{T} = (\Sigma_0,\Sigma_1,\Sigma_2)$.  Then, as $\Sigma \in CO(k,n)$ ranges over all double $\mathcal{T}$-extensions $\Sigma$, the CEGM integrands $\mathcal{I}(\Sigma)$ satisfy a unique linear identity with all coefficients $\pm1$.  We call such a relation among integrands a \textit{fundamental decoupling identity}. 
\end{conjecture}

Moreover, we also claim that all other irreducible decoupling identities for $k=3$ that can be obtained using the algorithm explained in section \ref{sec4} are linear combinations of the fundamental ones.  
 
We have verified Conjecture \ref{conjecture: fundamental decouplings} for all decoupling identities among CEGM integrands for $CO_{3,6}$ and $CO_{3,7}$.

The fundamental decoupling identities also have a beautiful geometric interpretation. Consider a $(k,n-1)$ GCO expressed as an arrangement of $n-1$ hyperplanes with labels in the set $[n]\setminus \{i\}$. Now introduce another hyperplane, labeled $i$, in the arrangement in all possible ways but subject to the condition that $H_i\cap H_j$ defines a predetermined arrangement. This geometric picture will allow us the connect fundamental identities with the non-partitioning identities in section \ref{sec4} and hence provide a proof of conjecture \ref{conjecture: fundamental decouplings} for $k=3$.

\subsection{Relation to 
Non-partitioning Irreducible Identities and $k=3$ Proof}

In \cref{sec4}, we discussed two kinds of $k=3$ irreducible decoupling identities, partitioning and non-partitioning. Here we prove that non-partitioning identities coincide with the fundamental identities. 

Let us start by noting the following property.

\begin{prop}\label{nonpartitioningproperty}
All $(3,n)$ GCOs that participate in a non-partitioning irreducible identity where label $i$ was decoupled and which is generated by walls sharing label $j$ all have the same $(2,n-1)$ color ordering in the $j^{\rm th}$ position. In other words, their $k$-decreasing projections $\pi_{(j)}(\Sigma )$ coincide.
\end{prop}

\begin{proof}[Sketch of Proof]
Start with a $(3,n-1)$ GCO $\Sigma^{(i)}$ in the set $[n]\setminus \{ i \}$. This defines a decoupling set in $(3,n)$ by collecting all GCOs such that $\pi_i(\Sigma )=\Sigma^{(i)}$. $\Sigma^{(i)}$ also defines a configuration of $n-1$ points in $\mathbb{RP}^2$. Now, according to the construction in section \ref{sec4}, a non-partitioning irreducible decoupling is found by selecting one of the points, say $P_j$, drawing two lines intersecting at $P_j$ and containing two other points, such that the lines are adjacent. This means that the lines split $\mathbb{RP}^2$ into two regions so that all the other $n-3$ points are located in one of the regions. The identity is obtained by letting point $P_n$ wander around in the ``empty" region and collecting all GCOs obtained in doing so. Clearly, any such GCOs will have the same $(2,n-1)$ color ordering in the $j^{\rm th}$ position since such an ordering is obtained by reading the lines joining $P_j$ to each point, including $P_n$, as we circle around $P_j$, but all point are fixed, except $P_n$ (see section \ref{sec:translation}). However, $P_n$ is not allowed to escape the region bounded by two existing lines. This means that all GCOs participating in this identity have the same $(2,n-1)$ color ordering $\sigma^{(j)}$ in the $j^{\rm th}$ entry.
\end{proof}

A simple corollary is that all $k=3$ non-partitioning identities are fundamental identities. We leave the proof to the reader. Since $k=3$ non-partitioning identities were proven in \cref{appproof}, it follows that conjecture \ref{conjecture: fundamental decouplings} holds for $k=3$.

\section{Discussions\label{sec5}}

In this work we completed the second element of the triality proposed in \cite{Cachazo:2022pnx}. The triality refers to three different ways of computing color-dressed biadjoint scalar amplitudes. Color-dressed amplitudes are defined as a sum over generalized color orderings (GCOs) multiplied by partial amplitudes (see \cref{coBSintro}). In this first element of the triality we have the generalized Feynman diagram (GFD) technique described in detail in \cite{Cachazo:2022pnx}. Each partial amplitude is computed as a sum over all GFDs which are locally planar with respect to the GCOs defining the partial amplitude. The second element in the triality is the computation of biadjoint amplitudes using CEGM integrals. In this work, we proposed CEGM integrands associated with a given GCO so that when integrated against the scattering equations they would produce the corresponding generalized biadjoint amplitude. This was explicitly checked for all $(3,6)$ and $(3,7)$ partial amplitudes. The third element is construction in terms of a generalization of the positive tropical Grassmannian, associated with a given chirotope (and hence a given GCO),
which we called {\it chirotopal} tropical Grassmannians in section 13 of \cite{Cachazo:2022pnx}. 
In particular, in \cite{Cachazo:2022voc}, the Global Schwinger integral over $\mathbb{R}^{(k-1)(n-k-1)}$ was introduced as a compact formula for generalized partial amplitudes with GCOs of type 0 (i.e. $k=2$ descendants). The main question here to complete the triality is the following: to what extent can this story be generalized to other GCOs?

The CEGM integrands we introduced in this work provide  compatible \textit{systems} of rational functions, associated with the connected components of the real configuration space $X(k,n)$, as indexed by Generalized Color Orders.  The numerator of each such integrand is uniquely determined (up to a sign) on $X(3,n)$ for $n\le 8$, by identifying the one-dimensional residues on the common boundaries with each of the neighbors of the connected component.  Starting at $(3,9)$ and $(4,8)$, some integrands are not completely determined in this way.

Could it be that additional constraints could arise by identifying higher dimensional residues?  We do not a priori get a constraint from integrands of non-realizable pseudo-GCOs. Even so, we would not be able to eliminate all free parameters.   It is not clear if taking iterated residues could fix the numerators of the integrands (this means that it is not clear whether our situation matches the positive geometry axioms of \cite{Arkani-Hamed:2017tmz}). At this point, it seems that the natural way to fix the free parameters is by computing explicit generalized biadjoint amplitudes and requiring their values to match that obtained by summing over generalized Feynman diagrams.  

There are numerous questions for future investigations into the CEGM integrals with integrands of general types. First of all, it is desirable  to study whether we can read the kinematic poles \cite{Guevara:2020lek,He:2020ray} that appear in the doubly partial amplitudes  $m^{(k)}(\Sigma,\tilde\Sigma)$ by just looking at the CEGM integrands, especially the Plucker variables in their denominators.

In CHY formulas, it is very easy to determine the poles of  $m^{(2)}(\alpha,\beta)$  by just considering the Plucker variables in the denominators of the standard Parke-Taylor factors (c.f. \cite{Cachazo:2013gna,Baadsgaard:2015voa,Cachazo:2015nwa,Dolan:2013isa,Broedel:2013tta}). By making use of compatibility rules of poles, which is universal for any orderings, one can easily construct the whole amplitudes just from the poles \cite{Gao:2017dek,He:2021lro}, and in the case of generalized biadjoint amplitudes, \cite{Early:2019eun} using the planar basis and matroid subdivisions. It would be interesting to explore how far away one can go for CEGM integrals in this direction. Other directions include studying  the contribution of singular solutions  \cite{Cachazo:2019ble} or the behavior of amplitudes \cite{GarciaSepulveda:2019jxn,Abhishek:2020xfy} in soft and hard limits, factorization \cite{Early:2022mdn}, smoothly splitting amplitudes and semi-locality \cite{Cachazo:2021wsz}, minimal kinematics to simplify the scattering equations and the resulting amplitudes \cite{Cachazo:2020uup,Cachazo:2020wgu}, etc. 

In this paper, we also studied irreducible decoupling identities in terms of CEGM integrands and it is highly desirable to find a systematic way to make use of them to generate algebraically linear independent sets of CEGM integrands, as the analog of the Kleiss-Kuijf basis in the stand $k=2$ case \cite{Kleiss:1988ne}. Looking even further into the future, one would like to find the BCJ-like relations among CEGM integrands on the support of scattering equations, or BCJ-like basis \cite{Bern:2008qj,Bjerrum-Bohr:2009ulz,Stieberger:2009hq,Cachazo:2013gna} of integrands. The latter may lay a foundation for possible BCJ-like double copy relations \cite{Bern:2008qj,Bern:2019prr, Bern:2022wqg,Adamo:2022dcm} among generalized amplitudes including the generalized biadjoint amplitudes.

It is well known that the CHY formula is closely related to the leading order of  disk integrals of open string theory \cite{koba1969manifestly}. The latter was formally generalized to higher $k$ formulas called  Grassmannian stringy integrals \cite{Arkani-Hamed:2019mrd}, whose leading orders are closely related to the CEGM integrals.  

There are many more avenues to explore. It is known that for globally planar integrands, what we call type 0, there is a connection to cluster algebras \cite{Arkani-Hamed:2020tuz,Arkani-Hamed:2019plo,He:2021zuv,Drummond:2019qjk,Drummond:2020kqg,Gates:2021tnp,Henke:2021ity,Cachazo:2019apa}. It is also known that for $X(3,n)$ with $n<9$ there is an action of $W(E_n)$, the Weyl group of $E_n$, acting on the points mapping chambers to chambers \cite{sekiguchi1997w,sekiguchi1999configurations}.  It would be fascinating to make a connection with either of these topics.  In fact, after the original posting of this article, the investigation was initiated in \cite{Early:2023cly} for n=$6,7$, where the connected components of the del Pezzo moduli space $Y(3,n)$ were constructed.

In 2009, Arkani-Hamed, Cachazo, Cheung, and Kaplan (ACCK) conjectured a  formula that uses the higher $k$ Parke-Tayler factors \eqref{PTform} as integrands and expresses tree-level amplitudes or loop discontinuities in ${\cal N} = 4$ SYM 
 in the sector with $k$ negative helicity gluons in the planar limit as contour integrals \cite{Arkani-Hamed:2009ljj}. 
 The CEGM integrands can be thought of as a natural generalization of the 
the higher $k$ Parke-Taylor factors and it would be interesting to see their possible applications to the ACCK formulas, especially the possible connections to non-planar on-shell diagrams and their Grassmannian formulations (see  \cite{Bourjaily:2016mnp,Frassek:2016wlg,Franco:2015rma,Paranjape:2022ymg}).

\subsection{Towards a Realization of Generalized Color Factors}

In this work, we have treated the generalized color factors, ${\bf c}(\Sigma)$, as unknown. And their purpose so far has been that of a bookkeeping device for the partial amplitudes \eqref{coBSintro}, 
\begin{equation}\label{coBdis}
    {\mathcal M}^{(k)}_n = \sum_{I,J} {\bf c}(\Sigma_I){\bf c}(\Sigma_J)\, m_n^{(k)}(\Sigma_I,\Sigma_J). 
\end{equation}
When $k=2$, color factors are given in terms of traces of generators of a Lie algebra. Decoupling identities are automatically obtained by a judicious choice of the generators and by asking the full color-dressed amplitude to vanish. 

The most pressing problem in this line of research is to find the explicit realization of generalized color factors, ${\bf c}(\Sigma)$, and use them to directly obtain decoupling identities. 

In this work, we have derived identities among CEGM integrands and called them ``decoupling" in analogy with the $k=2$ version. Of course, for them to actually decouple something, it is necessary that the color-dressed amplitude vanishes when the color factors behave in a particular way.   

As explained in \cite{Cachazo:2022pnx}, when we ``decouple" one particle in a $(3,6)$ amplitude, say 1, all 372 GCOs are identified with one of twelve possibilities. Each set contains $31$ GCOs that make up a decoupling set. We do not know the precise behavior of color factors, but if we assume that when two GCOs, say $\Sigma_I$ and $\Sigma_J$, are identified, the corresponding color factors satisfy ${\bf c}(\Sigma_I)=\pm {\bf c}(\Sigma_J)$, then \eqref{coBdis} would split into the $12$ sets with $31$ partial amplitudes in each. However, the $31$ partial amplitudes can be made to cancel in smaller sets, i.e. the set is reducible. 
 
The new element we learned in section \ref{ref4B} is that if we consider the double extension procedure one always gets irreducible identities. 

This hints at the fact that for $k>2$, one has to select a pair of particles $i\neq j$ to define the decoupling. It is as if one were decoupling $i$ with respect to $j$. In this case, the $31$ GCO in a decoupling set are not all identities but they are separated into four sets. From the view point of the $372$ GCOs, the $\{i,j\}$ decoupling splits \eqref{coBdis} (for $(k=3,n=6)$), into $48$ sets. Each set contains either $7$,$8$, or $9$ GCOs.   
Moving on to $(3,7)$, the 27240 GCOs are partitioned into $372$ under any $i^{\rm th}$ decoupling and into $1860$ sets under any $\{i,j\}$ decoupling.

We expect that requiring the generalized color factors to produce all identities, including the corresponding signs, would lead to enough constraints to provide a hint as to what algebraic structure is underlying their construction.

\subsection{Number of Linearly Independent Integrands}

For $k=2$ it is known that the space of integrands has dimension $(n-2)!$. In other words, all $(n-1)!/2$ can be expressed in terms of a basis of $(n-2)!$ integrands. This has a geometric interpretation as the dimension of the top cohomology group of $X(2,n)$. This is not an accident as when each integrand is combined with the measure on $\mathbb{CP}^1$ it becomes an element in the top cohomology. It is also known that $U(1)$ decoupling identities are enough to reduce the space to the correct dimension. The generalization to $k>2$ and any $n$ is an open problem. However, it is natural to expect that the fundamental decoupling identities can help us find the algebraically linearly independent CEGM integrands.

Let us comment on what we know so far and the prospects for the future. For $(3,6)$, we have proven that fundamental decoupling identities are indeed enough to reduce all 372 CEGM integrands down to 126 
linearly independent ones. 

The situation for $(3,7)$ is more complicated and we do not have a definite answer at this point. At this point, we have been able to show that fundamental identities can be used to express the $27240$ CEGM integrands in terms of only $7890$. However, it is known that the dimension should be $7470=7890-420$. One possible explanation for the $420$ discrepancy is that identities from the dual $(4,7)$ space are needed to reduce the number. We leave this exploration for the future. 

However, it is worth mentioning that there are computations using finite fields that can be used to compute information of $X(2,n)$ and $(3,n<10)$. The basic idea is to determine the number of rank $k$ uniform matroids over $\mathbb{F}_q$. Somewhat surprisingly, in the cases mentioned above, there is a quasi-polynomial expression for the numbers. In appendix A of
\cite{Agostini:2021rze}, T. Lam uses such polynomials to compute the Euler characteristic of the configuration spaces by evaluating the quasi-polynomials at $q=1$. As it turns out, evaluating the quasi-polynomials at $q=0$ provides the dimensions we are after \cite{skorobogatov1996number}. For $X(3,6)$ one indeed gets $126$ while for $(3,7)$ one gets $7470$. 

It is fascinating that we encounter matroid theory in all its guises in this subject. From positroids (i.e., matroids over the reals with special conditions), chirotopes (oriented matroids), to matroids over finite fields.

\section*{Acknowledgements}

The authors thank Bernd Sturmfels, Song He,  Jianrong Li, Matteo Parisi and Bruno Umbert for useful discussions. 
YZ would like to thank Alex Edison and Zhengjie Li for coding to find solutions of huge linear equations.
This research was supported in part by a grant from the Gluskin Sheff/Onex Freeman Dyson Chair in Theoretical Physics and by Perimeter Institute.  Research at Perimeter Institute is supported in part by the Government of Canada through the Department of Innovation, Science and Economic Development Canada and by the Province of Ontario through the Ministry of Colleges and Universities.  This research received funding from the European Research Council (ERC) under the European Union’s Horizon 2020 research and innovation programme (grant agreement No 725110), Novel structures in scattering amplitudes.

\appendix

\section{Polygons in Arrangement of Lines
\label{appa}}

\begin{claim}
The arrangement of lines associated with a $k=3$ color ordering $\Sigma =( \sigma^{(1)},\sigma^{(2)},\ldots ,\sigma^{(n)} )$ has a polygon bounded by lines $L_{i_1},L_{i_2},\cdots, L_{i_m}$  in sequence with $3\leq m\leq n-1$
if and only if labels in the sets $\{ i_m,i_2\}$,\,\, $\{ i_1,i_3\} \cdots $, and $\{ i_{m-1},i_1\}$ 
are consecutive in $\sigma^{(i_1)}$, $\sigma^{(i_2)}\cdots$, and $\sigma^{(i_m)}$ respectively.  
\end{claim}

The proof from a polygon to the consecutive conditions is straightforward. Now we want to prove it backward.
\begin{proof}
Let us denote the
intersecting point of $L_i$ and $L_j$ as $P_{ij}$.  Obviously, the connected broken line  $P_{i_1i_2}-P_{i_2i_3}-\cdots-P_{i_mi_1}-P_{i_1i_2}$, which has the same topology as a circle,
must divide the whole ${\mathbb {RP}}^2$ into two regions. Since we consider an arrangement of $n$ generic lines with $m>n$, there is at least one more line $L_j\notin \{L_{i_1},L_{i_2},\cdots, L_{i_m} \} $ that does not intersect with the connected broken line $P_{i_1i_2}-P_{i_2i_3}-\cdots-P_{i_mi_1}-P_{i_1i_2}$. That line also divides the whole ${\mathbb {RP}}^2$ space into two regions and the whole connected broken line is localized in one of the regions, which is forced 
 to form a polygon.
\end{proof}

When there is no other straight line that does not intersect with the broken line, we are not sure whether 
the broken line forms a polygon.  
Take the $(3,5)$ GCO  which descends from the standard canonical $(2,5)$ color ordering as an example. 
Combinatorically, one would get two connected broken lines with 5 segments, $P_{12}-P_{23}-\cdots-P_{51}-P_{12}$ and $P_{13}-P_{35}-P_{52}-P_{24}-P_{41}-P_{13}$. However, once we draw out the corresponding arrangement of lines, as shown in the graph on the left in \cref{35arrangement}, only the first broken line (in blue in the figure) forms a pentagon $\{1,2,3,4,5\}$, which is consistent with the canonical $(2,5)$ color ordering. But the second 
broken line (in red in the figure) fails.\footnote{One can easily figure out that the whole ${\mathbb {RP}}^2$ plane has been spanned by one pentagon, five triangles, and five quadrangles. So there is no space for the second broken line to form a polygon}.
In an equivalent representation of the arrangement of lines, as shown on the right in the figure, the second 
broken line is manifestly connected and divides the whole ${\mathbb {RP}}^2$ into two parts but it fails to become a polygon. As one can imagine, if a connected broken line does not have any intersection with a straight line, it is forced to form a polygon even in the projective space.

 \def \inter #1,#2 \inter { (intersection of  A#1--B#1 and A#2--B#2) }

 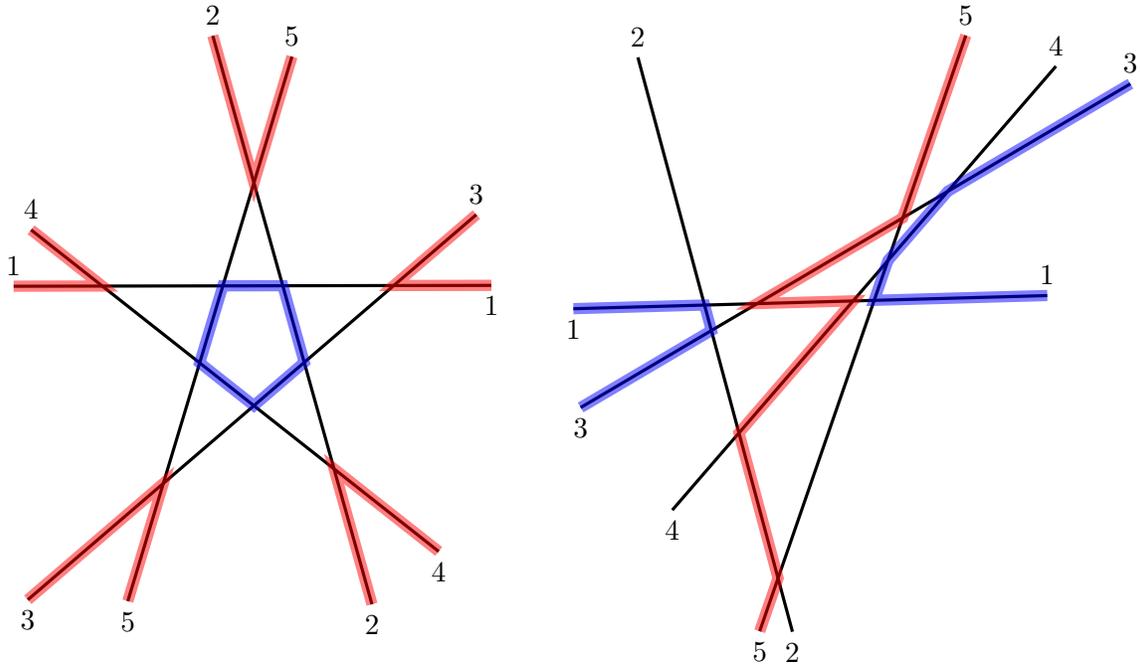
\begin{figure}
     \centering
\begin{tikzpicture}[scale=1.4]
  \draw [very thick](0.259384,0.894526) coordinate(A3) node [below]{3}
  node [below=18pt]{~
  {\color{white}.}}
  -- (4.51871,4.55482) coordinate (B3) node [above]{3} ;
 \draw [very thick](1.20921,0.881451) coordinate(A5) node [below]{5}-- (2.7679,6.05625) coordinate (B5) node [above]{5} ;
 \draw [very thick](3.52735,0.852126) coordinate(A2) node [below]{2}-- (2.01611,6.25446) coordinate (B2) node [above]{2} ;
 \draw [very thick](4.16121,1.35392) coordinate(A4) node [below]{4}-- (0.290148,4.41026) coordinate (B4) node [above]{4} ;
 \draw [very thick](4.66115,3.88243) coordinate(A1) node [below]{1}-- (0.124827,3.87435) coordinate (B1) node [above]{1} ;
 
 \draw [blue, line width=1.5mm, opacity=0.5] \inter 1,2 \inter --\inter 2,3 \inter --\inter 3,4 \inter--\inter 3,4 \inter --\inter 4,5 \inter--\inter 5,1 \inter-- cycle;
 
  \draw [red, line width=1.5mm, opacity=0.5]  (B1)--\inter 1,4 \inter--(B4)
   (B2)--\inter 2,5 \inter--(B5)
      (B3)--\inter 3,1 \inter--(A1)
            (A4)--\inter 4,2 \inter--(A2)
                (A5)--\inter 5,3 \inter--(A3)

  ;
 
    \end {tikzpicture}
    ~
    ~
     \begin{tikzpicture}[scale=1.4]
 \draw [very thick](4.80608,4.738) coordinate(A1) node [above]{1}-- (0.306075,4.61301) coordinate (B1) node [below]{1} ;
 \draw [very thick](5.59551,6.75014) coordinate(A3) node [above]{3}-- (0.375204,3.67355) coordinate (B3) node [below]{3} ;
 \draw [very thick](4.88872,6.91745) coordinate(A4) node [above]{4}-- (1.24561,2.69935) coordinate (B4) node [below]{4} ;
 \draw [very thick](4.03299,7.20941) coordinate(A5) node [above]{5}-- (2.07544, 1.54756) coordinate (B5) node [below]{5} ;
 \draw [very thick](0.918898,7.0017) coordinate(A2) node [above]{2}-- (2.38667,1.54307) coordinate (B2) node [below]{2} ;

     \draw [red, line width=1.5mm, opacity=0.5] (B5)-- \inter 5,2 \inter --\inter 2,4 \inter --\inter 4,1 \inter--\inter 1,3 \inter --\inter 3,5 \inter--(A5);
     
     \draw [blue, line width=1.5mm, opacity=0.5] (B1)-- \inter 1,2 \inter --\inter 2,3 \inter -- (B3)
    (A1)-- \inter 1,5 \inter--\inter 5,4 \inter --\inter 4,3 \inter --(A3);
  
        \end {tikzpicture}

     \caption{
     Two equivalent representations of a $(3,5)$ arrangement of lines.  The connected broken line in blue forms a pentagon but the one in red fails.
     }
     \label{35arrangement}
 \end{figure}

Anyway, the above claim does not apply to $n$-gons. However, it is easy to see that only type 0 $(3,n)$ GCO has an $n$-gon whose ordering is consistent with the $k=2$ color ordering from which the type 0 $(3,n)$ GCO descends. 
Therefore, we are still able to find all polygons for any $k=3$ GCOs combinatorially.

\begin{claim}
Any $(3,n)$ arrangement of lines has $n(n-1)/2+1
$ polygons. 
\end{claim}

This can be proved easily as there are 4 polygons in $(3,3)$ arrangement of lines and whenever we add a generic line to an $(3,n-1)$ arrangement of lines, we get $n-1$ more polygons.

\section{Comments on (3,9) Pseudo-GCOs}

In this appendix, we collect two interesting properties of $(3,9)$ pseudo-GCOs that we expect to be present in all $(3,n>8)$.

\subsection{$(3,9)$ Non-realizable Pseudo-GCOs \label{secnonreal39}}

Combinatorially, the non-realizable pseudo-GCO \eqref{pseudoGCO39} is indistinguishable from any other GCO and so one could try to compute its corresponding integrand by looking at the set of ``combinatorial" triangles, i.e., look for all triples $\{i,j,k\}$, such that $\{j,k\}$, $\{i,k\}$ and $\{i,j\}$ are consecutive in the $i^{\rm th}$, $j^{\rm th}$, and $k^{\rm th}$ orderings respectively. Following this recipe, the set of triangles is 
\be\label{pseudo39T}
\{1,4,7\},\{1,5,9\},\{1,6,8\},\{2,4,9\},\{2,5,8\},\{2,6,7\},\{3,4,8\},\{3,5,7\},\{3,6,9\},\{4,5,6\}.
\ee 
This would suggest that the corresponding integrand should have ten poles of the form $1/\Delta_{abc}$, with $\{a,b,c\}$, in the list of triangles. For $n=9$, this would be a basic integrand with $p=1$ and hence completely determined by the torus action on it. A simple counting of labels reveals that the numerator should be $\Delta_{456}$. Now we find a contradiction. The list \eqref{pseudo39T} contains the triangle $\{4,5,6\}$ and hence the denominator has a factor of $\Delta_{456}$. But now we find that it cancels with the numerator. It is as if the triangle had disappeared.

\usetikzlibrary{calc}

\begin{figure}

	\def \inter #1,#2 \inter { (intersection of  A#1--B#1 and A#2--B#2) }

\centering
		
  \begin{tikzpicture}[scale=4]
  \draw [very thick,black](-0.860811,1.26763) coordinate(A1) node [left]{1}-- (1.08682,-0.648311) coordinate (B1) node [below]{1} ;
 \draw [very thick,black](1.62205,0.137639) coordinate(A2) node [above]{2}-- (-0.787374,-0.52919) coordinate (B2) node [below]{2} ;
 \draw [very thick,black](-0.0424588,1.03193) coordinate(A3) node [above]{3}-- (-0.648549,-1.30657) coordinate (B3) node [below]{3} ;
 \draw [very thick,blue](-0.819615,1.41962) coordinate(A4) node [left]{4}-- (0.54641,-0.94641) coordinate (B4) node [below]{4} ;
 \draw [very thick,blue](-0.75,0.) coordinate (B5) node [below]{5} ;
  \draw [very thick,blue](1.7,0.) coordinate(A5) node [above]{5} ;
 \draw [very thick,blue](0.43923,0.76077) coordinate(A6) node [above]{6}-- (-0.768653,-1.33135) coordinate (B6) node [below]{6} ;
 \draw [very thick,red](-0.716712,1.56957) coordinate(A7) node [left]{7};
  \draw [very thick,red](-0.0694833,-0.927656) coordinate (B7) node [below]{7} ;
 \draw [very thick,red](1.69886,-0.158896) coordinate(A8) node [above]{8} ;
  \draw [very thick,red](-0.789933,0.529898) coordinate (B8) node [below]{8} ;
 \draw [very thick,red](1.04752,0.609653) coordinate(A9) node [above]{9};
  \draw [very thick,red](-0.96888,-1.37394) coordinate (B9) node [below]{9} ;

   \coordinate (O) at (intersection of A4--B4 and A6--B6);
  
  \draw [very thick,blue]  [rounded corners]  (A5) --  ($ (A5)! 9.5/10 ! (O) $) -- ($(O)+(0,.05)$)--  ($ (O)! 1.0/10 ! (B5) $) -- (B5);

    \coordinate (O0) at (intersection of A1--B1 and A4--B4);
     \coordinate (O1) at (intersection of A3--B3 and A5--B5);
      \coordinate (O2) at (intersection of A2--B2 and A6--B6);
  
  \draw [very thick,red]  [rounded corners]  (A7) -- ($ (A7)! 8.5/10 ! (O0) $) -- ($(O0)+(.03,.015)$) -- ($ (O0)! 1.5/10 ! (O1) $) -- ($ (O0)! 8.5/10 ! (O1) $) -- ($(O1)+(-.03,-.015)$) --  ($ (O1)! 2.5/10 ! (O2) $) --  ($ (O1)! 7.5/10 ! (O2) $) -- ($(O2)+(.03,.03)$)--   ($ (O2)! 1.0/10 ! (B7) $) -- (B7);
  
      \coordinate (O0) at (intersection of A3--B3 and A6--B6);
         \coordinate (O1) at (intersection of A2--B2 and A4--B4);
      \coordinate (O2) at (intersection of A1--B1 and A5--B5);

    \draw [very thick,red]  [rounded corners]  (B9) -- ($ (B9)! 8.5/10 ! (O0) $) -- ($(O0)+(-.015,.03)$) -- ($ (O0)! 1.5/10 ! (O1) $) -- ($ (O0)! 8.5/10 ! (O1) $) -- ($(O1)+(.03,-.015)$)--  ($ (O1)! 2.5/10 ! (O2) $) --  ($ (O1)! 7.5/10 ! (O2) $) -- ($(O2)+(-.03,.03)$)--    ($ (O2)! 1.0/10 ! (A9) $) -- (A9);
  
        \coordinate (O0) at (intersection of A2--B2 and A5--B5);
           \coordinate (O1) at (intersection of A1--B1 and A6--B6);
      \coordinate (O2) at (intersection of A3--B3 and A4--B4);
      
    \draw [very thick,red]  [rounded corners]  (A8) -- ($ (A8)! 8.5/10 ! (O0) $) -- ($(O0)+(0,-.04)$) -- ($ (O0)! 1.5/10 ! (O1) $) -- ($ (O0)! 8.5/10 ! (O1) $)  -- ($(O1)+(0,.035)$)-- ($ (O1)! 2.5/10 ! (O2) $) --  ($ (O1)! 7.5/10 ! (O2) $) -- ($(O2)+(0,-.04)$)--     ($ (O2)! 1.0/10 ! (B8) $) -- (B8);
    
		\end{tikzpicture}
			
\caption{An arrangement of pseudo-lines for a non-realizable pseudo-GCO.}	
	
\label{arrangementpseudo}

\end{figure}
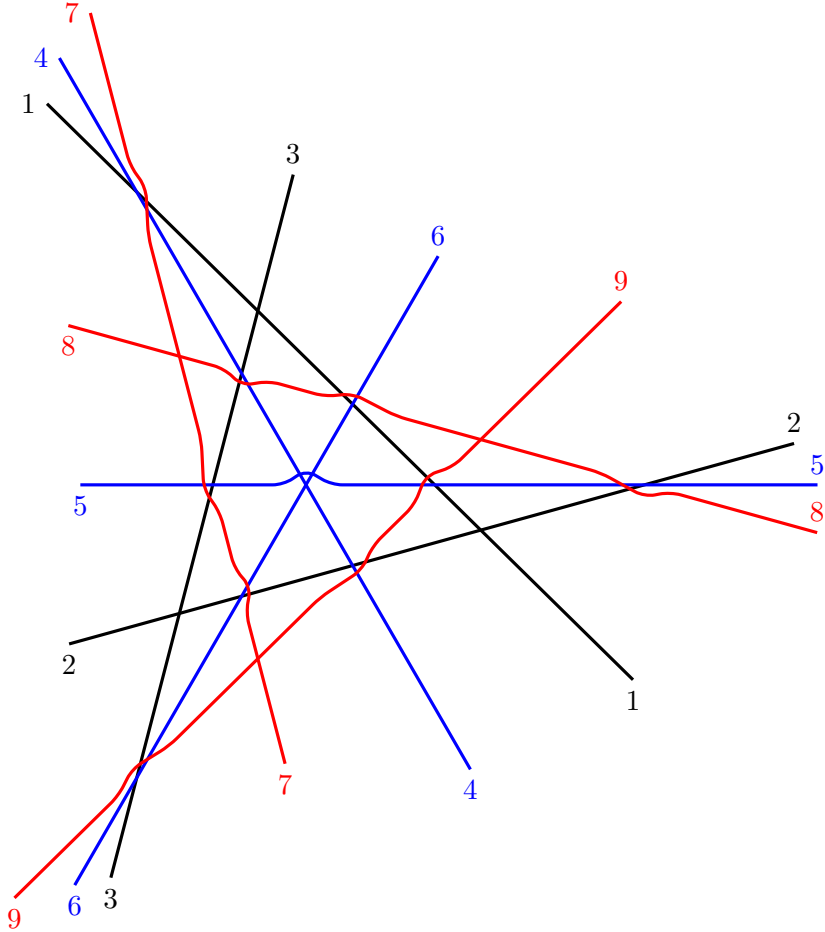

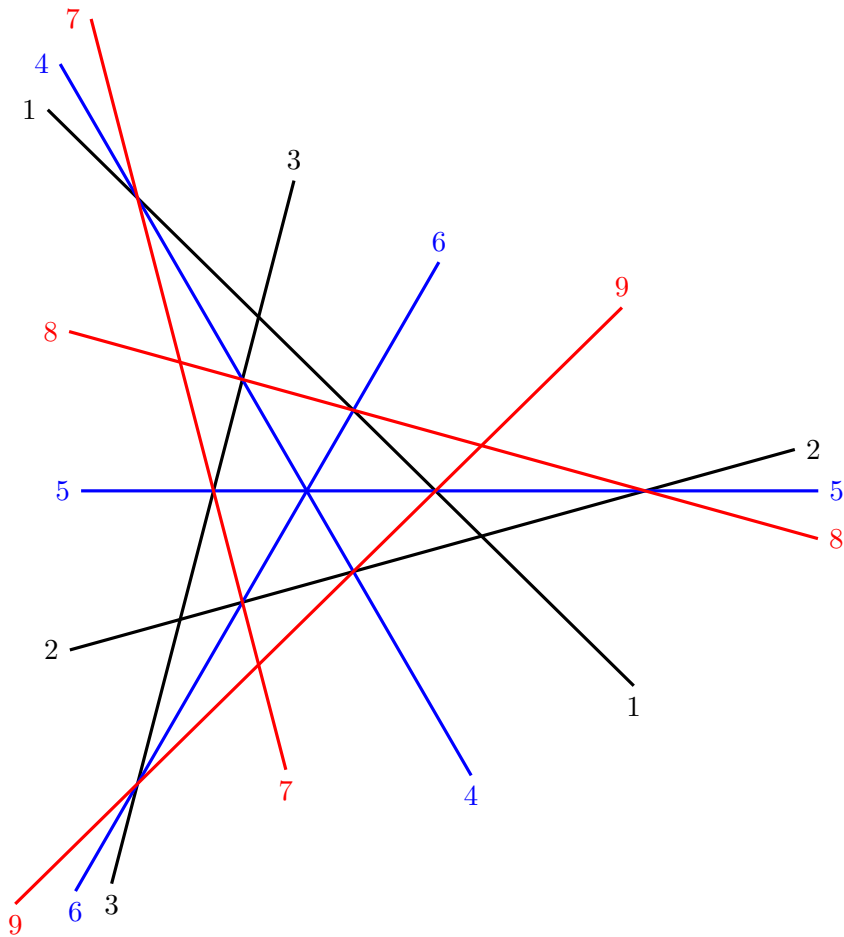
\begin{figure}

\centering

 \begin{tikzpicture}[scale=4]
  \draw [very thick,black](-0.860811,1.26763) coordinate(A1) node [left]{1}-- (1.08682,-0.648311) coordinate (B1) node [below]{1} ;
 \draw [very thick,black](1.62205,0.137639) coordinate(A2) node [right]{2}-- (-0.787374,-0.52919) coordinate (B2) node [left]{2} ;
 \draw [very thick,black](-0.0424588,1.03193) coordinate(A3) node [above]{3}-- (-0.648549,-1.30657) coordinate (B3) node [below]{3} ;
 \draw [very thick,blue](-0.819615,1.41962) coordinate(A4) node [left]{4}-- (0.54641,-0.94641) coordinate (B4) node [below]{4} ;
 \draw [very thick,blue](1.7,0.) coordinate(A5) node [right]{5}-- (-0.75,0.) coordinate (B5) node [left]{5} ;
 \draw [very thick,blue](0.43923,0.76077) coordinate(A6) node [above]{6}-- (-0.768653,-1.33135) coordinate (B6) node [below]{6} ;
 \draw [very thick,red](-0.716712,1.56957) coordinate(A7) node [left]{7}-- (-0.0694833,-0.927656) coordinate (B7) node [below]{7} ;
 \draw [very thick,red](1.69886,-0.158896) coordinate(A8) node [right]{8}-- (-0.789933,0.529898) coordinate (B8) node [left]{8} ;
 \draw [very thick,red](1.04752,0.609653) coordinate(A9) node [above]{9}-- (-0.96888,-1.37394) coordinate (B9) node [below]{9} ;	
 \end{tikzpicture}	
		
\caption{Straightening the arrangement of pseudo-lines for a non-realizable pseudo-GCO.}		
	
\label{arrangementpseudo2}	
		
\end{figure}

This phenomenon has a beautiful geometric interpretation. While it is impossible to realize \eqref{pseudoGCO39} as an arrangement of lines, as shown in \cref{arrangementpseudo}, which was first proposed in \cite{grunbaumarrangements} (see also \cite{celaya2020oriented}), it is possible to construct an arrangement of pseudo-lines, i.e., lines that bend, that has the intersection structure dictated by  \eqref{pseudoGCO39}. If one tries to ``straighten" the pseudo-lines to get an arrangement of lines, as shown in  \cref{arrangementpseudo2},  one finds that all pseudo-triangles with bent edges, including $\{4,5,6\}$, are forced to degenerate into a point and hence disappear!

\def \inter #1,#2 \inter { (intersection of  A#1--B#1 and A#2--B#2) }

\def \tri #1,#2,#3 \tri {  \inter #1,#2 \inter   -- \inter #2,#3 \inter  --  \inter #3,#1 \inter  --\inter #1,#2 \inter }

\begin{figure}
    \centering

 \begin{tikzpicture}[scale=2.4]
	 \draw [very thick,black](2.92196,5.40318) coordinate(A1) node [above]{1}-- (6.46941,0.805727) coordinate (B1) node [below]{1} ;
 \draw [very thick,black](0.687892,2.72869) coordinate(A2) node [below]{2}-- (5.93527,2.39488) coordinate (B2) node [below]{2} ;
 \draw [very thick,black](4.26116,5.41813) coordinate(A3) node [above]{3}-- (1.28691,1.71782) coordinate (B3) node [below]{3} ;
 \draw [very thick,blue](6.43791,1.19309) coordinate(A4) node [above]{4}-- (1.15056,3.98044) coordinate (B4) node [below]{4} ;
 \draw [very thick,blue](5.80718,3.09859) coordinate(A5) node [right]{5}-- (0.651911,2.9394) coordinate (B5) node [above]{5} ;
 \draw [very thick,blue](4.55269,4.85227) coordinate(A6) node [above]{6}-- (1.95718,0.927295) coordinate (B6) node [below]{6} ;
 \draw [very thick,red](6.68899,1.02507) coordinate(A7) node [below]{7}-- (0.89753,3.56653) coordinate (B7) node [left]{7} ;
 \draw [very thick,red](6.07725,3.22153) coordinate(A8) node [right]{8}-- (1.08425,3.21453) coordinate (B8) node [left]{8} ;
 \draw [very thick,red](5.90186,4.71263) coordinate(A9) node [above]{9}-- (2.73848,0.876012) coordinate (B9) node [below]{9} ;

\fill [fill=gray!50]  \tri 1,6,8 \tri ;

	 \draw [very thick,black](2.92196,5.40318) coordinate(A1) node [above]{1}-- (6.46941,0.805727) coordinate (B1) node [below]{1} ;
 \draw [very thick,black](0.687892,2.72869) coordinate(A2) node [below]{2}-- (5.93527,2.39488) coordinate (B2) node [below]{2} ;
 \draw [very thick,black](4.26116,5.41813) coordinate(A3) node [above]{3}-- (1.28691,1.71782) coordinate (B3) node [below]{3} ;
 \draw [very thick,blue](6.43791,1.19309) coordinate(A4) node [above]{4}-- (1.15056,3.98044) coordinate (B4) node [below]{4} ;
 \draw [very thick,blue](5.80718,3.09859) coordinate(A5) node [right]{5}-- (0.651911,2.9394) coordinate (B5) node [above]{5} ;
 \draw [very thick,blue](4.55269,4.85227) coordinate(A6) node [above]{6}-- (1.95718,0.927295) coordinate (B6) node [below]{6} ;
 \draw [very thick,red](6.68899,1.02507) coordinate(A7) node [below]{7}-- (0.89753,3.56653) coordinate (B7) node [left]{7} ;
 \draw [very thick,red](6.07725,3.22153) coordinate(A8) node [right]{8}-- (1.08425,3.21453) coordinate (B8) node [left]{8} ;
 \draw [very thick,red](5.90186,4.71263) coordinate(A9) node [above]{9}-- (2.73848,0.876012) coordinate (B9) node [below]{9} ;

\end{tikzpicture}	

\caption{
An arrangement of lines whose flip via triangle $\{1,6,8\}$ is geometrically forbidden.
}
 \label{neighborarrangement}
\end{figure}
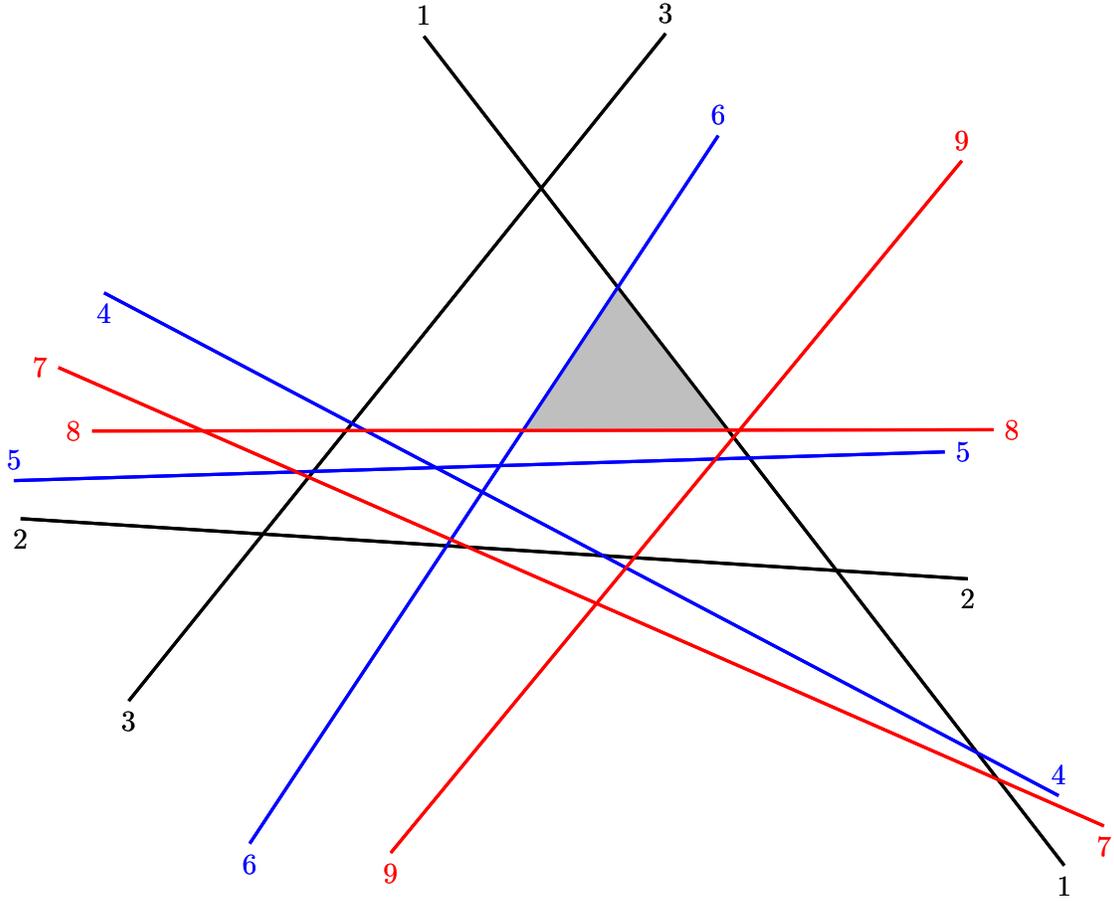

Due to the presence of the non-realizable pseudo-GCOs, some triangle flips of GCOs are geometrically forbidden even though their arrangements of lines indeed contain those triangles. See an example in \cref{neighborarrangement} whose GCO reads,
\begin{align}\label{pseudoGCO39nei}
\big(&
(24736895),(18536749),(14857296),(17385629),(18273469),(13927458),
\nonumber\\
&(14853629),(16437259),(15247638)
\big)\,.
\end{align}
Combinatorially, this GCO is connected to the non-realizable pseudo-GCO \eqref{pseudoGCO39} by a flip via triangle $\{1,6,8\}$ as their six $k=2$ color orderings are the same and their three remaining color orderings are related by flipping a pair of labels.  However, the triangle flip of this GCO via $\{1,6,8\}$ is geometrically forbidden because if we force triangle $\{1,6,8\}$ to shrink in \cref{neighborarrangement}, then the arrangement of lines will be forced to degenerate into the singular case in \cref{arrangementpseudo2} where all other triangles are also gone. We cannot blow up the triangle  $\{1,6,8\}$ again in the opposite direction unless we allow the straight lines to bend a little bit, which leads to the arrangements of pseudo-lines in \cref{arrangementpseudo}.

Now let us comment on what we can do in Algorithm I once a triangle flip of a target GCO leads to a pseudo-GCO. By making use of the other 11 triangle flips of the GCO \eqref{pseudoGCO39nei} and the residue relations there, the integrand for that GCO is fixed up to a single parameter, 
\begin{align}
\label{pseudoGCO39neiinte}
\frac{\Delta _{168} \left(\Delta _{139} \Delta _{456}-\Delta _{169} \Delta _{345}\right)+y_0  \Delta _{136} \Delta _{189}  \Delta _{456}}{\Delta _{136} \Delta _{147} \Delta _{159} \Delta _{168} \Delta _{189} \Delta _{249} \Delta _{258} \Delta _{267} \Delta _{348} \Delta _{357} \Delta _{369} \Delta _{456}} .
\end{align}
If we assume that the associated integrand for the non-realizable pseudo-GCO is zero and
require that the residue of \eqref{pseudoGCO39neiinte} at $ \Delta _{168} $ to vanish, i.e., that it has a spurious pole, $1/ \Delta _{168} $, we can fix the parameter, $y_0\to 0$.  

If instead, we could force the residue to match that of the non-realizable psedu-GCO. Recall we concluded that the $(3,9)$ pseudo-GCO is basic, i.e., $p=1$. So we can assign it  a ``naive integrand'',
$
1/(\Delta _{147} \Delta _{159} \Delta _{168} \Delta _{249} \Delta _{258} \Delta _{267} \Delta _{348} \Delta _{357} \Delta _{369})$. 
Then we find $y_0\to \pm1$ instead and \eqref{pseudoGCO39neiinte} would have the same residue at $ \Delta _{168} $ up to a sign with the naive integrand of the non-realizable pseudo-GCO.

\begin{figure}
    \centering

 \begin{tikzpicture}[scale=2.2]
  \draw [very thick](1.14745,2.71824) coordinate(A9) node [below]{9}-- (5.09675,4.26533) coordinate (B9) node [above]{9} ;
 \draw [very thick](0.952212,2.16088) coordinate(A8) node [below]{8}-- (4.63319,4.72623) coordinate (B8) node [above]{8} ;
 \draw [very thick](0.579394,3.43845) coordinate(A2) node [below]{2}-- (5.56247,6.89843) coordinate (B2) node [above]{2} ;
 \draw [very thick](1.79259,1.74797) coordinate(A7) node [below]{7}-- (5.5188,7.1901) coordinate (B7) node [above]{7} ;
 \draw [very thick](0.713652,3.11897) coordinate(A3) node [below]{3}-- (5.10205,7.17195) coordinate (B3) node [above]{3} ;
 \draw [very thick](1.59027,1.7718) coordinate(A6) node [below]{6}-- (4.22768,7.04896) coordinate (B6) node [above]{6} ;
 \draw [very thick](0.966711,1.81416) coordinate(A4) node [below]{4}-- (2.88379,6.7929) coordinate (B4) node [above]{4} ;
 \draw [very thick](3.44367,2.03684) coordinate(A5) node [below]{5}-- (1.633,6.15922) coordinate (B5) node [above]{5} ;

 \draw [very thick,dashed](5.44779,3.98417) coordinate(A1) node [right]{1}-- (0.236765,3.54077) coordinate (B1) node [left]{1} ;

  \draw [very thick,dashed](5.44779,4.22417) coordinate(A1) node [right]{1}-- (0.236765,3.78077) coordinate (B1) node [left]{1} ;
 
    \end {tikzpicture}

    \caption{
    Two $(3,9)$ arrangements of lines connected by a flip via the triangle $\{1,5,6\}$ but also sharing all the remaining triangles. The bottom dash line corresponds to GCO \eqref{bottomgco39} while the upper one corresponds to \eqref{uppergco39}.  The key feature is that the triangle $\{1,5,6\}$ in either arrangement of lines is surrounded by quadrangles or pentagons but no triangles. Hence when the bottom dash line $L_1$ moves up to the position of the upper dash line, all other triangles remain the same and all the standard $k=2$ color ordering in the lines stay the same except in  $L_1,L_5$ and $L_6$. 
    }
    \label{fig:twoGCO39}
\end{figure}
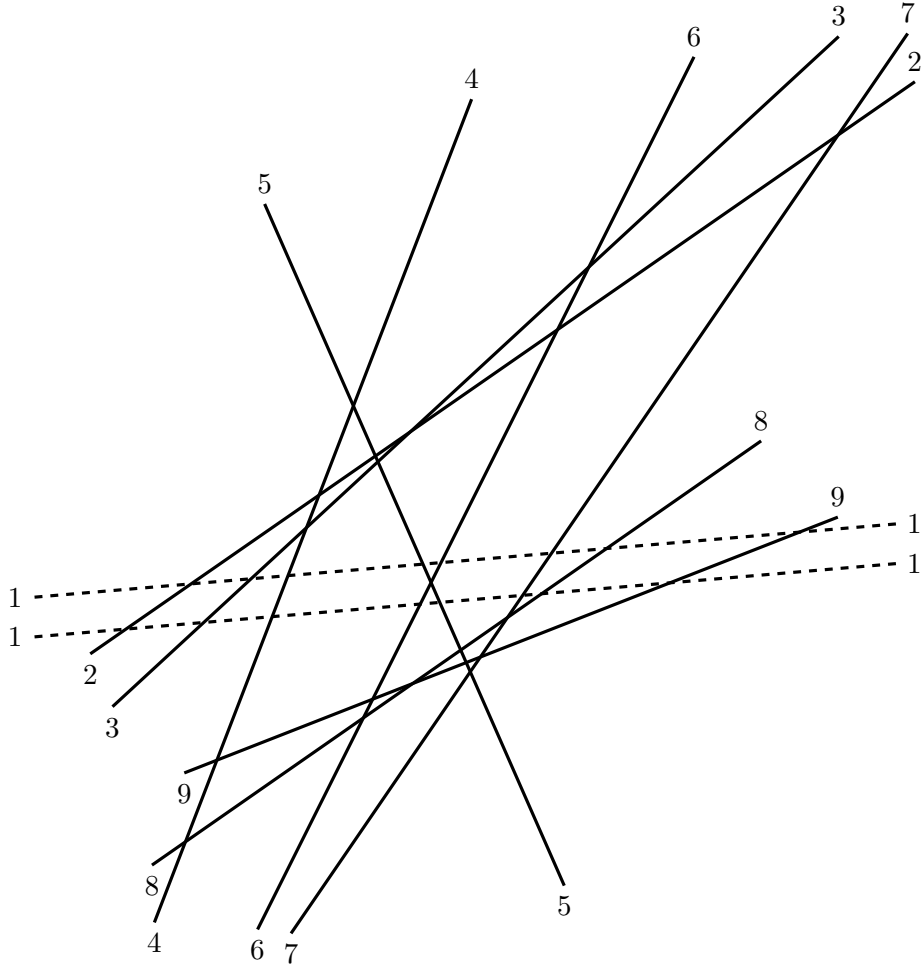

\subsection{Another New Phenomenon for $(3,9)$ GCOs\label{39newphe}}
In $(3,9)$, there are GCOs connected by a triangle flip and share the same set of triangles.
 For example, the following two GCOs 
\begin{align}
\label{uppergco39}
   \big( (23456789),(14536789),(14526789),(13256789),(13247986),(12347895),\qquad&
\nonumber   \\
   (12346598),(12346957),(12346857)\big)\,,&
   \\
   \label{bottomgco39}
      \big((23465789),(14536789),(14526789),(13256789),(16324798),(15234789),\qquad&
  \nonumber  \\
    (12346598),(12346957),(12346857)\big)\,,&
    \end{align}
are connected to each other by a flip of triangle $\{1,5,6\}$. Their arrangements of lines are shown in \cref{fig:twoGCO39}. By construction, they share the triangle  $\{1,5,6\}$
but they also share the remaining 10 triangles,
\begin{align}\nonumber
&  \{1,2,9\},\{1,3,4\},
  \{1,7,8\},\{2,3,5\},\{2,3,6\}, \\  & \{2,4,5\},\{4,6,7\},\{5,7,9\},\{5,8,9\},\{6,8,9\}  \,.
\end{align}
Their integrands have the same denominators but different numerators,
\begin{align}
    \frac{\Delta _{159} \Delta _{256}}{\Delta _{129} \Delta _{134} \Delta _{156} \Delta _{178} \Delta _{235} \Delta _{236} \Delta _{245} \Delta _{467} \Delta _{579} \Delta _{589} \Delta _{689}}\,,
    \\
    -\frac{\Delta _{125} \Delta _{569}}{\Delta _{129} \Delta _{134} \Delta _{156} \Delta _{178} \Delta _{235} \Delta _{236} \Delta _{245} \Delta _{467} \Delta _{579} \Delta _{589} \Delta _{689}}
  \,.
\end{align}
The two GCOs are related by $2\leftrightarrow9,3\leftrightarrow8,4\leftrightarrow7$ as are their integrands.

Such GCOs are not necessarily of the same type. For example, the following two GCOs are also connected by flipping triangle $\{1,5,6\}$,
\begin{align}
\label{uppergco3922}
   \big( (23465789),(14635789),(14625789),(13265789),(12347896),(13249785),
   \qquad&
  \nonumber  \\
   (12345968),(12345967),
    (12345876) \big)\,,
    \\
    \label{bottomgco3922}
    \big(
    (23456789),(14635789),(14625789),(13265789),(16234789),(15324978),
    \qquad&
  \nonumber  \\(12345968),(12345967),
    (12345876)
    \big)\,,
\end{align}
and share the same set of triangles,
\begin{align}\nonumber &
 \{1,5,6\},\,\{1,2,9\},\{1,3,4\},
 \{1,7,8\},\{2,3,5\},\{2,3,6\},\\ &\{2,4,6\},\{4,5,7\},\{5,8,9\},\{6,7,8\},\{6,7,9\}\,,
\end{align}
but they are not related by relabelling. Here are their integrands, 
\begin{align}
\frac{
\Delta _{167} \Delta _{256}
}{\Delta _{129} \Delta _{134} \Delta _{156} \Delta _{178} \Delta _{235} \Delta _{236} \Delta _{246} \Delta _{457} \Delta _{589} \Delta _{678} \Delta _{679}}\,,
\\
-\frac{\Delta _{126} \Delta _{567}}{\Delta _{129} \Delta _{134} \Delta _{156} \Delta _{178} \Delta _{235} \Delta _{236} \Delta _{246} \Delta _{457} \Delta _{589} \Delta _{678} \Delta _{679}}
\,,
\end{align}
where their numerators are related by $2\leftrightarrow7$ but their denominators are not.

\section{Proof of $k=3$ Irreducible Decoupling Identities \label{appproof}}
Here we prove Proposition \ref{irreProof} in the main text by making use of the residue theorem.

\begin{proof}
Assume that $L = \{\Sigma_1,\Sigma_2,\ldots ,\Sigma_N\}$ is a set of GCOs produced by using algorithm II. Compute the corresponding CEGM integrands $\{ {\cal I}(\Sigma_1),{\cal I}(\Sigma_2), \ldots ,{\cal I}(\Sigma_N)\}$. An identity is found if there exists an $N$-dimensional vector $v\in \{-1,1\}^N$ such that 
\be\label{goal} 
\left( {\cal I}(\Sigma_1),{\cal I}(\Sigma_2), \ldots ,{\cal I}(\Sigma_N)\right)\cdot v = 0.
\ee 

Treat each integrand as a rational function of a single complex variable $z$ by using a parameterization of $X(3,n)$ of the form,
\be\label{cegm3para} 
	M := \left[ \begin{array}{cccccccc}
		1 & 0 & 0 & 1 & y_{15}  & y_{16} & \cdots & y_{1n}+z  \\
		0 & 1 & 0 & 1 &  y_{25}  & y_{26} & \cdots & y_{2n} \\
		0 & 0 & 1 & 1 &    1  & 1   & \cdots & 1 
\end{array}  \right] .
\ee 
Here all $y$ variables are considered to be generic and fixed. Every minor $\Delta_{abn}$ becomes a linear function of $z$ and without loss of generality, we assume that no two $\Delta_{abn}(z)$ share a root. 

This produces a rational function for every choice of $v$,
\be 
f_v(z):= \left( {\cal I}(\Sigma_1,z),{\cal I}(\Sigma_2,z), \ldots ,{\cal I}(\Sigma_N,z) \right)\cdot v .
\ee 

It is clear that as $z\to \infty$, $f_v(z) = {\cal O}(z^{-3})$. The reason is the torus scaling \eqref{torusS}. We also know that all possible poles of $f_v(z)$ are simple and of the form $\Delta_{abn}(z)$. We now have to prove that there exists a choice of $v$ such that $f_v(z)$ has a zero residue at $z^*$, and such that $\Delta_{abn}(z^*)=0$ for any $\{a,b\}$.  

Consider the poles in ${\cal I}(\Sigma_1,z)$. If $\{a,b\} \neq \{i,j\}$ and $\{a,b\} \neq \{p,q\}$ then, by construction (i.e. steps 2 and 3 in Algorithm II), there exist another GCO in $L$ connected to $\Sigma_1$ by a triangle flip along $\{a,b,n\}$, say it is $\Sigma_2$. This means that 
\be\label{oneSigma} 
{\rm Res}_{\Delta_{abn}(z^*)=0} {\cal I}(\Sigma_1,z) = \pm {\rm Res}_{\Delta_{abn}(z^*)=0} {\cal I}(\Sigma_2,z).
\ee 
Denote the entries in $v=(v_1,v_2,\ldots ,v_N)$. Clearly, one can set $v_1=1$ since if $v$ satisfies \eqref{goal} so does $-v$. Using \eqref{oneSigma}, one can fix $v_2$ so that ${\cal I}(\Sigma_1,z)+v_2 {\cal I}(\Sigma_2,z)$ has zero residue at $z^*$.

The procedure can be continued for all poles corresponding to the triangle flips performed during the implementation of the algorithm. 

Finally, we are left to prove that the choice of $v$ found above ensures that $f_v(z)$ has a zero residue at the solution to $\Delta_{ijn}(z)=0$ and to $\Delta_{pqn}(z)=0$.

In order to proceed it is necessary to introduce another degree of freedom by modifying our parametrization to include a second coordinate, $w$,
\be\label{cegm3para2D} 
	M := \left[ \begin{array}{cccccccc}
		1 & 0 & 0 & 1 & y_{15}  & y_{16} & \cdots & y_{1n}+z  \\
		0 & 1 & 0 & 1 &  y_{25}  & y_{26} & \cdots & y_{2n}+w \\
		0 & 0 & 1 & 1 &    1  & 1   & \cdots & 1 
\end{array}  \right] .
\ee 
Now, localizing on $\Delta_{ijn}(z,w)=0$ by solving for $z=z(w)$ and treating each column in $M$ \eqref{cegm3para2D} as a point in $\mathbb{CP}^{2}$ implies that the points labeled $i,j,n$ are on a projective line, i.e., a $\mathbb{CP}^{1}$. We can take $w$ as the inhomogenous coordinate of the line. 

Our function $f_v(z(w),w) = g_v(w)$ is now a rational function of the $\mathbb{CP}^{1}$ defined by points $i$ and $j$. Those elements in  $\{ {\cal I}(\Sigma_1),{\cal I}(\Sigma_2), \ldots ,{\cal I}(\Sigma_N)\}$ that contain poles on the $w$, $\mathbb{CP}^{1}$, different from $\Delta_{pqn}(z(w),w)=0$, by construction always have a companion with the same residue but opposite in sign. This is because such poles correspond to triangle flips used in the algorithm. Therefore the only possible singularities of the function $g_v(w)$ is a simple pole at $\Delta_{pqn}(z(w),w)=0$. However, $g_v(w)={\cal O}(w^{-2})$ as $w\to \infty$ and since $\Delta_{pqn}(z(w),w)$ is linear in $w$ it must be that $g_v(w)=0$. Hence the residue of $f_v(z,w)$ at $\Delta_{ijn}(z,w)=0$ vanishes.

Repeating the same argument but localizing to the line $\Delta_{pqn}(z,w)=0$ we conclude that the function $f_v(z)$ is identically zero and hence we have a decoupling identity.

The last statement to prove is that the identity is irreducible. But this follows from the fact that on the $\mathbb{CP}^{1}$, the functions cancel their poles pairwise in analogy to a $k=2$ decoupling identity, and therefore there are no proper subsets that can be made to vanish.

\end{proof}

\section{Tables of GCOs and Identities}

In this appendix, we present useful information needed to follow some of the arguments in the main text.

\subsection{$(3,7)$ GCOs \label{37gco}}

Here we list the 11 representatives of (3,7) GCOs of different types given in Table 2 of \cite{Cachazo:2022pnx} for completeness, 
\allowdisplaybreaks[1]
\begin{align}
 &\Sigma_0=      
((234567),(134567),(124567),(123567),(123467),(123457),(123456)) \,,
\\ 
 &\Sigma_I= 
((234567),(134567),(124567),(123567),(123476),(123475),(123465))\,,
\\ 
& \Sigma_{II}= 
((234567),(134567),(124567),(123576),(123476),(123745),(123645))\,,
\\ 
& \Sigma_{III}= 
((234567),(134567),(124567),(123756),(123746),(123745),(123654))\,,
\\ 
& \Sigma_{IV}= 
((234567),(134567),(124576),(123756),(123746),(127345),(126354))\,,
\\ 
& \Sigma_{V}= 
((234567),(134576),(124576),(123756),(123746),(154327),(145326))\,,
\\ 
& \Sigma_{VI}= 
((234756),(134576),(124567),(123567),(164327),(127345),(143625))\,,
\\ 
& \Sigma_{VII}= 
((234756),(134756),(124567),(123567),(127346),(127345),(125634))\,,
\\ 
& \Sigma_{VIII}= 
((234567),(134576),(124576),(123765),(123764),(145327),(145326))\,,
\\ 
& \Sigma_{IX}= 
((234567),(134576),(124756),(123765),(127364),(145327),(143526))\,,
\\ 
& \Sigma_{X}= 
((234576),(134576),(124756),(123675),(127364),(123574),(125364))   \,.
\end{align}
\allowdisplaybreaks[0]

\subsection{ 31 Explicit GCOs and Integrands in Figure \ref{P36point1Text}  \label{appb}}

Here are the 31 explicit forms of GCOs, their types, and the integrands in \cref{P36point1Text}. 

\allowdisplaybreaks[1]
\small{
\begin{align}
&\text{C}_1\,\,\text{I}\,\,\text{((23456), (13456), (12456), (12365), (12364), (12354))}\,\,-\frac{1}{\Delta _{123} \Delta _{126} \Delta _{145} \Delta _{234} \Delta _{356} \Delta _{456}}\nonumber \\ &\text{C}_2,0\,\,\text{((23456), (13456), (12456), (12356), (12346), (12345))}\,\,\frac{1}{\Delta _{123} \Delta _{126} \Delta _{156} \Delta _{234} \Delta _{345} \Delta _{456}}\nonumber \\ &\text{C}_3\,\,\text{I}\,\,\text{((23465), (13456), (12456), (12356), (14326), (14325))}\,\,\frac{1}{\Delta _{123} \Delta _{146} \Delta _{156} \Delta _{234} \Delta _{256} \Delta _{345}}\nonumber \\ &\text{C}_4\,\,\text{II}\,\,\text{((23645), (13456), (12456), (15326), (14326), (13254))}\,\,\frac{\Delta _{134}}{\Delta _{123} \Delta _{136} \Delta _{145} \Delta _{146} \Delta _{234} \Delta _{256} \Delta _{345}}\nonumber \\ &\text{C}_5\,\,\text{II}\,\,\text{((23465), (13465), (12456), (12356), (12634), (12534))}\,\,\frac{\Delta _{235}}{\Delta _{123} \Delta _{125} \Delta _{146} \Delta _{234} \Delta _{256} \Delta _{345} \Delta _{356}}\nonumber \\ &\text{C}_6\,\,\text{I}\,\,\text{((25436), (13456), (15426), (15326), (14326), (12543))}\,\,\frac{1}{\Delta _{126} \Delta _{136} \Delta _{145} \Delta _{234} \Delta _{256} \Delta _{345}}\nonumber \\ &\text{C}_7\,\,\text{I}\,\,\text{((23465), (13465), (12465), (12356), (12364), (12354))}\,\,-\frac{1}{\Delta _{123} \Delta _{125} \Delta _{146} \Delta _{234} \Delta _{356} \Delta _{456}}\nonumber \\ &\text{C}_8\,\,\text{III}\,\,\text{((23645), (13465), (12456), (15326), (12634), (13524))}\,\,\frac{\Delta _{123} \Delta _{156} \Delta _{246} \Delta _{345}-\Delta _{126} \Delta _{135} \Delta _{234} \Delta _{456}}{\Delta _{123} \Delta _{125} \Delta _{136} \Delta _{145} \Delta _{146} \Delta _{234} \Delta _{246} \Delta _{256} \Delta _{345} \Delta _{356}}\nonumber \\ &\text{C}_9\,\,\text{II}\,\,\text{((23645), (13465), (12465), (15326), (12364), (14235))}\,\,-\frac{\Delta _{124}}{\Delta _{123} \Delta _{125} \Delta _{145} \Delta _{146} \Delta _{234} \Delta _{246} \Delta _{356}}\nonumber \\ &\text{C}_{10}\,\,\text{II}\,\,\text{((25436), (13465), (15426), (15326), (12634), (13425))}\,\,-\frac{\Delta _{245}}{\Delta _{125} \Delta _{136} \Delta _{145} \Delta _{234} \Delta _{246} \Delta _{256} \Delta _{345}}\nonumber \\ &\text{C}_{11}\,\,\text{I}\,\,\text{((25436), (13645), (15426), (12635), (12634), (13245))}\,\,\frac{1}{\Delta _{125} \Delta _{136} \Delta _{145} \Delta _{236} \Delta _{246} \Delta _{345}}\nonumber \\ &\text{C}_{12}\,\,\text{II}\,\,\text{((23645), (13645), (12456), (12635), (12634), (12453))}\,\,\frac{\Delta _{135}}{\Delta _{123} \Delta _{125} \Delta _{136} \Delta _{145} \Delta _{246} \Delta _{345} \Delta _{356}}\nonumber \\ &\text{C}_{13}\,\,\text{I}\,\,\text{((23645), (13645), (12465), (12635), (12364), (12435))}\,\,\frac{1}{\Delta _{123} \Delta _{125} \Delta _{145} \Delta _{246} \Delta _{346} \Delta _{356}}\nonumber \\ &\text{C}_{14},0\,\,\text{((23645), (13645), (12645), (12365), (12364), (12345))}\,\,-\frac{1}{\Delta _{123} \Delta _{125} \Delta _{145} \Delta _{236} \Delta _{346} \Delta _{456}}\nonumber \\ &\text{C}_{15},0\,\,\text{((25436), (15436), (12645), (12635), (12634), (12345))}\,\,-\frac{1}{\Delta _{125} \Delta _{126} \Delta _{145} \Delta _{236} \Delta _{345} \Delta _{346}}\nonumber \\ &\text{C}_{16}\,\,\text{I}\,\,\text{((25436), (15436), (12465), (12365), (12634), (12435))}\,\,-\frac{1}{\Delta _{125} \Delta _{126} \Delta _{145} \Delta _{234} \Delta _{346} \Delta _{356}}\nonumber \\ &\text{C}_{17},0\,\,\text{((23465), (13465), (12465), (12365), (12346), (12345))}\,\,\frac{1}{\Delta _{123} \Delta _{125} \Delta _{156} \Delta _{234} \Delta _{346} \Delta _{456}}\nonumber \\ &\text{C}_{18},0\,\,\text{((23456), (15436), (15426), (15326), (14326), (12345))}\,\,\frac{1}{\Delta _{126} \Delta _{145} \Delta _{156} \Delta _{234} \Delta _{236} \Delta _{345}}\nonumber \\ &\text{C}_{19}\,\,\text{I}\,\,\text{((25436), (15436), (12456), (12356), (12346), (12543))}\,\,\frac{1}{\Delta _{125} \Delta _{126} \Delta _{136} \Delta _{234} \Delta _{345} \Delta _{456}}\nonumber \\ &\text{C}_{20}\,\,\text{I}\,\,\text{((23456), (13465), (12465), (12365), (14326), (14325))}\,\,\frac{1}{\Delta _{123} \Delta _{145} \Delta _{156} \Delta _{234} \Delta _{256} \Delta _{346}}\nonumber \\ &\text{C}_{21}\,\,\text{I}\,\,\text{((23456), (13645), (12645), (15326), (14326), (13245))}\,\,\frac{1}{\Delta _{123} \Delta _{145} \Delta _{156} \Delta _{236} \Delta _{246} \Delta _{345}}\nonumber \\ &\text{C}_{22}\,\,\text{I}\,\,\text{((23456), (13456), (12645), (12635), (12634), (12543))}\,\,-\frac{1}{\Delta _{123} \Delta _{126} \Delta _{145} \Delta _{256} \Delta _{345} \Delta _{346}}\nonumber \\ &\text{C}_{23}\,\,\text{II}\,\,\text{((25436), (15436), (12456), (12365), (12364), (12453))}\,\,\frac{\Delta _{156}}{\Delta _{125} \Delta _{126} \Delta _{136} \Delta _{145} \Delta _{234} \Delta _{356} \Delta _{456}}\nonumber \\ &\text{C}_{24}\,\,\text{II}\,\,\text{((23456), (13456), (12465), (12365), (12634), (12534))}\,\,\frac{\Delta _{236}}{\Delta _{123} \Delta _{126} \Delta _{145} \Delta _{234} \Delta _{256} \Delta _{346} \Delta _{356}}\nonumber \\ &\text{C}_{25}\,\,\text{II}\,\,\text{((23645), (15436), (15426), (12356), (12346), (13254))}\,\,\frac{\Delta _{346}}{\Delta _{125} \Delta _{136} \Delta _{146} \Delta _{234} \Delta _{236} \Delta _{345} \Delta _{456}}\nonumber \\ &\text{C}_{26}\,\,\text{I}\,\,\text{((23465), (15436), (15426), (15326), (12346), (14325))}\,\,-\frac{1}{\Delta _{125} \Delta _{146} \Delta _{156} \Delta _{234} \Delta _{236} \Delta _{345}}\nonumber \\ &\text{C}_{27}\,\,\text{I}\,\,\text{((23645), (13645), (12645), (12356), (12346), (12354))}\,\,\frac{1}{\Delta _{123} \Delta _{125} \Delta _{146} \Delta _{236} \Delta _{345} \Delta _{456}}\nonumber \\ &\text{C}_{28}\,\,\text{I}\,\,\text{((23645), (15436), (15426), (12365), (12364), (13245))}\,\,-\frac{1}{\Delta _{125} \Delta _{136} \Delta _{145} \Delta _{234} \Delta _{236} \Delta _{456}}\nonumber \\ &\text{C}_{29}\,\,\text{II}\,\,\text{((23456), (13465), (12645), (12635), (14326), (13425))}\,\,-\frac{\Delta _{456}}{\Delta _{123} \Delta _{145} \Delta _{156} \Delta _{246} \Delta _{256} \Delta _{345} \Delta _{346}}\nonumber \\ &\text{C}_{30}\,\,\text{I}\,\,\text{((23465), (13465), (12645), (12635), (12346), (12435))}\,\,\frac{1}{\Delta _{123} \Delta _{125} \Delta _{156} \Delta _{246} \Delta _{345} \Delta _{346}}\nonumber \\ &\text{C}_{31}\,\,\text{II}\,\,\text{((23465), (13645), (12645), (15326), (12346), (14235))}\,\,-\frac{\Delta _{126}}{\Delta _{123} \Delta _{125} \Delta _{146} \Delta _{156} \Delta _{236} \Delta _{246} \Delta _{345}}
\nonumber
\end{align}
}
\allowdisplaybreaks[0]

\def \wi {.8}

\setlength{\abovecaptionskip}{-5pt}

\begin{figure}[h!]
	\centering
   \includegraphics
   [width=\wi \linewidth]
{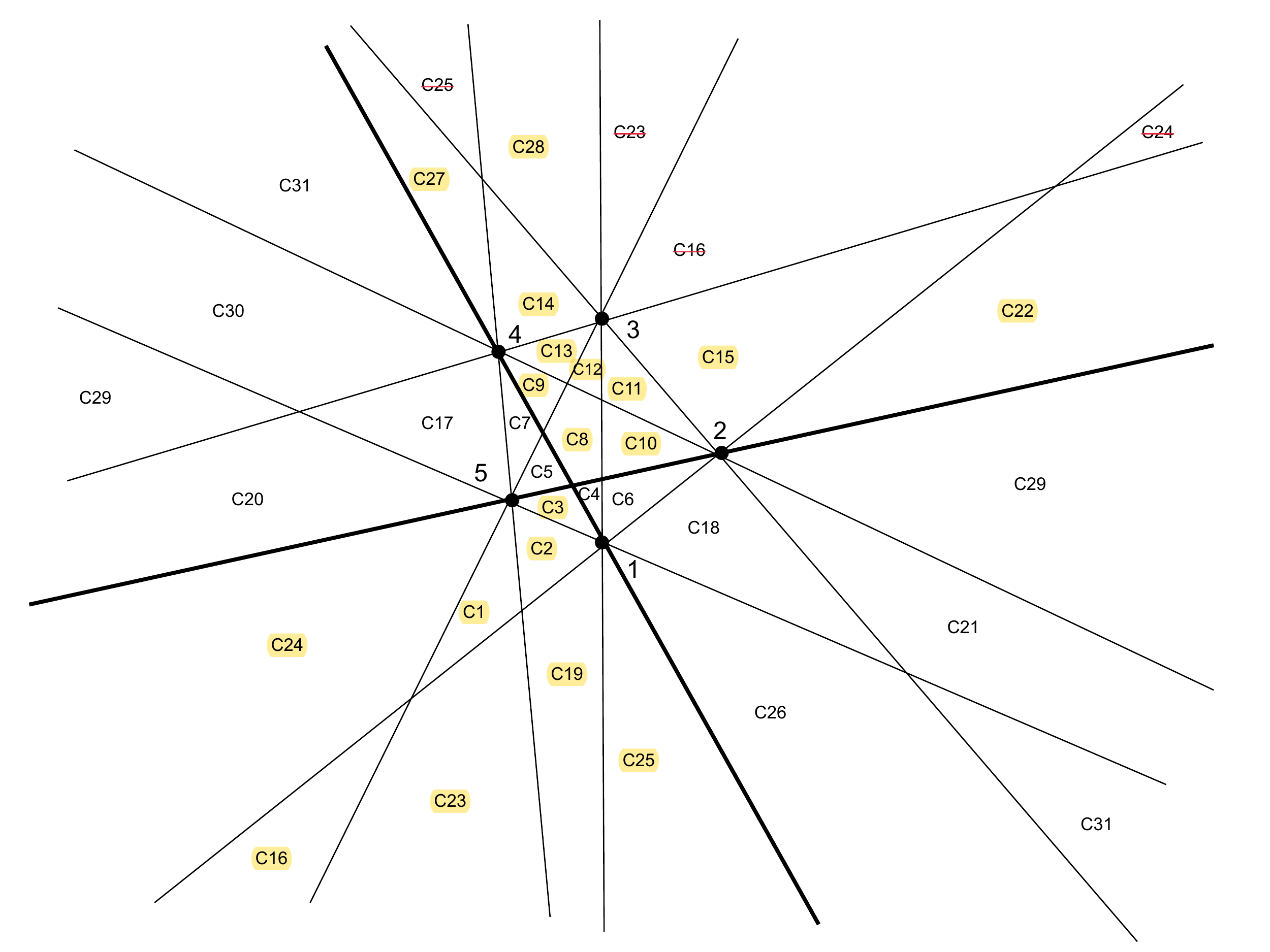}
\caption{
Partitioning identity obtained by considering lines $L_{14}$ and $L_{25}$. The identity that contains chamber $C_8$ has $12$ terms while the one that contains $C_{17}$ has $19$ terms.   \label{P36Irred1219A} }
\end{figure}

\begin{figure}[h!]
	\centering
   \includegraphics
   [width=\wi \linewidth]
{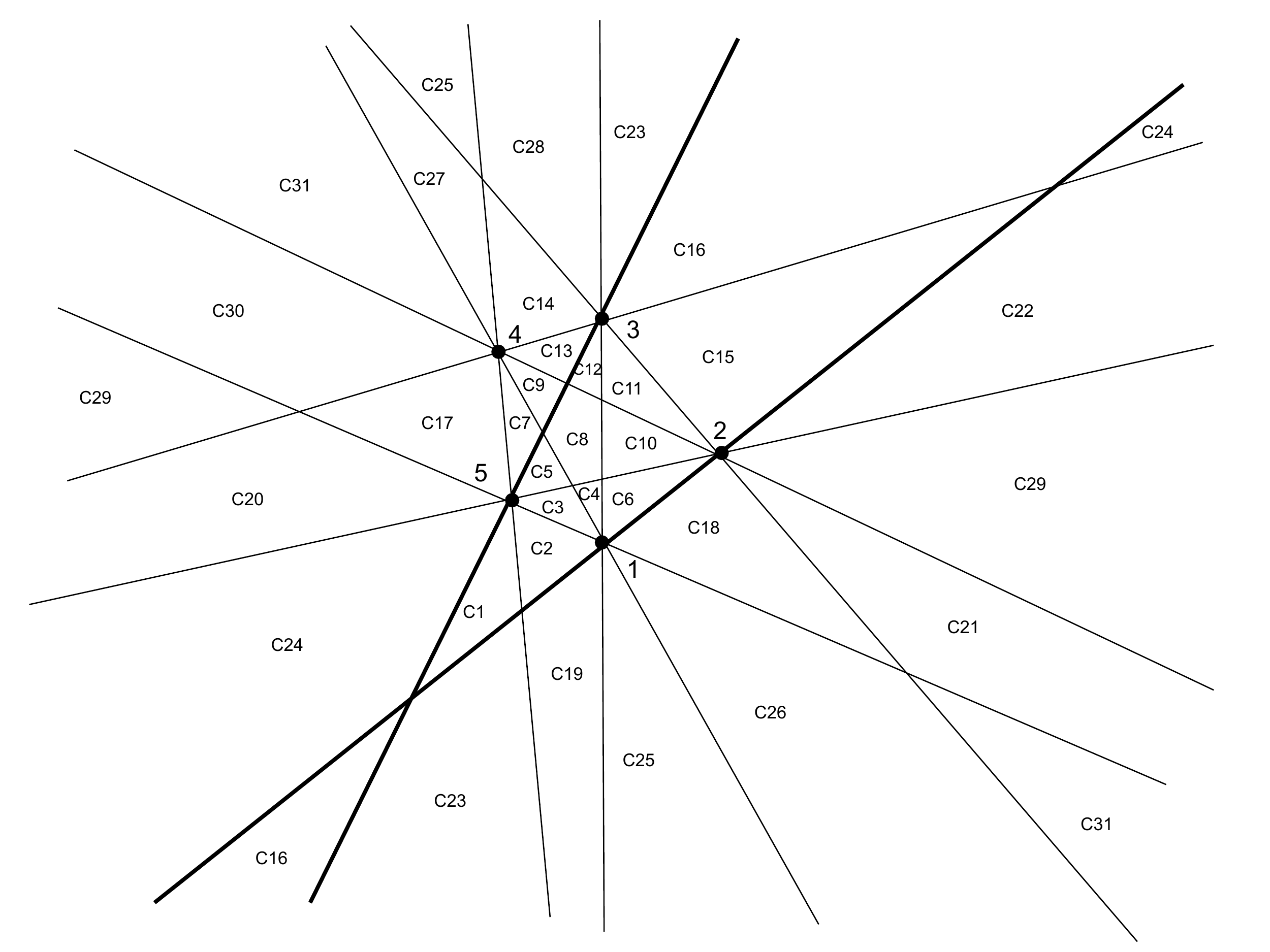}
\caption{
Partitioning identity obtained by considering lines $L_{12}$ and $L_{35}$. The identity that contains chamber $C_8$ has $12$ terms while the one that contains $C_{17}$ has $19$ terms.   \label{P36Irred1219B} }
\end{figure}

\begin{figure}[h!]
	\centering
   \includegraphics
   [width=\wi \linewidth]
{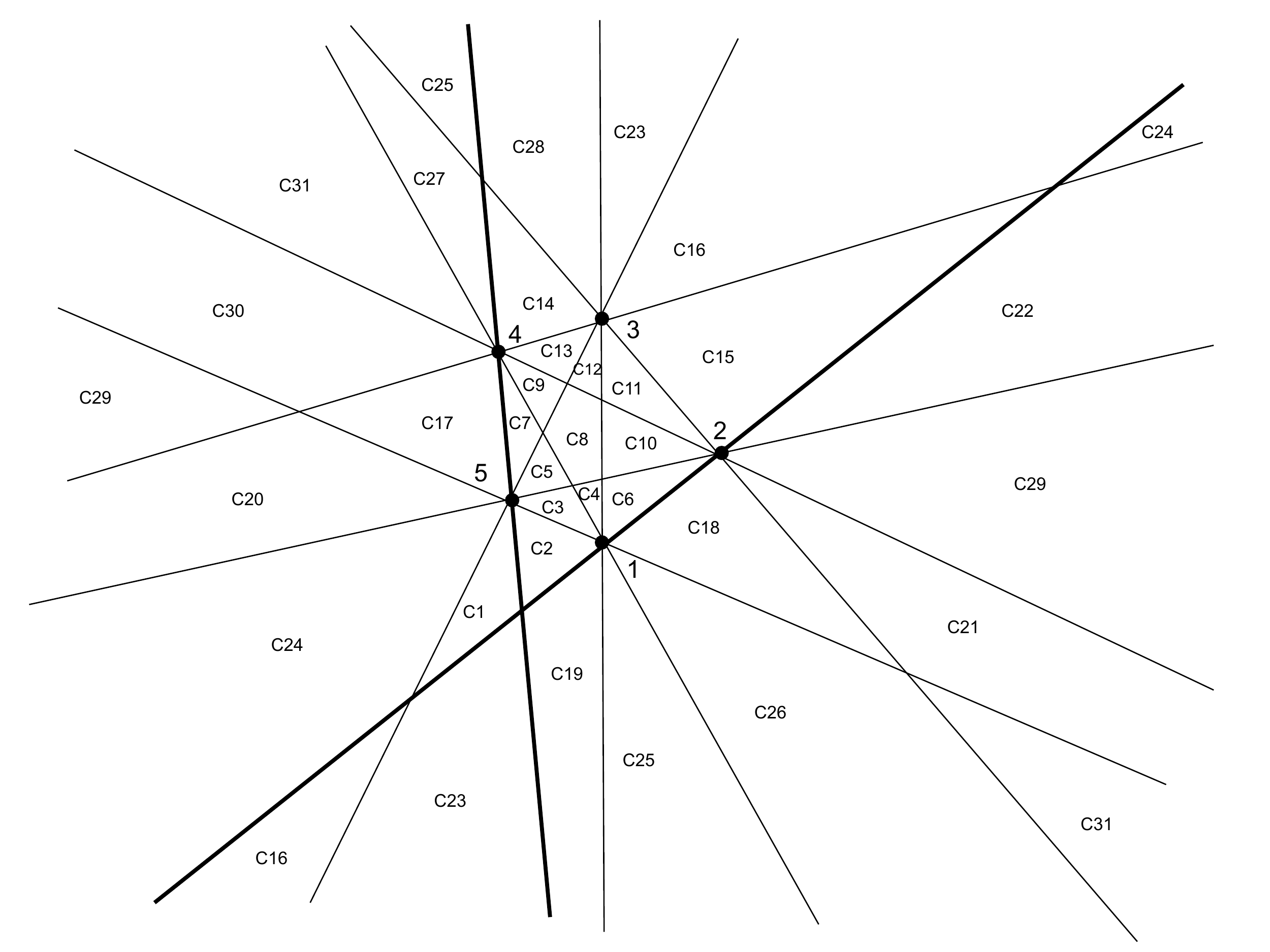}
\caption{
Partitioning identity obtained by considering lines $L_{12}$ and $L_{45}$. The identity that contains chamber $C_8$ has $14$ terms while the one that contains $C_{17}$ has $17$ terms.   \label{P36Irred1417}}
\end{figure}

\subsection{All $(3,6)$ Irreducible Identities \label{appc}}

In \cref{sec4}, we introduced irreducible identities and a combinatorial way to find them. In the case $(3,6)$, one can explicitly visualize the identities by selecting two of the ten lines in $\mathbb{RP}^2$ constructed using the five points left after removing point $6$. Each pair of lines intersects at a point and separates $\mathbb{RP}^2$ into two regions. 

When the intersection point is not one of the five special points, then the identities obtained by combining the GCOs on either one of the two regions give rise to an irreducible identity, i.e., these are partition identities. There are three different such configurations in $(3,6)$. Two of them give rise to a pair of a $12$-term and a $19$-term identities. The other one leads to a pair of identities with $14$ and $17$ terms. These are shown in figure \ref{P36Irred1219A}, \ref{P36Irred1219B}, and \ref{P36Irred1417}.

When the intersection point is one of the five special points, then the identities obtained by combining the GCOs on either one of the two regions are not guaranteed to be irreducible. This is the case discussed in the main text but we repeat the description here for completeness. 

Any special point belongs to four lines. The four lines split $\mathbb{RP}^2$ into four regions, and each region gives rise to an identity. This is shown in 
\cref{P36point1Text}
where one can easily see that the identities contain $7$, $7$, $8$ and $9$ terms. Clearly, $7+7+8+9=31$ since they partition the projective plane.

\bibliographystyle{jhep}
\bibliography{references}

\end{document}